\documentclass{article}

\usepackage[margin=1in]{geometry}
\usepackage[onehalfspacing]{setspace}
\usepackage{amsmath,mathtools}
\usepackage{amssymb}
\usepackage{pgfplots}
\usepackage{subcaption}
\usepackage{booktabs}
\usepackage[braket, qm]{qcircuit}
\usepackage{graphicx}
\usepackage{authblk}
\usepackage{xcolor}
\usepackage{booktabs}
\usepackage{subcaption}
\usepackage{comment}
\usepackage[linesnumbered,ruled,vlined]{algorithm2e}
\usepackage{qcircuit}
\usepackage{pgffor}
\usepackage{ifthen}
\usepackage{tcolorbox}

\SetCommentSty{mycommfont}

\SetKwInput{KwInput}{Input}                %
\SetKwInput{KwOutput}{Output}              %
\SetKwProg{Fn}{Function}{}{}

\newcommand{\baseW}{\Wfun} %

\usepackage{siunitx} %

\usepackage{xinttools}
\usepackage{xintexpr}
\usepackage{multicol}
\usepackage{centeredline}
\usepackage{nicematrix}
\usepackage{arydshln}

\ExplSyntaxOn
\NewDocumentCommand{\mymat}{mmmm}
{%
	\tl_clear:N \l_tmpa_tl
	\int_step_inline:nnn  {0}{ #2 -1 }
	{%
		\int_step_inline:nnn {0}{ #3 -1}
		{%
			\tl_put_right:Nn \l_tmpa_tl {$#1{\csname #4 \endcsname {##1}{####1} {#2} {#3}}$}
			\int_compare:nNnTF {####1} < {#3-1} 
			{\tl_put_right:Nn \l_tmpa_tl {&}} {}
		}
		\tl_put_right:Nn \l_tmpa_tl {\\}
	}
	\tl_use:N \l_tmpa_tl
}

\cs_new:Npn \myAij #1#2#3#4{
	\pgfmathtruncatemacro{\val}{#1 * #2 }
	\val
}

\cs_new:Npn \myHadamard #1#2#3#4{
	\bitsetSetDec{A}{#1}
	\bitsetSetDec{B}{#2}
	\bitsetAnd{A}{B} %
	\pgfmathtruncatemacro{\val}{mod(\bitsetCardinality{A}, 2)}
	\ifnum\val=0
	\pgfmathtruncatemacro{\val} {1}		
	\else 
	\pgfmathtruncatemacro{\val} {-1}	
	\fi
	\val
}

\ExplSyntaxOff

\xintdefiifunc fA(x,y) := \xintexpr {(x>y) ? {1}{0}}\relax;%

\newcommand{\genMatrix}[3]{%
	\xintdefiivar Matrix = ndmap(#3,1..#1; 1..#2); 
	\[
	\def\xintexpralignbegin {\begin{bmatrix*}[r]}%
		\def\xintexpralignend {\end{bmatrix*}}%
	\def\xintexpralignlinesep {\noexpand\\}%
	\def\xintexpraligninnersep {&}%
	\let\xintexpralignleftbracket\empty \let\xintexpralignleftsep\empty
	\let\xintexpralignrightbracket\empty \let\xintexpralignrightsep\empty
	\xintthealign	\xintiiexpr Matrix\relax
	\]
}

\usepackage[english]{babel}

\usepackage{pdflscape}
\usepackage{fancyvrb}
\usepackage{calc}
\usepackage{rotating}
\usepackage{pgf,tikz}
\usepackage{mathrsfs}
\usetikzlibrary{arrows}
\usepackage{mathtools}

\usepackage{listings}
\usepackage[]{qcircuit}
\usepackage{braket}
\usepackage{mathtools}
\usepackage{varwidth}
\usepackage{caption}
\usepackage{subcaption}
\usepackage{graphicx}
\usepackage[export]{adjustbox}

\usepackage{float}

\makeatletter
\def\paragraph{\@startsection{paragraph}{4}%
	\z@\z@{-\fontdimen2\font}%
	{\normalfont\bfseries}}
\makeatother

\usepackage{graphicx}
\usepackage{pgffor}
\usepackage{tikz,tikz-cd}
\usetikzlibrary{backgrounds,fit,decorations.pathreplacing}  %
\usepackage{booktabs}
\usepackage{enumerate}

\usepackage{amssymb, amsmath, amsthm}

\usepackage{xypic}%
\usepackage{scalerel}%
\usepackage{ulem}

\usepackage{graphicx}              %
\usepackage{amsmath}               %
\usepackage{amsfonts}              %
\usepackage{amsthm}                %
\usepackage{xkeyval}               %
\usepackage{amssymb}
\usepackage{enumerate}             %
\usepackage{color}
\usepackage{breqn}                 %
\usepackage{bm}                    %
\usepackage{cite}

\allowdisplaybreaks[1]
\usepackage{nicefrac}
\usepackage{booktabs}

\usepackage{comment}

\usepackage{kpfonts}
\usepackage{stackengine}
\usepackage{calc}
\newlength\shlength
\newcommand\xshlongvec[2][0]{\setlength\shlength{#1pt}%
	\stackengine{-5.6pt}{$#2$}{\smash{$\kern\shlength%
			\stackengine{7.55pt}{$\mathchar"017E$}%
			{\rule{\widthof{$#2$}}{.57pt}\kern.4pt}{O}{r}{F}{F}{L}\kern-\shlength$}}%
	{O}{c}{F}{T}{S}}

\usepackage{nicematrix}

\usepackage{enumitem}

\usepackage{hyperref}
\hypersetup{
	colorlinks   = true,    %
	citecolor    = red      %
}

\graphicspath{{./figures/}}

\newcommand{\RN}[1]{%
	\textup{\uppercase\expandafter{\romannumeral#1}}%
}

\allowdisplaybreaks

\newcommand{\meqref}[1]{\text{Eq}.~\eqref{#1}}
\newcommand{\mref}[1]{Sec.~$ \!\ref{#1} $}
\newcommand{\mfig}[1]{Fig.~$ \!\ref{#1} $}

\newtheorem{thm}{Theorem}[section]

\newtheorem{defn}[thm]{Definition} 
 
\newtheorem{remark}[thm]{Remark}

\arraycolsep=3pt\def\arraystretch{0.5}

\DeclareMathOperator{\Gfun}{G}

\DeclareMathOperator{\Wfun}{W}

\def\<{\langle}
\def\>{\rangle}

\numberwithin{equation}{section}

\usetikzlibrary{matrix,fit}

\usepackage{pgfplots}
\pgfplotsset{compat=1.17}

\makeatletter
\def\smallunderbrace#1{\mathop{\vtop{\m@th\ialign{##\crcr
				$\hfil\displaystyle{#1}\hfil$\crcr
				\noalign{\kern3\p@\nointerlineskip}%
				\tiny\upbracefill\crcr\noalign{\kern3\p@}}}}\limits}
\makeatother

\begin{document}
    \title{Generalized tensor transforms and their applications in classical and quantum computing}
	\author[1]{Alok Shukla \thanks{Corresponding author.}}
	\author[2]{Prakash Vedula}
	\affil[1]{School of Arts and Sciences, Ahmedabad University, India}
	\affil[1]{alok.shukla@ahduni.edu.in}
	\affil[2]{School of Aerospace and Mechanical Engineering, University of Oklahoma, USA}
	\affil[2]{pvedula@ou.edu}
	
\date{}

	\maketitle


\begin{abstract}

We introduce a novel framework for Generalized Tensor Transforms (GTTs), constructed through an $n$-fold tensor product of an arbitrary $b \times b$ unitary matrix $W$. This construction generalizes many established transforms, by providing a adaptable set of orthonormal basis functions. Our proposed fast classical algorithm for GTT achieves an exponentially lower complexity of $O(N \log_b N)$ in comparison to a naive classical implementation that has an associated computational cost of $O(N^2)$. For quantum applications, our GTT-based algorithm, implemented in the natural spectral ordering, achieves both gate complexity and circuit depth of $O(\log_b N)$, where $N = b^n$ denotes the length of the input vector. This represents a quadratic improvement over Quantum Fourier Transform (QFT), which requires $O((\log_b N)^2)$ gates and depth for $n$ qudits, and an exponential advantage over classical Fast Fourier Transform (FFT) based and Fast Walsh-Hadamard Transform (FWHT) based methods, which incur a computational cost of $O(N \log_b N)$.

We explore diverse applications of GTTs in quantum computing, including quantum state compression and transmission, function encoding and quantum digital signal processing. The proposed framework provides fine-grained control of the transformation through the adjustable parameters of the base matrix $W$. This versatility allows precise shaping of each basis function while preserving their effective Walsh-type structure, thus tailoring basis representations to specific quantum data and computational tasks. Our numerical results demonstrate that GTTs enable improved performance in quantum state compression and function encoding compared to fixed transforms (such as FWHT or FFT), achieving higher fidelities with fewer retained components. We also provided novel classical and quantum digital signal filtering algorithms based on our GTT framework.

\end{abstract}

\section{Introduction}
\label{sec:introduction}

The ability to efficiently analyze, process, and interpret complex data is fundamental across all areas of science and engineering.
Transform methods, which enable the representation of data using alternative sets of basis functions, are essential tools in a wide range of applications, including signal compression, noise reduction, feature extraction, and pattern recognition.
As data becomes increasingly high dimensional and intricate, there is a growing need for more flexible and customized transformation techniques. In this work, we present a novel framework for creating a class of orthonormal basis functions based on Generalized Tensor Transforms (GTTs). The GTT is constructed through the $n$-fold tensor product of an arbitrary $b \times b$ unitary matrix $W$. This approach extends classical transforms such as the Walsh-Hadamard and multidimensional Fourier transforms, and enables efficient computation through a fast butterfly-type algorithm (similar to the Fast Fourier Transform algorithm) with wide-ranging applications in several domains including signal and image processing.

Further, it is shown that the proposed GTT framework provides a powerful and adaptable mathematical 
tool for a range of applications in quantum computing, including in quantum signal processing and efficient function encoding on a quantum computer.

In the following, we describe the motivation of our work and our contributions in more detail.

\subsection{Background and motivation}
\label{ssec:background}
Unitary transforms form the cornerstone of modern signal processing, data analysis, and numerous scientific computing disciplines. By projecting data onto an alternative set of orthogonal or orthonormal basis functions, these transforms facilitate fundamental operations such as dimensionality reduction, noise filtering, and the extraction of salient features. Well-known examples include the Discrete Fourier Transform (DFT), crucial for frequency domain analysis, and the Walsh-Hadamard Transform (WHT), valued for its real-valued basis functions and applications in spectroscopy and coding theory. Their utility lies in their ability to decorrelate data, concentrate energy into fewer coefficients, and provide alternative perspectives that simplify subsequent processing tasks.

The motivation to study generalized transforms emerges from the field of quantum computing. In quantum computing, an $n$-qubit (or more generally, an $n$-qudit) system naturally exists in a tensor product Hilbert space. Consequently, many fundamental operations that act independently on each quantum subsystem are represented as tensor products of smaller unitary transforms. This inherent structure naturally prompts the question of performing a spectral analysis on these composite transforms and identifying the types of classical or quantum signals that can be efficiently represented or processed using their resulting basis functions. It turns out that such a generalized transform, derived from the tensor product of an arbitrary base unitary matrix, naturally encompasses and extends many well-known transforms, including the Walsh-Hadamard Transform (WHT) and the multidimensional Discrete Fourier Transform (DFT), which find diverse applications from spectroscopy to coding theory. This versatility highlights the imperative for a flexible and adaptable framework for generating transform bases to unlock new analytical capabilities in high-dimensional and quantum datasets.

While transforms such as the Walsh-Hadamard and Discrete Fourier Transforms have proven invaluable, their inherent design presents certain limitations. These traditional transforms are fixed and generate a predetermined set of basis functions derived from specific underlying algebraic structures ($\mathbb{Z}_2^n$ for WHT, $\mathbb{Z}_N$ for DFT). Consequently, their ability to optimally represent or efficiently process all forms of data is constrained. For instance, data exhibiting specific symmetries, correlations, or sparsity patterns that do not align with the intrinsic properties of sine/cosine waves or Walsh functions may not be optimally compressed or analyzed. The need for a more adaptable and customizable framework for generating transform bases, capable of being tailored to the specific nuances of a given dataset, thus becomes apparent. 

Further, while unitary gates and their tensor products naturally constitute the fabric of quantum computing, a systematic exploration of their spectral properties, and the implications for designing a generalized family of customizable transforms, has not been thoroughly investigated in the existing literature. This represents a significant gap, as understanding and harnessing the spectral characteristics of such generalized tensor transforms is important for unlocking the full potential of quantum computing. By leveraging these properties, we can develop highly adaptable and efficient quantum algorithms, enabling optimal data representation and processing tailored to specific problem structures. This innovative approach opens new avenues for addressing complex challenges across a wide range of domains, including quantum signal and image processing, control systems, and various data analysis and machine learning tasks, where classical frameworks like the Fast Fourier Transform (FFT) or Fast Walsh-Hadamard Transform (FWHT) are traditionally applied. Our work bridges this gap by systematically exploring these spectral properties, providing a novel mathematical foundation and demonstrating its practical utility.

While several generalizations of the Walsh-Hadamard Transform have been discussed in the literature \cite{grozdanov_inequality_nodate, horadam_generalised_2005, wang_improved_2010, watari1958generalized, episkoposian_greedy_2011, yuan2021generalized}, our approach to Generalized Tensor Transforms (GTTs) introduces a distinct and more powerful paradigm. Unlike these existing methods, which often impose specific structural or algebraic restrictions on the base matrix, the GTT framework allows for the selection of an arbitrary unitary base matrix $\Wfun$. 
This flexibility allows for fine-grained adaptation of basis functions to suit the specific requirements of diverse classical and quantum applications. In particular, the unique tunability is especially valuable for quantum computing, where, to the best of our knowledge, this is the first work to systematically study and analyze the spectral properties of unitary transformations with a focus on their application in quantum computing. This opens new avenues for quantum solutions in domains that rely on spectral techniques, such as signal processing, image analysis, and function encoding.

The Generalized Tensor Transform (GTT) framework fundamentally encompasses and extends well-established transforms such as the Walsh-Hadamard Transform (WHT) and the multi-dimensional Discrete Fourier Transform (DFT). This inherent generalization implies that GTTs inherit the broad applicability of WHT in diverse fields. For instance, some of the applications of WHT include signal processing \cite{beauchamp1975walsh, shapiro_walsh_1974, shukla2023quantumdsp, broadbent_analysis_1992}, image processing~\cite{smith_walsh_1979, shukla2022hybridimage, rohida2024hybrid, o1978edge}, differential equations~\cite{shukla2023hybridode}, control systems~\cite{palanisamy_minimum_1983, paraskevopoulos_transfer_1980}, genetic algorithms \cite{goldberg_genetic_nodate}, quantum search~\cite{grover1996fast, shukla2025efficientsearch}, uniform quantum state preparation~\cite{shukla2024efficient, bernstein1993quantum, shukla2023generalizationpbv} and cryptography~\cite{shor1994algorithms}. However, the GTT's unique tunability, derived from an arbitrary unitary base matrix, goes beyond mere equivalence by offering the potential for significantly enhanced performance and novel solutions in these domains, particularly by enabling optimal data representation where fixed transforms may be suboptimal.

\subsection{Generalized tensor transforms and their basis functions}
\label{ssec:contribution}
We introduce a novel and significantly generalized family of orthonormal basis functions. These Generalized Tensor Transform basis functions (or GTT basis functions) are recursively defined based on an arbitrary $b \times b$ matrix $\baseW$. If $\Wfun$ is a unitary matrix, these basis functions inherently form a complete orthonormal set. This allows us to precisely tailor the basis to specific data characteristics, inherent symmetries, or computational demands, offering a flexibility not afforded by fixed transforms.

We note that the Generalized Tensor Transform (GTT) is the $n$-fold tensor product of an arbitrary $b \times b$ unitary matrix, denoted as $\Gfun_N = \Wfun^{\otimes n}$, where $N=b^n$. Indeed, the columns of this GTT matrix are precisely obtained using these novel recursively defined basis functions. This unified framework significantly broadens the family of available unitary transforms, inherently encompassing and extending well-established methods such as the Walsh-Hadamard Transform (where $\Wfun$ is the $2 \times 2$ Hadamard matrix) and the multi-dimensional Discrete Fourier Transform (where $\Wfun$ is the $b \times b$ DFT matrix). Furthermore, we develop a fast, butterfly-type algorithm (like FFT) for the efficient computation of these generalized transforms, ensuring their practical applicability even for high-dimensional data. While these GTTs offer a unified approach to classical fast transforms with potential benefits across diverse signal and image processing tasks, their primary motivation and a central contribution of this work is their profound utility and unique advantages in quantum computing applications.

\subsection{Applications in quantum computing}
\label{ssec:applications_qc}

Quantum computing is a rapidly emerging field with the potential to transform numerous scientific domains. In recent years, several quantum algorithms have been developed that either demonstrate a theoretical advantage over the best-known classical algorithms or offer promising applications across diverse areas \cite{deutsch1992rapid, bernstein1993quantum, simon1997power, grover1996fast, Brassard_2002, shukla2025efficientsearch, suzuki2020amplitude,  shukla2024partialsum, shor1994algorithms, harrow2009quantum}.

The flexible structure of Generalized Tensor Transforms makes them particularly well suited for a wide range of applications in quantum computing. Quantum systems naturally employ tensor product structures to describe multi-qubit and multi-qudit states, and many fundamental quantum operations are expressed as tensor products of single-subsystem unitaries. Our framework builds directly on this structure, enabling the systematic design of novel unitary operators tailored for quantum architectures.

These transforms are especially relevant for quantum digital signal processing, where they can support tasks such as encoding classical signals into quantum states, extracting information from quantum-encoded signals, and implementing quantum filtering operations. By selecting an appropriate base unitary matrix \( \Wfun \), the resulting basis functions can be adapted to align with the structure of the underlying quantum data, potentially enhancing signal reconstruction, improving robustness to noise, and facilitating more effective feature separation in the quantum domain.

Moreover, the ability to generate a rich family of unitary transforms opens new possibilities in quantum data analysis and machine learning, enabling the construction of domain-specific transformations that support algorithmic innovation in areas such as quantum state compression and efficient transmission (ref.~\mref{ssec:state_compression}), quantum function encoding (ref.~\mref{sec:fun_encoding}), and signal filtering (ref.~\mref{ssec:filtering}). The efficient quantum circuit implementation of these transforms further underscores their practical relevance for near-term and future quantum hardware platforms.

\subsection*{A brief outline}
\label{ssec:structure}

In Section \ref{sec:gtt_basis_functions}, we provide the mathematical foundations, formally define the GTT basis functions, present their recursive definitions for a general base $b$, and prove their orthonormality. This section also explores GTT-Fourier Series and the discretization leading to the unitary GTT matrix.
Section \ref{sec:algorithm}, focuses on the computational aspects, presenting a fast butterfly-type algorithm, along with an analysis of its computational advantages.
Section \ref{sec:gtt_applications} demonstrates the practical utility of GTTs in quantum computing, covering quantum state compression (classical, hybrid, and fully quantum protocols with numerical examples), complexity analysis, and function encoding applications.
Finally, Section \ref{sec:conclusion}, summarizes the paper's key findings, emphasizing the significance of our GTT framework, its efficient algorithm, and its wide-ranging implications for quantum computing, while also suggesting future research directions.

\section{Generalized tensor transform basis functions}
\label{sec:gtt_basis_functions}

At the heart of any transform are the basis functions it employs. The Generalized Tensor Transform (GTT) is intrinsically linked to a complete orthonormal set of such basis functions. These functions, which we refer to as GTT basis functions, offer a robust generalization of established families including the Walsh-Hadamard and Vilenkin-Chrestenson functions. In this section, we provide a detailed account of their recursive construction and explore their fundamental mathematical properties.




\subsection{Notation for basis functions}
\label{ssec:notation_basis}
For clarity in defining these functions, we adopt the following notations:
\begin{itemize}
    \item $\{ x \} = x - \lfloor x \rfloor$: Denotes the fractional part of $x$, where $\lfloor x \rfloor$ is the floor function.
    \item Matrix elements $\Wfun_{i,j}$ refer to the element in the $i$-th row and $j$-th column of matrix $\Wfun$, with indices starting from $0$. 
    \item For a sequence $k_{n-1}, \ldots, k_1, k_0 \in \{0, 1, \ldots, b-1\}$, we can associate a base $b$ number $k = b^{n-1} k_{n-1} + \cdots + b k_1 + k_0$.
\end{itemize}

\subsection{Recursive definition for general base $b$}
\label{ssec:recursive_general_b}
The GTT basis functions of order $N=b^n$, denoted $f_j^{(b^n)}(x)$, are defined recursively over the interval $[0,1)$. The base case for $n=0$ (i.e., $N=b^0=1$) is:
\begin{equation} \label{eq:gtt_base_case_general}
f_0^{(1)} (x) = 1, \quad \quad 0 \leq x < 1,
\end{equation}
and $0$ otherwise. For $n \geq 1$, and for any index $j \in \{0, 1, \ldots, b^n-1\}$, the GTT basis functions are defined as:
\begin{equation} \label{eq:gtt_recursive_general_b_main}
f_{j}^{( b^{n} )} (x) = \Wfun_{\lfloor b x \rfloor, \lfloor j / b^{n-1} \rfloor} \cdot f_{j \pmod{b^{n-1}}}^{(b^{n-1})}\left(\{b x\}\right), \quad \text{for } 0 \leq x < 1.
\end{equation}
In this definition, the index $j$ represents the specific basis function, ranging from $0$ to $b^n-1$. Further, $\lfloor b x \rfloor$ determines the row index of the base matrix $\Wfun$ to use, corresponding to the current sub-interval of $x$. The expression $\lfloor j / b^{n-1} \rfloor$ determines the column index of the base matrix $\Wfun$ to use, based on the most significant ``digit'' of $j$ in base $b$. The term $j \pmod{b^{n-1}}$ extracts the less significant ``digits'' of $j$ for the recursive call and $\{b x\}$ maps the current sub-interval of $x$ back to the $[0,1)$ domain for the recursive call.

This concise form captures the essence of the tensor product operation, where the value of a function at a point $x$ is determined by a product of elements from $\Wfun$ corresponding to the ``digits'' of $x$ in base $b$ and the ``digits'' of the function index $j$ in base $b$. These GTT basis functions are piece-wise constant over $b^n$ sub-intervals of $[0,1)$.

\begin{remark}
    We note that \meqref{eq:gtt_recursive_general_b_main} is equivalent to the following form, which highlights the role of the base $b$, digits of $x$ and $j$:
    \begin{equation}
        f_{j_{n-1}j_{n-2}\ldots j_0}^{(b^n)} (0.x_1x_2\ldots) = \Wfun_{x_1, j_{n-1}} \cdot f_{j_{n-2}\ldots j_0}^{(b^{n-1})} (0.x_2x_3\ldots)
    \end{equation}
    where $j = j_{n-1} j_{n-2} \ldots j_0$ is the base $b$ representation of the index $j$, and $x = 0.x_1 x_2 \ldots $ is the base $b$ representation of $x$ in the interval $[0,1)$.
\end{remark}

\subsection{Specific cases: Base $b=2$ and $b=3$}
\label{ssec:specific_cases}
To further illustrate the general recursive definition, we provide specific examples for common bases.

\subsubsection{Case $b=2$}
When the base matrix $\Wfun$ is a $2 \times 2$ unitary matrix with entries $\Wfun_{0,0}$, $\Wfun_{0,1}$, $\Wfun_{1,0}$, and $\Wfun_{1,1}$, the general definition simplifies to the following. For $n \geq 1$, let $k = 2^{n-1}$. The functions are recursively defined for $0 \leq i \leq k-1$:
\begin{align}\label{eq:gtt_recur_first_b2}
f_{i}^{(2k)} (x) &=
\begin{cases}
    \Wfun_{0,0} \, f_i^{(k)} (2x), & 0 \leq x < \frac{1}{2}, \\
    \Wfun_{1,0} \, f_i^{(k)} (2x-1), & \frac{1}{2} \leq x < 1, \\
\end{cases} \\
f_{k+i}^{(2k)} (x) &=
\begin{cases}
    \Wfun_{0,1} \, f_i^{(k)} (2x), & 0 \leq x < \frac{1}{2}, \\
    \Wfun_{1,1} \, f_i^{(k)} (2x-1), & \frac{1}{2} \leq x < 1.\\
\end{cases}\label{eq:gtt_recur_second_b2}
\end{align}
For the simplest case where $N=2$ (i.e., $n=1$, $k=1$, $i=0$), these definitions yield $f_0^{(2)}(x)$ and $f_1^{(2)}(x)$, which correspond precisely to the columns of the base matrix $\Wfun$. Importantly, when $\Wfun_{0,0} = \Wfun_{0,1} = \Wfun_{1,0} =1$ and $\Wfun_{1,1} = -1$ (i.e., $\Wfun = \frac{1}{\sqrt{2}}\begin{bmatrix} 1 & 1 \\ 1 & -1 \end{bmatrix}$ if normalized, or a scaled version), these functions become the well-known Walsh-Hadamard basis functions in natural order. The Walsh basis functions in natural order,  for $N=8$, are shown in \mfig{fig:Walsh_bfN8}. 

Walsh basis functions are a set of orthonormal functions defined over the interval \([0,1]\) that take only the values \(\pm 1\). Unlike the sinusoidal basis used in Fourier analysis, Walsh functions are piecewise constant. 
This binary structure allows them to capture discontinuities and rapid transitions more naturally than smooth trigonometric functions. They form a complete orthogonal basis in \(L^2[0,1]\), making them suitable for a wide range of functional expansions \cite{beauchamp1975walsh, walsh1923closed, shukla2023hybridode}. 

Walsh functions have been widely employed in both classical and quantum domains due to their fast transform algorithms and binary simplicity. The Walsh--Hadamard Transform (WHT) is particularly useful in data compression, noise filtering~\cite{shukla2022hybridimage, shukla2023quantumdsp, rohida2024hybrid}, and pattern recognition, where it provides a computationally efficient alternative to the discrete Fourier transform, especially on digital hardware that favors additions and subtractions over multiplications. Because of their ability to represent abrupt transitions, Walsh functions are also effective in representing and solving certain types of differential equations~\cite{beer1981walsh,ahner1988walsh,shukla2023hybridode}, particularly those with discontinuous or piecewise-defined boundary conditions. Furthermore, Walsh expansions have been used in control theory, digital communication systems (e.g., code division multiple access), and biomedical signal processing (e.g., EEG compression), owing to their compact support and efficient hardware implementation. 
Besides the natural ordering of Walsh basis functions considered above, other orderings like sequency and dyadic based orderings have also been explored~\cite{geadah1977natural, shukla2024sequency,shukla2022quantumzerocrossings}. In specific applications, some ordering might be preferable over others.


\begin{figure}[htbp]
  \centering

  \begin{minipage}{0.23\textwidth}
    \centering
    \includegraphics[width=\linewidth]{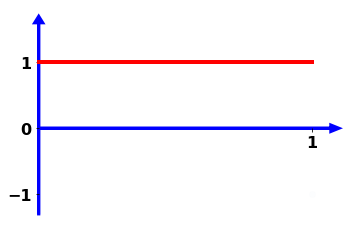}
    \subcaption{$f_0^{(8)}(x)$ }
  \end{minipage}
  \hfill
  \begin{minipage}{0.23\textwidth}
    \centering
    \includegraphics[width=\linewidth]{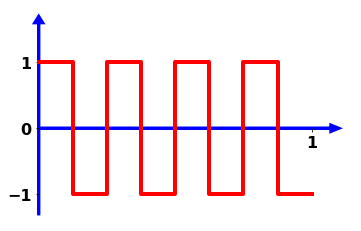}
    \subcaption{$f_1^{(8)}(x)$ }
  \end{minipage}
  \hfill
  \begin{minipage}{0.23\textwidth}
    \centering
    \includegraphics[width=\linewidth]{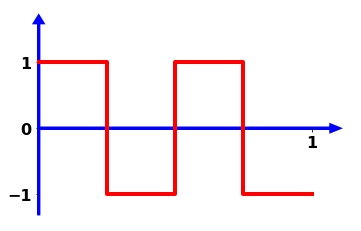}
    \subcaption{$f_2^{(8)}(x)$ }
  \end{minipage}
  \hfill
  \begin{minipage}{0.23\textwidth}
    \centering
    \includegraphics[width=\linewidth]{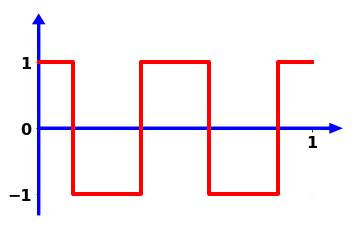}
    \subcaption{$f_3^{(8)}(x)$ }
  \end{minipage}

  \vspace{1em}

  \begin{minipage}{0.23\textwidth}
    \centering
    \includegraphics[width=\linewidth]{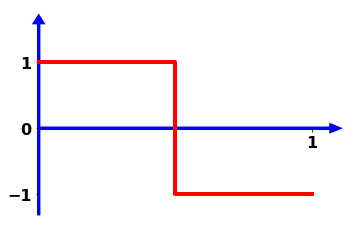}
    \subcaption{$f_4^{(8)}(x)$ }
  \end{minipage}
  \hfill
  \begin{minipage}{0.23\textwidth}
    \centering
    \includegraphics[width=\linewidth]{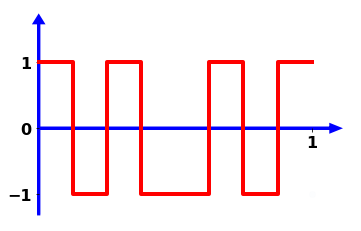}
    \subcaption{$f_5^{(8)}(x)$ }
  \end{minipage}
  \hfill
  \begin{minipage}{0.23\textwidth}
    \centering
    \includegraphics[width=\linewidth]{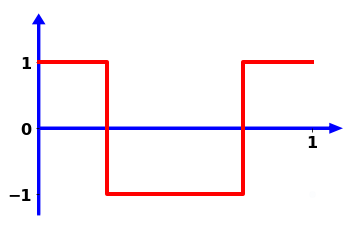}
    \subcaption{$f_6^{(8)}(x)$ }
  \end{minipage}
  \hfill
  \begin{minipage}{0.23\textwidth}
    \centering
    \includegraphics[width=\linewidth]{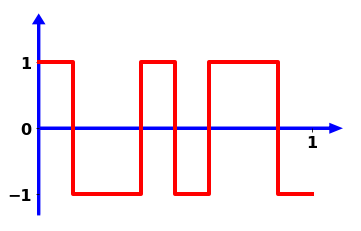}
    \subcaption{$f_7^{(8)}(x)$ }
  \end{minipage}

  \caption{Walsh basis functions in natural order labeled from $f_0^{(8)}(x)$  to $f_7^{(8)}(x)$. }

\label{fig:Walsh_bfN8}
  
\end{figure}

\begin{figure}[htbp]
  \centering
  \begin{minipage}{0.48\textwidth}
    \centering
    \includegraphics[width=\linewidth]{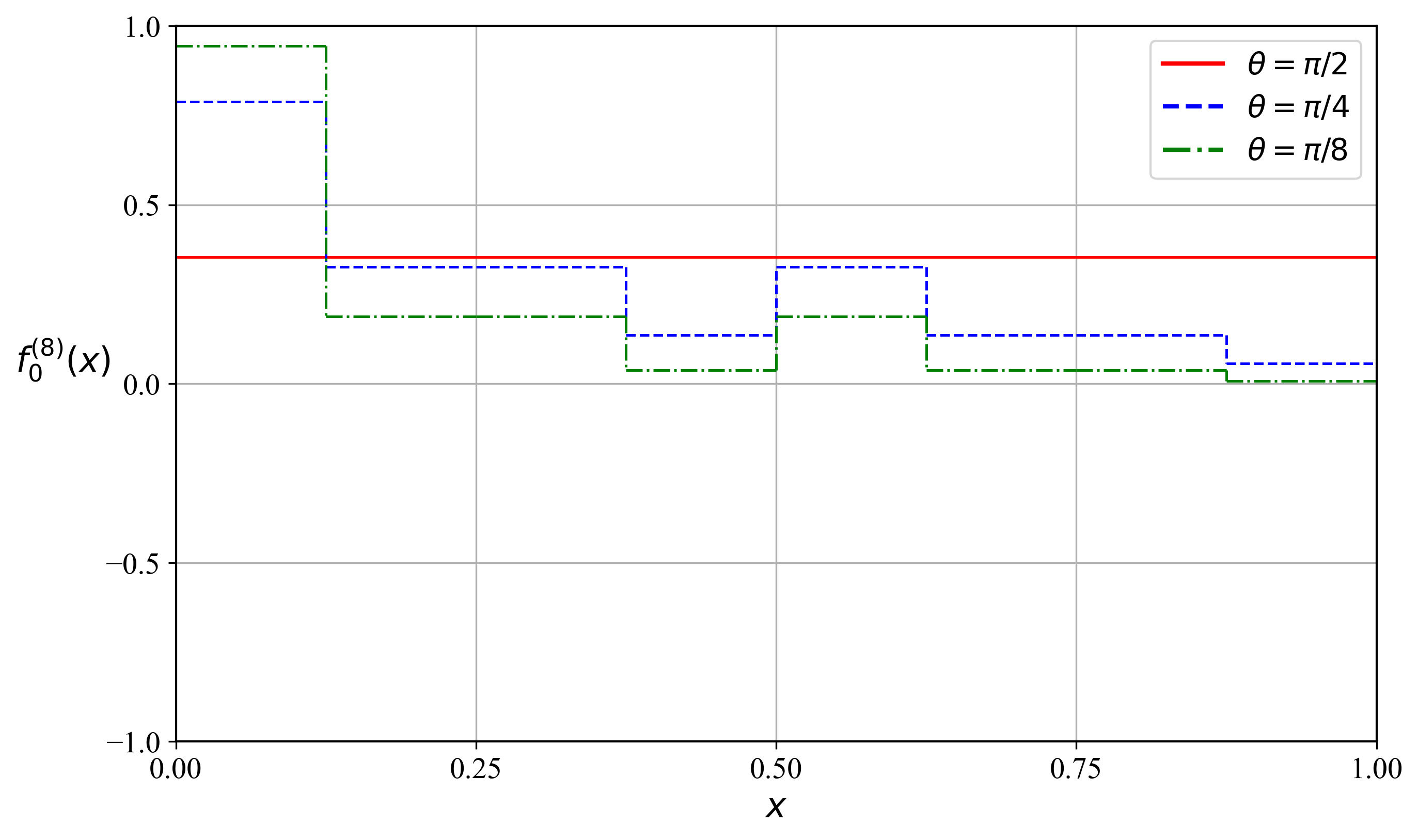}
    \subcaption{}
  \end{minipage}%
  \hfill
  \begin{minipage}{0.48\textwidth}
    \centering
    \includegraphics[width=\linewidth]{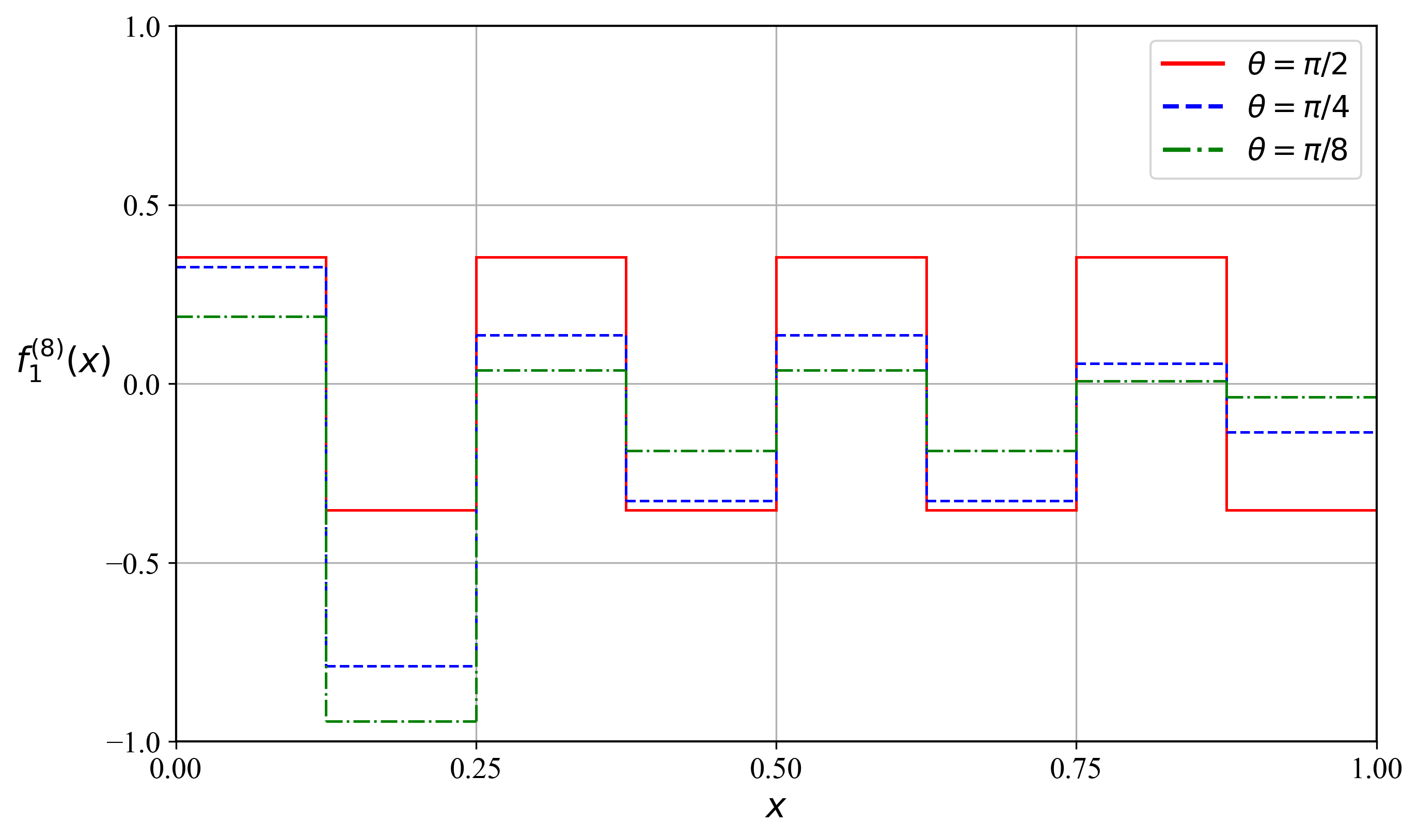}
    \subcaption{}
  \end{minipage}

  \vspace{1em}

  \begin{minipage}{0.48\textwidth}
    \centering
    \includegraphics[width=\linewidth]{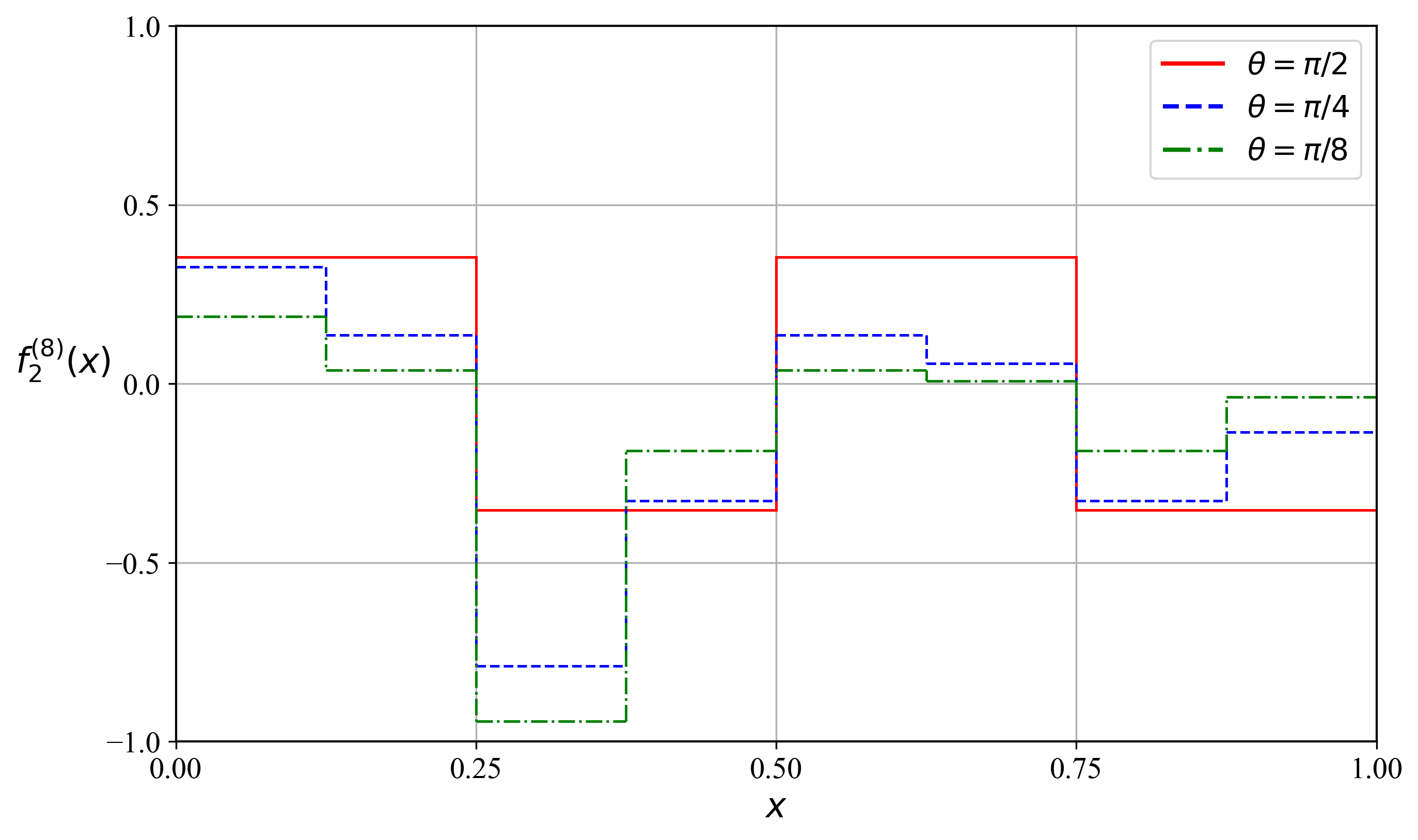}
    \subcaption{}
  \end{minipage}%
  \hfill
  \begin{minipage}{0.48\textwidth}
    \centering
    \includegraphics[width=\linewidth]{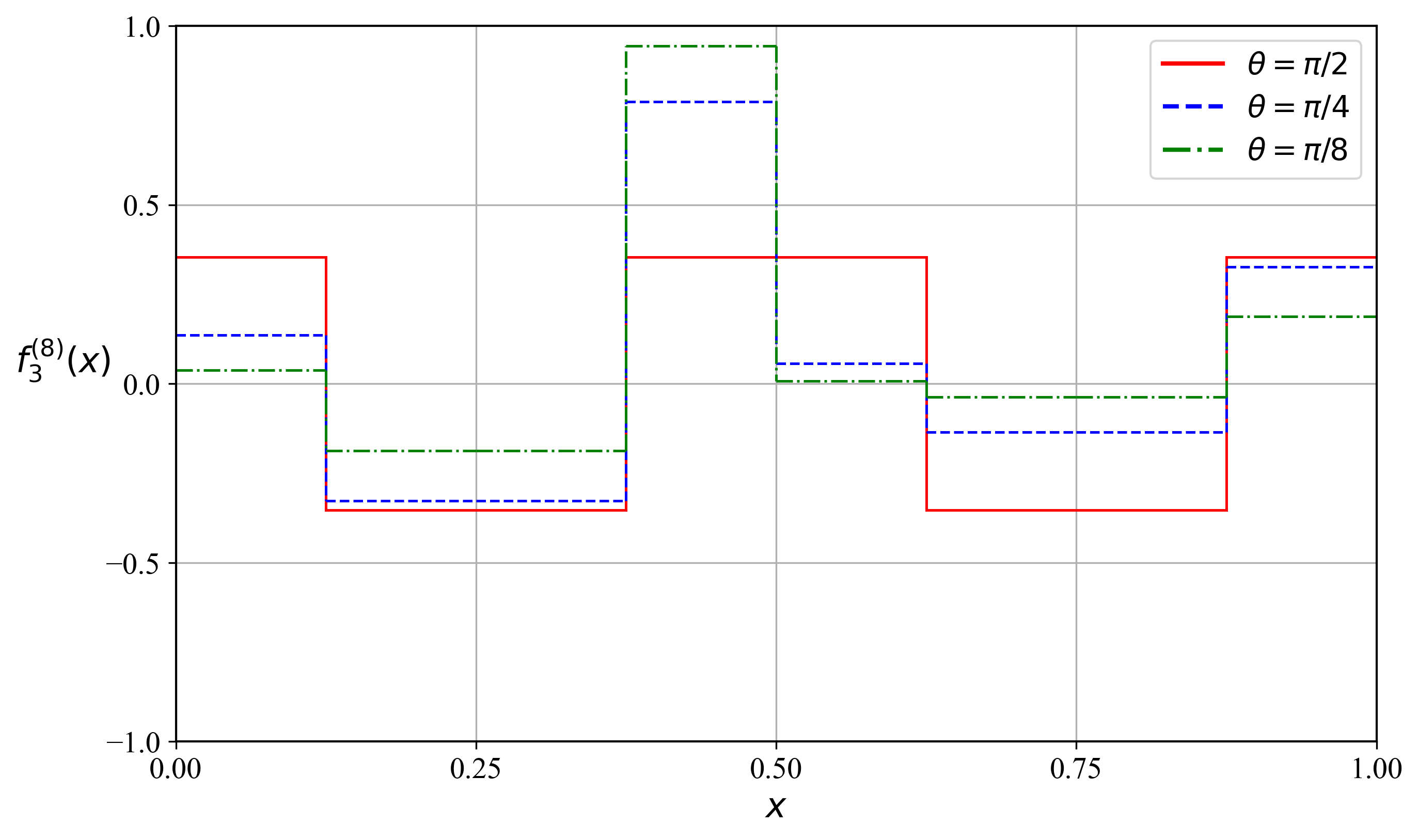}
    \subcaption{}
  \end{minipage}

  \vspace{1em}

  \begin{minipage}{0.48\textwidth}
    \centering
    \includegraphics[width=\linewidth]{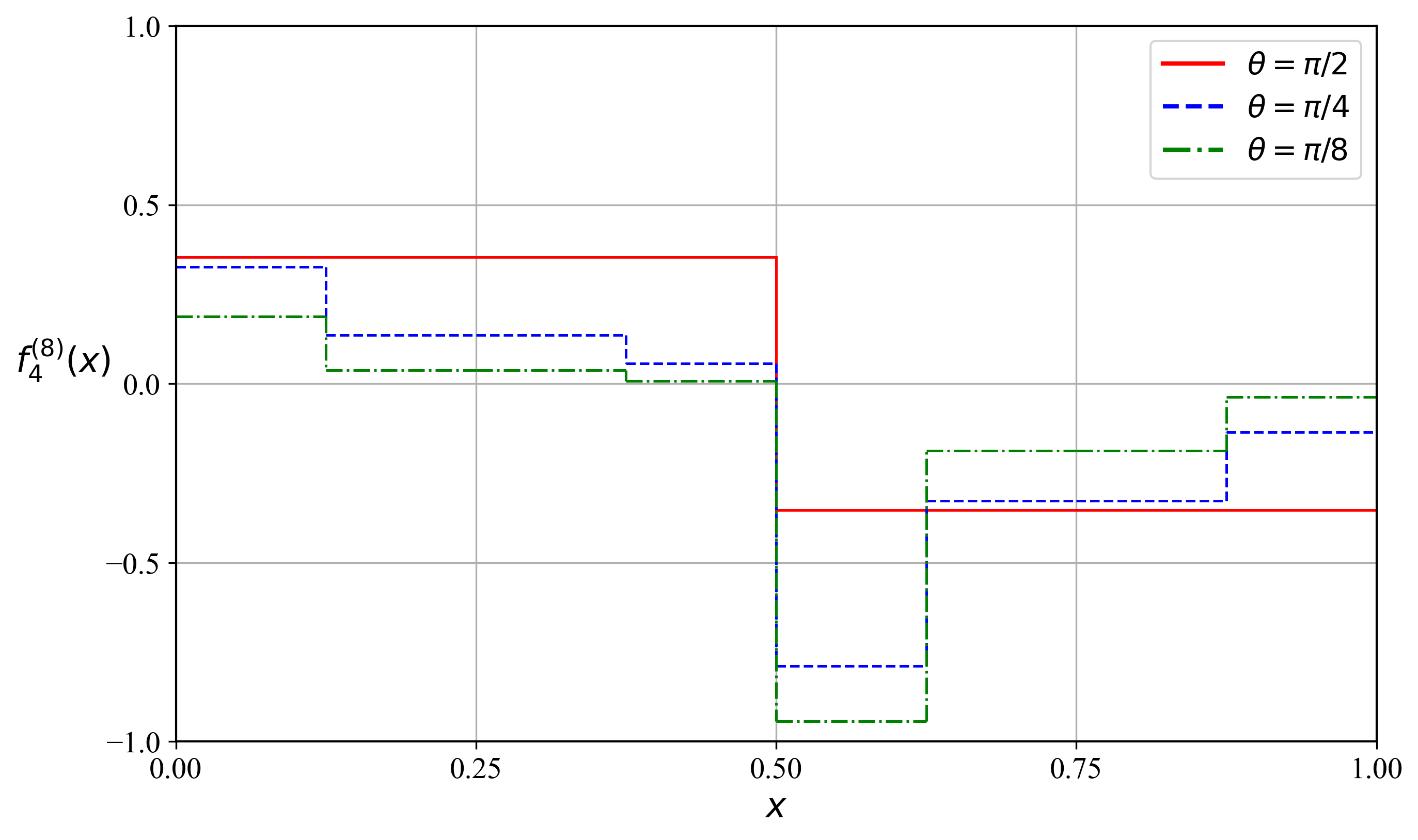}
    \subcaption{}
  \end{minipage}%
  \hfill
  \begin{minipage}{0.48\textwidth}
    \centering
    \includegraphics[width=\linewidth]{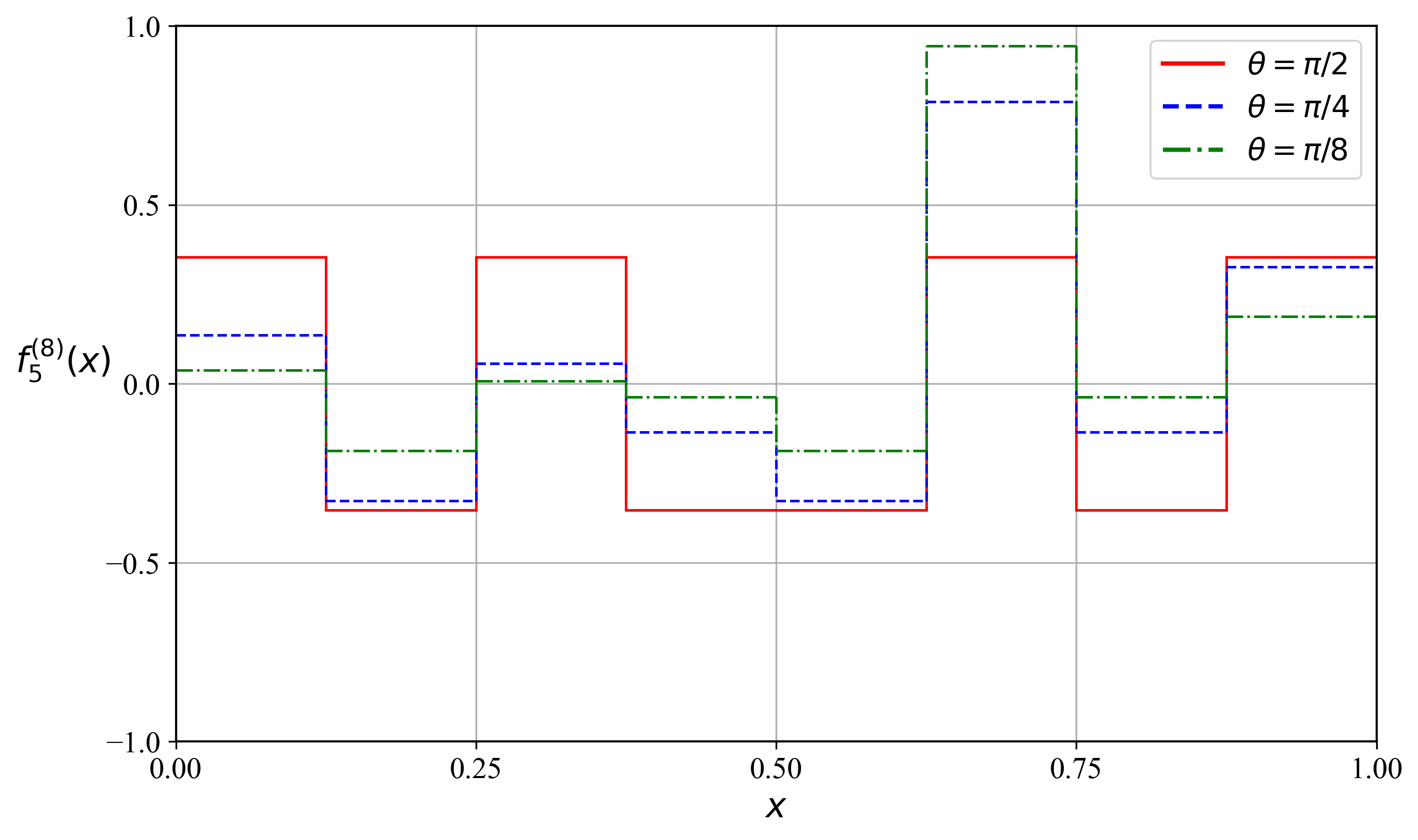}
    \subcaption{}
  \end{minipage}

  \vspace{1em}

  \begin{minipage}{0.48\textwidth}
    \centering
    \includegraphics[width=\linewidth]{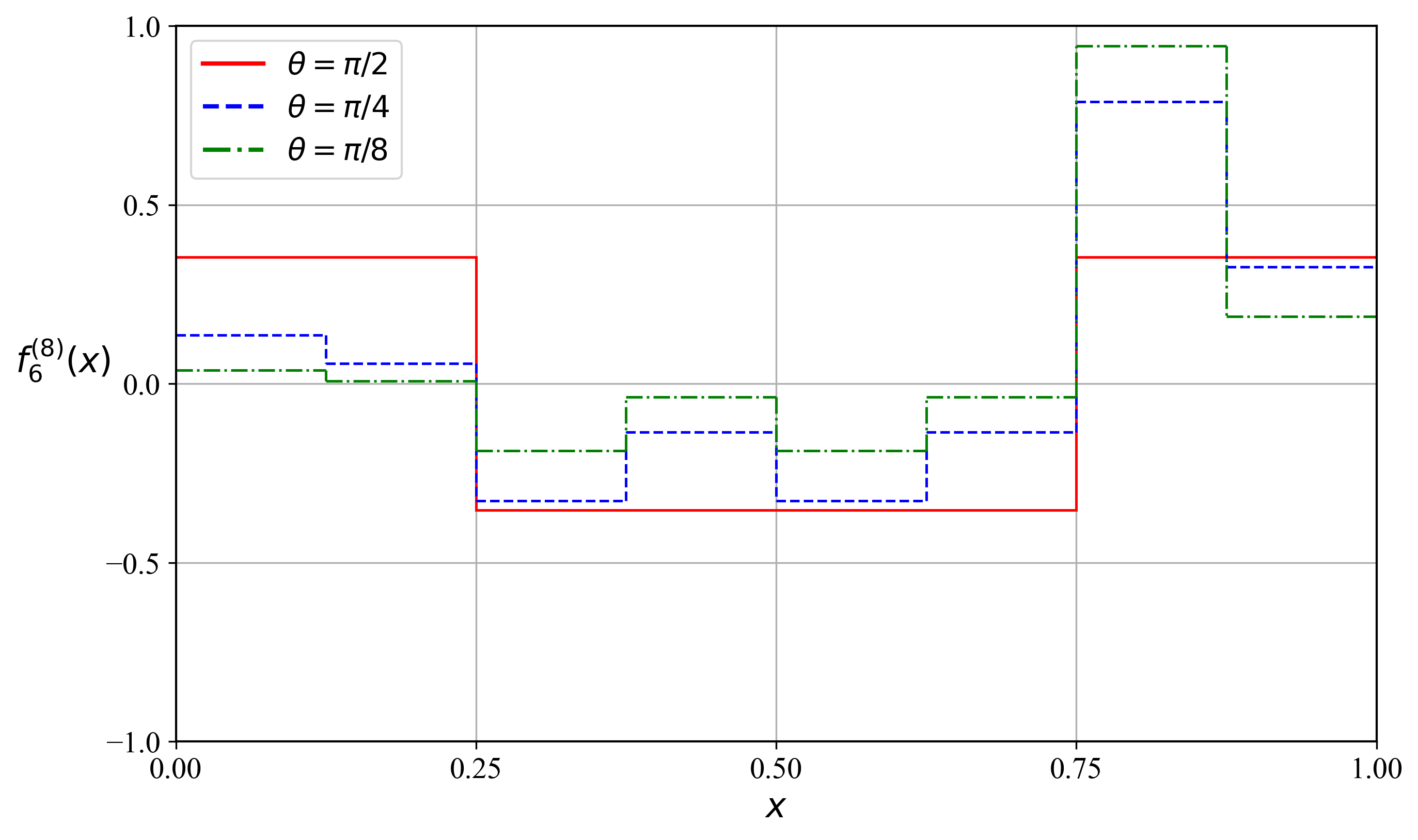}
    \subcaption{}
  \end{minipage}%
  \hfill
  \begin{minipage}{0.48\textwidth}
    \centering
    \includegraphics[width=\linewidth]{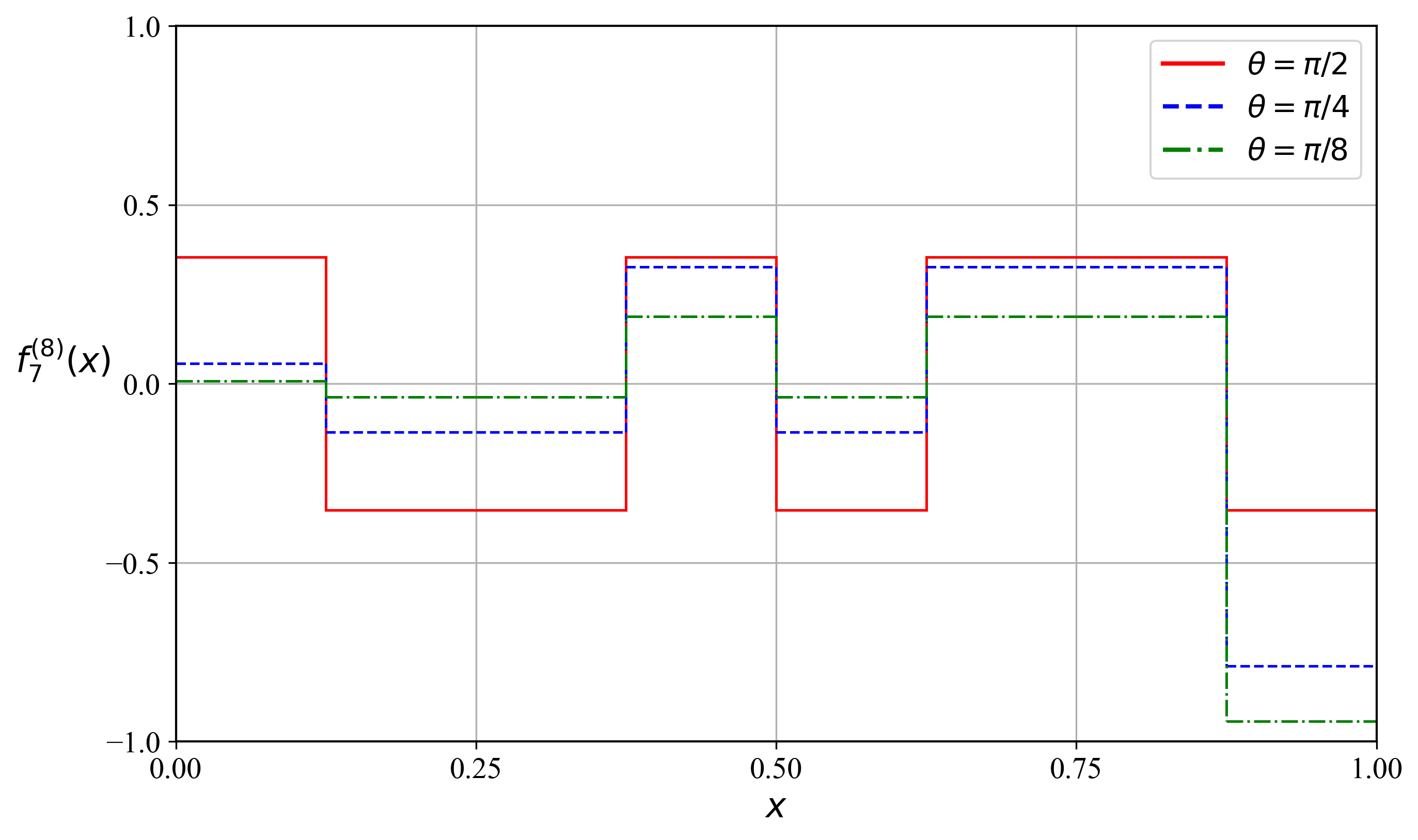}
    \subcaption{}
  \end{minipage}


\caption{
GTT basis functions of order N=8 for $\theta = \pi/2, ~\pi/4, ~\pi/8$ are shown in subfigures (a)--(h). Note that the $\theta = \pi/2$ case corresponds to Walsh basis functions in natural order shown in \mfig{fig:Walsh_bfN8}.
}

\label{fig:GTT_fun}

\end{figure}

GTT basis functions provide a powerful generalization of the Walsh basis by enabling the use of an arbitrary unitary base matrix $\Wfun$ to construct the basis set. As noted earlier, choosing 
\[
\Wfun = \frac{1}{\sqrt{2}}\begin{bmatrix} 1 & 1 \\ 1 & -1 \end{bmatrix}
\]
reproduces the canonical Walsh functions in natural order.
To illustrate the flexibility of the GTT framework, we consider a general unitary matrix of the form
\begin{equation} \label{eq:U3}
    \Wfun = U3(\theta, \phi, \lambda) =
    \begin{bmatrix}
        \cos\left(\frac{\theta}{2}\right) & -e^{i\lambda} \sin\left(\frac{\theta}{2}\right) \\
        e^{i\phi} \sin\left(\frac{\theta}{2}\right) & e^{i(\phi + \lambda)} \cos\left(\frac{\theta}{2}\right)
    \end{bmatrix},
\end{equation}
which corresponds to the \texttt{U3Gate} used in IBM's quantum computing framework Qiskit. For simplicity, we fix the phase parameters to $\phi = 0$ and $\lambda = \pi$, allowing only the angle $\theta$ to vary.
By adjusting $\theta$, one can systematically alter the structure of the resulting GTT basis functions. \mfig{fig:GTT_fun} visualizes the GTT basis functions of order $N = 8$ for three representative values of the tunable parameter $\theta$: $\pi/2$, $\pi/4$, and $\pi/8$, shown in subfigures (a)--(h). 
While all basis functions retain the hallmark discontinuous, piecewise-constant behavior characteristic of Walsh functions, their waveform geometries can be systematically varied with $\theta$, offering fine-grained control over basis structure. The special case $\theta = \pi/2$ recovers the standard Walsh basis, also shown separately in \mfig{fig:Walsh_bfN8}.

This example highlights the expressive versatility of the GTT framework: by varying $\Wfun$, one can tailor the shape of each basis function without sacrificing the discrete, piecewise-constant discontinuous behavior that makes Walsh-type bases useful in both classical and quantum applications.

\subsubsection{Case $b=3$}
For a $3 \times 3$ unitary base matrix $\Wfun$, the general recursive definition expands for $k=3^{n-1}$ and $0 \leq i \leq k-1$ to:
\begin{align}
f_{i}^{(3k)} (x) &=
\begin{cases}
    \Wfun_{0,0} \, f_i^{(k)} (3x), & 0 \leq x < \frac{1}{3}, \\
    \Wfun_{1,0} \, f_i^{(k)} (3x-1), & \frac{1}{3} \leq x < \frac{2}{3}, \\
    \Wfun_{2,0} \, f_i^{(k)} (3x-2), & \frac{2}{3} \leq x < 1, \\
\end{cases} \\
f_{k+i}^{(3k)} (x) &=
\begin{cases}
    \Wfun_{0,1} \, f_i^{(k)} (3x), & 0 \leq x < \frac{1}{3}, \\
    \Wfun_{1,1} \, f_i^{(k)} (3x-1), & \frac{1}{3} \leq x < \frac{2}{3}, \\
    \Wfun_{2,1} \, f_i^{(k)} (3x-2), & \frac{2}{3} \leq x < 1, \\
\end{cases} \\
f_{2k+i}^{(3k)} (x) &=
\begin{cases}
    \Wfun_{0,2} \, f_i^{(k)} (3x), & 0 \leq x < \frac{1}{3}, \\
    \Wfun_{1,2} \, f_i^{(k)} (3x-1), & \frac{1}{3} \leq x < \frac{2}{3}, \\
    \Wfun_{2,2} \, f_i^{(k)} (3x-2), & \frac{2}{3} \leq x < 1. \\
\end{cases}
\end{align}

\subsection{Properties of GTT basis functions}
\label{ssec:properties_basis}
The recursively defined GTT basis functions possess several important properties that underpin their utility as a powerful analytical tool. Foremost among these, if the base matrix $\Wfun$ is unitary, then the set of normalized GTT basis functions $  b^{n/2} \{f_j^{(b^n)}(x)\}_{j=0}^{b^n-1}$ forms a complete orthonormal set over the interval $[0,1)$. A detailed proof of this orthonormality is provided in Section~\ref{sssec:proof_orthonormality}. This property is a direct consequence of the tensor product construction, as the $N \times N$ matrix $\Wfun^{\otimes n}$ is itself unitary when $\Wfun$ is unitary, and its columns (when sampled appropriately over $b^n$ sub-intervals) constitute these basis functions. Each GTT basis function is also piece-wise constant across $b^n$ equally spaced sub-intervals of $[0,1)$, with its value within each sub-interval determined by the product of elements from the base matrix $\Wfun$. This piecewise constancy, combined with the tunability offered by the choice of $\Wfun$, provides immense flexibility for representing and transforming signals that exhibit localized features or specific frequency characteristics. The ability to tailor the basis through $\Wfun$ provides a potential advantage in a broad range of classical and quantum applications including efficient quantum state compression and transmission, quantum function encoding, and quantum digital signal processing applications (ref.~\mref{sec:gtt_applications}).

\subsubsection{Proof of orthonormality}
\label{sssec:proof_orthonormality}

We define the standard $L^2$ inner product on $[0,1)$ as:
\begin{equation}
    \braket{f | g} = \int_0^1 \overline{f(x)} g(x) \, dx.
\end{equation}
Next we define the normalized GTT basis functions $\phi_j^{(b^n)}(x)$ as:
\begin{equation}
    \phi_j^{(b^n)}(x) = b^{n/2} f_j^{(b^n)}(x).
    \label{eq:normalized_basis_phi}
\end{equation}

\begin{thm}[Orthonormality of GTT Basis Functions] \label{thm_ortho}
    Let $\Wfun$ be a unitary matrix. Then the normalized GTT basis functions $\phi_{j}^{(b^n)}(x)$, for $j = 0, 1, \dots, b^n - 1$, form a complete orthonormal set with respect to the $L^2$ inner product, such that:
    \begin{equation}
        \braket{\phi_j^{(b^n)} | \phi_k^{(b^n)}} = \delta_{jk}.
    \end{equation}
\end{thm}

\begin{proof}
We will prove this by first demonstrating that $\braket{f_j^{(b^n)} | f_k^{(b^n)}} = \frac{\delta_{jk}}{b^n}$ by induction on $n$, as the orthonormality of $\phi_j^{(b^n)}(x)$ follows form this.

\noindent \textit{Base Case ($n=1$):}
For $n=1$, the basis functions are $f_j^{(b)}(x) = \Wfun_{\lfloor b x \rfloor, j}$. These functions take the value $\Wfun_{m,j}$ for $x \in [\frac{m}{b}, \frac{m+1}{b})$, where $m = \lfloor b x \rfloor$.
The inner product is:
\begin{align*}
\braket{f_j^{(b)} | f_k^{(b)}} &= \int_0^1 \overline{\Wfun_{\lfloor b x \rfloor, j}} \Wfun_{\lfloor b x \rfloor, k} \, dx \\
&= \sum_{m=0}^{b-1} \int_{\frac{m}{b}}^{\frac{m+1}{b}} \overline{\Wfun_{m,j}} \Wfun_{m,k} \, dx \\
&= \sum_{m=0}^{b-1} \overline{\Wfun_{m,j}} \Wfun_{m,k} \left(\frac{m+1}{b} - \frac{m}{b}\right) \\
&= \frac{1}{b} \sum_{m=0}^{b-1} \overline{\Wfun_{m,j}} \Wfun_{m,k}.
\end{align*}
Since $\Wfun$ is a unitary matrix, its columns are orthonormal, meaning $\sum_{m=0}^{b-1} \overline{\Wfun_{m,j}} \Wfun_{m,k} = \delta_{jk}$.
Therefore, for $n=1$, 
\begin{equation}
\braket{f_j^{(b)} | f_k^{(b)}} = \frac{\delta_{jk}}{b}.
\label{eq:base_case_f_ortho}
\end{equation}
This establishes the base case for $f_j^{(b^n)}(x)$.

\noindent \textit{Inductive Step:}
Next we assume that the orthogonality condition holds for some integer $n-1 \geq 1$, i.e., $\braket{f_{\tilde{j}}^{(b^{n-1})} | f_{\tilde{k}}^{(b^{n-1})}} = \frac{\delta_{\tilde{j}\tilde{k}}}{b^{n-1}}$ for all $0 \leq \tilde{j}, \tilde{k} < b^{n-1}$.
We now show it holds for $n$. We use the recursive definition for $f_j^{(b^n)}(x)$ from Equation~\eqref{eq:gtt_recursive_general_b_main}.
Recall that $j = b^{n-1} j_n + \tilde{j}$ and $k = b^{n-1} k_n + \tilde{k}$, where $j_n = \lfloor j / b^{n-1} \rfloor$, $k_n = \lfloor k / b^{n-1} \rfloor$, $\tilde{j} = j \pmod{b^{n-1}}$, and $\tilde{k} = k \pmod{b^{n-1}}$.

Consider the inner product $\braket{f_j^{(b^n)} | f_k^{(b^n)}}$:
\begin{align*}
\braket{f_j^{(b^n)} | f_k^{(b^n)}} &= \int_{0}^{1} \overline{f_j^{(b^n)}(x)} f_k^{(b^n)}(x) \, dx \\
&= \sum_{m=0}^{b-1} \int_{\frac{m}{b}}^{\frac{m+1}{b}} \overline{f_j^{(b^n)}(x)} f_k^{(b^n)}(x) \, dx .\\
\intertext{Substituting the recursive definition from Equation~\eqref{eq:gtt_recursive_general_b_main}, where $\lfloor b x \rfloor = m$ for $x \in [\frac{m}{b}, \frac{m+1}{b})$,}
\braket{f_j^{(b^n)} | f_k^{(b^n)}} &= \sum_{m=0}^{b-1} \int_{\frac{m}{b}}^{\frac{m+1}{b}} \overline{\Wfun_{m,j_n} f_{\tilde{j}}^{(b^{n-1})}(\{b x\})} \Wfun_{m,k_n} f_{\tilde{k}}^{(b^{n-1})}(\{b x\}) \, dx \\
&= \sum_{m=0}^{b-1} \overline{\Wfun_{m,j_n}} \Wfun_{m,k_n} \int_{\frac{m}{b}}^{\frac{m+1}{b}} \overline{f_{\tilde{j}}^{(b^{n-1})}(\{b x\})} f_{\tilde{k}}^{(b^{n-1})}(\{b x\}) \, dx .\\
\intertext{Let $y = bx - m$. Then $dy = b \, dx$, so $dx = \frac{1}{b} \, dy$. When $x = \frac{m}{b}$, $y = 0$. When $x = \frac{m+1}{b}$, $y = 1$. And $\{b x\} = y$.}
\braket{f_j^{(b^n)} | f_k^{(b^n)}} &= \sum_{m=0}^{b-1} \overline{\Wfun_{m,j_n}} \Wfun_{m,k_n} \frac{1}{b} \int_{0}^{1} \overline{f_{\tilde{j}}^{(b^{n-1})}(y)} f_{\tilde{k}}^{(b^{n-1})}(y) \, dy. \\
\intertext{By the induction hypothesis, $\int_{0}^{1} \overline{f_{\tilde{j}}^{(b^{n-1})}(y)} f_{\tilde{k}}^{(b^{n-1})}(y) \, dy = \frac{\delta_{\tilde{j} \tilde{k}}}{b^{n-1}}$. Therefore,}
\braket{f_j^{(b^n)} | f_k^{(b^n)}} &= \frac{1}{b} \sum_{m=0}^{b-1} \overline{\Wfun_{m,j_n}} \Wfun_{m,k_n} \frac{\delta_{\tilde{j} \tilde{k}}}{b^{n-1}} \\
&= \frac{\delta_{\tilde{j} \tilde{k}}}{b^n} \sum_{m=0}^{b-1} \overline{\Wfun_{m,j_n}} \Wfun_{m,k_n} \\
& = \frac{\delta_{\tilde{j} \tilde{k}} \delta_{j_n k_n}}{b^n}.
\end{align*}
The last equality follows from the fact that $\Wfun$ is unitary, and its columns are orthonormal, meaning $\sum_{m=0}^{b-1} \overline{\Wfun_{m,j_n}} \Wfun_{m,k_n} = \delta_{j_n k_n}$.
Finally, note that $\delta_{\tilde{j} \tilde{k}} \delta_{j_n k_n} = \delta_{jk}$. This is because $j=k$ if and only if both their most significant base-$b$ digits are equal ($j_n=k_n$) and their remaining less significant digits are equal ($\tilde{j}=\tilde{k}$). If either condition fails, $j \neq k$, and thus $\delta_{jk}=0$.
Therefore, we have established that:
\begin{equation}
\braket{f_j^{(b^n)} | f_k^{(b^n)}} = \frac{\delta_{jk}}{b^n}.
\label{eq:f_orthogonality_result}
\end{equation}
This completes the proof.
\end{proof}

\subsection{GTT-Fourier series}
\label{ssec:gtt_fourier_series}

Building upon the orthonormality established in Theorem~\ref{thm_ortho}, the normalized Generalized Tensor Transform basis functions provide a powerful framework for approximating continuous or square-integrable functions. Any function \( g(x) \in L^2[0,1) \) can be accurately represented by a Fourier series expansion in terms of these orthonormal GTT basis functions \( \phi_m^{(b^n)}(x) \). This GTT-Fourier Series approximation is given by:
\[
    g(x) \approx \sum_{m=0}^{b^n - 1} C_m \phi_m^{(b^n)}(x),
\]
where the Fourier coefficients \( C_m \) are computed using the standard $L^2$ inner product for orthonormal bases:
\[
    C_m = \braket{\phi_m^{(b^n)} | g } = \int_0^1 \overline{\phi_m^{(b^n)}(x)} g(x) \, dx.
\]
This expansion provides a versatile method for representing and analyzing functions by decomposing them into components along the GTT basis functions, leveraging their customizable nature derived from the choice of $\Wfun$.

\subsubsection{Example: GTT-Fourier series approximation ($b=2$, $n=1$)}
\label{sssec:example_gtt_fourier}


We illustrate the GTT-Fourier Series for a simple case with base $b=2$ and $n=1$, meaning we use $N=2^1=2$ basis functions. To highlight the generalization beyond well-known transforms, consider the following unitary base matrix $\Wfun$:
\[
\Wfun = \frac{1}{\sqrt{2}} \begin{bmatrix} 1 & i \\ 1 & -i \end{bmatrix}.
\]
We first determine the unnormalized GTT basis functions $f_j^{(2)}(x)$ for $N=2$:
\begin{align*}
f_0^{(2)}(x) &= \begin{cases} \Wfun_{0,0} = 1/\sqrt{2}, & 0 \leq x < 1/2 \\ \Wfun_{1,0} = 1/\sqrt{2}, & 1/2 \leq x < 1 \end{cases} = \frac{1}{\sqrt{2}} \quad \text{for } x \in [0,1), \\
f_1^{(2)}(x) &= \begin{cases} \Wfun_{0,1} = i/\sqrt{2}, & 0 \leq x < 1/2 \\ \Wfun_{1,1} = -i/\sqrt{2}, & 1/2 \leq x < 1 \end{cases}.
\end{align*}
According to Equation~\eqref{eq:normalized_basis_phi}, the orthonormal GTT basis functions for $n=1$ ($b^{n/2} = 2^{1/2} = \sqrt{2}$) are obtained by scaling $f_j^{(2)}(x)$:
\begin{align*}
\phi_0^{(2)}(x) &= \sqrt{2} f_0^{(2)}(x) = \sqrt{2} \cdot \frac{1}{\sqrt{2}} = 1 \quad \text{for } x \in [0,1), \\
\phi_1^{(2)}(x) &= \sqrt{2} f_1^{(2)}(x) = \sqrt{2} \begin{cases} i/\sqrt{2}, & 0 \leq x < 1/2 \\ -i/\sqrt{2}, & 1/2 \leq x < 1 \end{cases} = \begin{cases} i, & 0 \leq x < 1/2 \\ -i, & 1/2 \leq x < 1 \end{cases}.
\end{align*}
Now, we approximate a simple step function $g(x)$ defined as:
\[
g(x) = \begin{cases} 1, & 0 \leq x < 1/2 \\ 0, & 1/2 \leq x < 1 \end{cases}.
\]
We compute the Fourier coefficients $C_0$ and $C_1$ as follows:
\begin{align*}
C_0 &= \int_0^1 \overline{\phi_0^{(2)}(x)} g(x) \, dx = \int_0^{1/2} (1) (1) \, dx + \int_{1/2}^{1} (1) (0) \, dx 
= 1 \cdot \frac{1}{2} = \frac{1}{2} \\
C_1 &= \int_0^1 \overline{\phi_1^{(2)}(x)} g(x) \, dx 
= \int_0^{1/2} \overline{(i)} (1) \, dx + \int_{1/2}^{1} \overline{(-i)} (0) \, dx,
= \int_0^{1/2} (-i) (1) \, dx = (-i) \cdot \frac{1}{2} = -\frac{i}{2}.
\end{align*}
Now, we construct the GTT-Fourier Series approximation $g(x) \approx C_0 \phi_0^{(2)}(x) + C_1 \phi_1^{(2)}(x)$.
Substituting the coefficients and basis functions:
\begin{align*}
g(x) &\approx \frac{1}{2} (1) + \left(-\frac{i}{2}\right) \begin{cases} i, & 0 \leq x < 1/2 \\ -i, & 1/2 \leq x < 1 \end{cases} 
\quad \text{ or } \quad  g(x) \approx  \begin{cases} 1, & 0 \leq x < 1/2 \\ 0, & 1/2 \leq x < 1 \end{cases}.
\end{align*}
This example demonstrates that even with a more general unitary base matrix $\Wfun$ involving complex values, the GTT-Fourier Series efficiently and perfectly reconstructs a piecewise constant function, underscoring the flexibility and power of this generalized transform framework.

\subsection{Discretization and the GTT matrix}
\label{ssec:discretization}

To facilitate practical computations and connect the continuous GTT basis functions to a discrete transform, we introduce the concept of discretization. Let $N = b^n$. Assume that the interval $[0, 1]$ is divided uniformly into $N$ subintervals of length $1/N$.

The Generalized Tensor Transform (GTT) matrix $\Gfun_N$ is constructed by sampling the unnormalized GTT basis functions $f_q^{(N)}(x)$ defined previously. Recall that $f_q^{(N)}(x)$ is piecewise constant over these $N$ subintervals. The $(p, q)$-th element of $\Gfun_N$ is obtained by sampling $f_q^{(N)}$ at any point within the $p$-th sub-interval, typically at its midpoint $x_p = \frac{2p+1}{2N}$, where $p$ and $q$ vary from $0$ to $N-1$.
In mathematical terms, the $(p, q)$-th element of $\Gfun_N$ is given by:
\begin{equation} \label{eq:defGNOne}
\Gfun_N(p, q) = f_q^{(N)} \left( x_p \right), \quad \text{where } x_p = \frac{2p+1}{2N} \text{ and } p, q = 0, 1, \dots, N-1.
\end{equation}

\begin{thm} \label{thm_unitary_G_N}
    The matrix $\Gfun_N$ defined by sampling the unnormalized GTT basis functions is unitary.
\end{thm}
\begin{proof}
Consider the inner product of the $j$-th and the $k$-th columns of $\Gfun_N$. This is given by the sum:
\[
\sum_{p=0}^{N-1} \overline{\Gfun_N(p, j)} \Gfun_N(p, k) = \sum_{p=0}^{N-1} \overline{f_j^{(N)} \left( x_p \right)} f_k^{(N)} \left( x_p \right).
\]
Since $f_j^{(N)}(x)$ is constant over each sub-interval $\left[\frac{p}{N}, \frac{p+1}{N}\right)$, and $x_p$ is a point within this interval, we can relate the sum to the $L^2$ inner product:
\[
\overline{f_j^{(N)} \left( x_p \right)} f_k^{(N)} \left( x_p \right) = N \int_{\frac{p}{N}}^{\frac{p+1}{N}} \overline{f_j^{(N)} \left( x \right)} f_k^{(N)} \left( x \right) \, dx.
\]
Summing over all sub-intervals:
\begin{align*}
\sum_{p=0}^{N-1} \overline{\Gfun_N(p, j)} \Gfun_N(p, k) &= \sum_{p=0}^{N-1} N \int_{\frac{p}{N}}^{\frac{p+1}{N}} \overline{f_j^{(N)} \left( x \right)} f_k^{(N)} \left( x \right) \, dx \\
&= N \int_{0}^{1} \overline{f_j^{(N)} \left( x \right)} f_k^{(N)} \left( x \right) \, dx \\
&= N \braket{f_j^{(N)} | f_k^{(N)}}.
\end{align*}
From Equation~\eqref{eq:f_orthogonality_result} in the proof of Theorem~\ref{thm_ortho} (which establishes that $\braket{f_j^{(N)} | f_k^{(N)}} = \frac{\delta_{jk}}{N}$), the final expression becomes:
\[
N \braket{f_j^{(N)} | f_k^{(N)}} = N \left( \frac{\delta_{jk}}{N} \right) = \delta_{jk}.
\]
Thus, the columns of $\Gfun_N$ are orthonormal, which implies that $\Gfun_N$ is a unitary matrix.
\end{proof}

It can be shown that the GTT matrix $\Gfun_N$ as defined in \meqref{eq:defGNOne} is precisely the $n$-th Kronecker product (tensor product) of the base matrix $\Wfun$. This property is fundamental to the GTT's structure and efficient computation. We have the following definition.

\begin{defn} \label{Def_GTT_tensor}
The GTT matrix $\Gfun_N$ is defined as the $n$-th Kronecker product (tensor product) of the base matrix $\Wfun$, i.e., $\Gfun_N = \Wfun^{\otimes n}$. 
\end{defn}

\begin{remark} \label{Remark_tunable_parameter}
\textit{Degrees of Freedom Associated with Tunable Parameters}:
It is well known that the unitary group \( U(b) \) forms a real Lie group of dimension \( b^2 \). A distinctive feature of the GTT framework is its flexibility in adapting the basis functions via tunable parameters. In classical contexts where the global phase carries physical meaning, such as in certain signal processing applications, a \( b \times b \) unitary matrix \( W \) admits \( b^2 \) real degrees of freedom. By contrast, in quantum computing, where global phase has no physical consequence, the effective number of parameters reduces to \( b^2 - 1 \), reflecting the fact that transformations differing only by a global phase are physically indistinguishable.
\end{remark}

\begin{remark} \label{Remark_general_tensor}
While the Generalized Tensor Transform (GTT) is defined as the $n$-fold tensor product of an identical $b \times b$ unitary matrix $\Wfun$ ($\Gfun= \Wfun^{\otimes n}$), an even more encompassing generalization can be considered. This involves allowing each constituent matrix in the tensor product to be distinct, potentially even of different dimensions. Specifically, one could define a GTT as:
$$ \Gfun_N = \bigotimes_{k=0}^{n-1} \Wfun_k $$
where $\Wfun_k$ is a $b_k \times b_k$ unitary matrix for each $k \in \{0, \dots, n-1\}$. In this advanced form, the total dimension of the transform becomes $N = \prod_{k=0}^{n-1} b_k$. This allows for an even greater degree of flexibility in constructing tailored transform bases, potentially unlocking applications in quantum information processing where heterogeneous local dimensions or specific non-uniform transformations are beneficial. Exploring the implications and applications of such an ultra-generalized GTT in quantum algorithms represents a promising direction for future research.
\end{remark}

It follows from Definition~\ref{Def_GTT_tensor} that
\begin{align}
    \Gfun_N(p, q) = \prod_{j=0}^{n-1} \Wfun_{p_j, q_j},
\end{align}
where $p = \sum_{j=0}^{n-1} p_j b^j$ and $q = \sum_{j=0}^{n-1} q_j b^j$ are the base-$b$ expansions of $p$ and $q$, respectively, with digits $p_j, q_j \in \{0, 1, \dots, b-1\}$.
Equivalently, the above can be written using a counting convention for digit pairs. Let $\alpha(i, j)$ denote the number of times the digit pair $(p_k, q_k)$ equals $(i, j)$ across the $n$ positions ($k=0, \dots, n-1$) in the base-$b$ expansions of $p$ and $q$. This can be formally defined as:
\begin{align}
    \Gfun_N(p, q) = \prod_{i=0}^{b - 1} \prod_{j=0}^{b - 1} \Wfun_{i, j}^{\alpha(i, j)},
    \label{eq:G_N_alternative_product_form}
\end{align}
where
\begin{align}
    \alpha(i, j) = \sum_{k = 0}^{n - 1} \chi_{i,j}(p_k, q_k),
\end{align}
and $\chi_{i,j}(p_k, q_k)$ is an indicator function defined as:
\begin{align}
    \chi_{i,j}(p_k, q_k) =
    \begin{cases}
        1, & \text{if } p_k = i \text{ and } q_k = j, \\
        0, & \text{otherwise}.
    \end{cases}
\end{align}
This form explicitly shows how the product of matrix elements is accumulated based on the specific digits in the base-$b$ representation of $p$ and $q$.

\begin{defn}[Action of the Generalized Tensor Transform]
Let $\Wfun$ be a $b \times b$ unitary matrix. Let $N=b^n$. The action of the Generalized Tensor Transform on a computational basis state $\ket{q}$, where $q=0, 1, \ldots, N-1$, is given by:
\begin{align}
    \Gfun_N \ket{q} = \sum_{p=0}^{N-1} \Gfun_N(p,q) \ket{p} = \sum_{p=0}^{N-1} \left( \prod_{i=0}^{b - 1} \prod_{j=0}^{b - 1} \Wfun_{i, j}^{\sum_{k = 0}^{n - 1} \chi_{i,j}(p_k, q_k)} \right) \ket{p}.
\end{align}
\end{defn}

We note that, for the specific base matrix $\Wfun = \frac{1}{\sqrt{2}} \begin{bmatrix} 1 & 1 \\ 1 & -1 \end{bmatrix}$ (the Hadamard matrix), the general product term $\prod_{i=0}^{b - 1} \prod_{j=0}^{b - 1} \Wfun_{i, j}^{\alpha(i, j)}$ reduces to the well-known Hadamard transform element. For $b=2$, the matrix elements are $\Wfun_{00} = 1/\sqrt{2}$, $\Wfun_{01} = 1/\sqrt{2}$, $\Wfun_{10} = 1/\sqrt{2}$, and $\Wfun_{11} = -1/\sqrt{2}$.
Substituting these into the product form:
\begin{align*}
\prod_{i=0}^{1} \prod_{j=0}^{1} \Wfun_{i, j}^{\alpha(i, j)} &= \Wfun_{0,0}^{\alpha(0,0)} \Wfun_{0,1}^{\alpha(0,1)} \Wfun_{1,0}^{\alpha(1,0)} \Wfun_{1,1}^{\alpha(1,1)} \\
&= \left(\frac{1}{\sqrt{2}}\right)^{\alpha(0,0)} \left(\frac{1}{\sqrt{2}}\right)^{\alpha(0,1)} \left(\frac{1}{\sqrt{2}}\right)^{\alpha(1,0)} \left(-\frac{1}{\sqrt{2}}\right)^{\alpha(1,1)} \\
&= \left(\frac{1}{\sqrt{2}}\right)^{\alpha(0,0)+\alpha(0,1)+\alpha(1,0)+\alpha(1,1)} (-1)^{\alpha(1,1)}.
\end{align*}
The sum of all $\alpha(i,j)$ for fixed $p$ and $q$ is simply the total number of digits, $n$ (since for each $k$, exactly one $\chi_{i,j}(p_k,q_k)$ will be 1). So, $\sum_{i,j} \alpha(i,j) = n$.
The term $\left(\frac{1}{\sqrt{2}}\right)^n = \frac{1}{\sqrt{2^n}} = \frac{1}{\sqrt{N}}$.

Furthermore, $\alpha(1,1) = \sum_{k=0}^{n-1} \chi_{1,1}(p_k, q_k)$ where $\chi_{1,1}(p_k, q_k)=1$ only if $p_k=1$ and $q_k=1$. For binary digits, $p_k q_k=1$ if and only if $p_k=1$ and $q_k=1$. Thus, $\alpha(1,1) = \sum_{k=0}^{n-1} p_k q_k$, which is precisely the bit-wise dot product $p \cdot q$.
Therefore, for the Hadamard base matrix, the action of the GTT reduces to:
\begin{align*}
    \Gfun_N \ket{q} = \sum_{p=0}^{N-1} \frac{1}{\sqrt{N}} (-1)^{p \cdot q} \ket{p},
\end{align*}
which is the standard unitary Hadamard transform.

\section{Fast algorithm for generalized tensor transforms}
\label{sec:algorithm}

We note that, given an input vector $\mathbf{x} \in \mathbb{C}^N$, computing the matrix-vector product $\mathbf{y} = \Gfun_N \mathbf{x}$ by explicitly constructing the full $\Gfun_N$ matrix is computationally expensive, requiring $\mathcal{O}(N^2)$ operations. However, the inherent separability of the Kronecker product enables a significantly more efficient recursive algorithm, achieving a computational complexity of $\mathcal{O}(N \log_b N)$. This approach is analogous to the well-established butterfly structure prevalent in Fast Fourier Transform algorithms.

The algorithm proceeds by iteratively applying local unitary transformations and recursing across successive tensor factors. The steps for computing $\mathbf{y} = \Gfun_N \mathbf{x}$ are outlined as follows:

\begin{enumerate}
    \item \textit{Input Vector Reshaping:} Given an input vector $\mathbf{x} \in \mathbb{C}^N$, it is interpreted as an element within a tensor product space, specifically $\mathbb{C}^{N} \cong \mathbb{C}^{N/b} \otimes \mathbb{C}^{b}$. This is realized by reshaping $\mathbf{x}$ into a matrix $\mathbf{X} \in \mathbb{C}^{(N/b) \times b}$. The element $\mathbf{X}_{p,q}$ corresponds to $\mathbf{x}_{pb+q}$, where $p$ serves as the `block' index (representing the first $n-1$ base-$b$ digits) and $q$ denotes the `intra-block' position (representing the last base-$b$ digit).

    \item \textit{Initial Stage Transformation (Last Tensor Factor):} The unitary base matrix $\Wfun \in \mathbb{C}^{b \times b}$ is applied along the final tensor index, which corresponds to a column-wise operation on $\mathbf{X}$. This effectively transforms the least significant (last) base-$b$ digit of the original index. This operation yields an intermediate matrix $\mathbf{Y} \in \mathbb{C}^{(N/b) \times b}$, defined as:
    \[
    \mathbf{Y} = \mathbf{X} \Wfun^\top.
    \]
    The presence of $\Wfun^\top$ in this formulation arises from performing a right-multiplication on $\mathbf{X}$, where $\Wfun$ is applied to each row of $\mathbf{X}$ (corresponding to a block of elements in $\mathbf{x}$ whose last index varies).

    \item \textit{Recursive Transformation (Remaining Tensor Factors):} Subsequently, for each fixed column index $j \in \{0, 1, \dots, b-1\}$, the corresponding column vector $\mathbf{Y}^{(j)} \in \mathbb{C}^{N/b}$ (i.e., the $j$-th column of $\mathbf{Y}$) is isolated. The generalized transform $\Gfun_{N/b}$ is then recursively applied to this column vector, resulting in a new matrix $\mathbf{Z} \in \mathbb{C}^{(N/b) \times b}$ whose columns are determined by:
    \[
    \mathbf{Z}^{(j)} = \text{FastGTT}(\mathbf{Y}^{(j)}, \Wfun, n-1).
    \]
    This step systematically transforms the initial $n-1$ base-$b$ digits for each specific value of the last digit.

    \item \textit{Output Vector Reassembly:} Finally, the matrix $\mathbf{Z}$ is reassembled into the output vector $\mathbf{y} \in \mathbb{C}^N$ by concatenating its rows in sequential order:
    \[
    \mathbf{y} =
    \begin{bmatrix}
    \mathbf{Z}_{0,:}^\top \\
    \mathbf{Z}_{1,:}^\top \\
    \vdots \\
    \mathbf{Z}_{(N/b)-1,:}^\top
    \end{bmatrix}
    \in \mathbb{C}^N.
    \]
    This re-concatenation effectively reorganizes the transformed tensor elements back into a flat vector representation.
\end{enumerate}

This recursive process terminates when $n = 1$, at which point the base case for $\Gfun_b$ simplifies to a direct matrix-vector multiplication by $\Wfun$, i.e., $\Gfun_b \mathbf{x} = \Wfun \mathbf{x}$. This formulation yields an efficient and structured algorithm for computing large-scale unitary transforms $\Gfun_N = \Wfun^{\otimes n}$, requiring only $\mathcal{O}(N \log_b N)$ operations and circumventing the computationally intensive explicit construction of the full Kronecker product matrix.

\subsection{Algorithm: Fast GTT}
\label{ssec:fast_GTT}

The following algorithm outlines the recursive procedure for the Fast Generalized Tensor Transform (FastGTT), using the approach  described above. 

\begin{algorithm}[H]
\RestyleAlgo{algoruled}
\SetAlgoLined
\KwIn{A vector $\mathbf{x} \in \mathbb{C}^N$, where $N = b^n$, and a unitary base matrix $\Wfun \in \mathbb{C}^{b \times b}$.}
\KwOut{The transformed vector $\mathbf{y} = \Gfun_N \mathbf{x} \in \mathbb{C}^N$, where $\Gfun_N = \Wfun^{\otimes n}$.}
\SetKwFunction{FTransform}{FastGTT}
\SetKwProg{Fn}{Function}{:}{}
\Fn{\FTransform{\(\mathbf{x}, \Wfun, n\)}}{
   \tcp{Base case: direct matrix multiplication}
    \If{\( n = 1 \)}{
        \Return \( \Wfun \cdot \mathbf{x} \)\; 
    }

    Let \( N_{curr} = b^n \)\;
    Let \( N_{prev} = N_{curr}/b \)\;

    \tcp{Step 1: Reshape $\mathbf{x}$ into a matrix $\mathbf{X}$}
    Define \( \mathbf{X} \in \mathbb{C}^{N_{prev} \times b} \) such that $\mathbf{X}_{p,q} = \mathbf{x}_{p \cdot b + q}$\;

    \tcp{Step 2: Apply $\Wfun$ to the ``column'' dimension of $\mathbf{X}$}
    Compute \( \mathbf{Y} = \mathbf{X} \cdot \Wfun^\top \in \mathbb{C}^{N_{prev} \times b} \)\;

    \tcp{Step 3: Recursively transform each column of $\mathbf{Y}$ and form $\mathbf{Z}$}
    Initialize \( \mathbf{Z} \in \mathbb{C}^{N_{prev} \times b} \) with zeros\;
    \For{column index \( j \) from \( 0 \) to \( b-1 \)}{
        Let \( \mathbf{y}^{(j)} = \mathbf{Y}_{:,j} \) (the $j$-th column of $\mathbf{Y}$)\;
        Set \( \mathbf{Z}_{:,j} = \text{\FTransform}(\mathbf{y}^{(j)}, \Wfun, n - 1) \)\;
    }

    \tcp{Step 4: Flatten the matrix $\mathbf{Z}$ into the output vector $\mathbf{y}$}
    Define \( \mathbf{y} \in \mathbb{C}^{N_{curr}} \) by concatenating the rows of \( \mathbf{Z} \) in order\;

    \Return \( \mathbf{y} \)\;
}
\caption{Fast Generalized Tensor Transform 
}
\label{alg_fast_GTT}
\end{algorithm}

\subsubsection*{Computational Complexity of Algorithm~\ref{alg_fast_GTT}.}  

At each recursive level, the algorithm reshapes the input vector of size \( N \) into a matrix of shape \( (N/b) \times b \), multiplies this matrix on the right by \( \Wfun^\top \), and then applies the same algorithm recursively to each of the \( b \) resulting columns of length \( N/b \). Let \( T(N) \) denote the total number of arithmetic operations required for an input of size \( N \). Then the recurrence relation satisfied by \( T(N) \) is
\[
T(N) = b \cdot T\left(\frac{N}{b}\right) + O(Nb),
\]
where the first term accounts for the \( b \) recursive calls on subproblems of size \( N/b \), and the second term accounts for the matrix multiplication and reshaping costs at the current level.

Unrolling this recurrence yields \( \log_b N = n \) recursive levels, each incurring a cost of \( O(Nb) \). Thus, the total complexity is
\[
T(N) = O(Nb \log_b N) = O(n N b).
\]
This is significantly more efficient than the naive approach, which would require \( O(N^2) \) operations, and matches the structure of other fast tensor product transforms, such as the Fast Fourier Transform when \( \Wfun \) is the DFT matrix of order \( b \).

\section{Applications of the generalized tensor transform}
\label{sec:gtt_applications}

In this section, we consider three distinct applications of Generalized Tensor Transform (GTT). The first application is quantum state compression and efficient transmission, which is an important problem in quantum information theory with potential applications in quantum communications and distributed quantum computing. It remains an active area of research, and many general and specific methods have been explored in the literature \cite{Bai2020QuantumCompression, Rozema2014QuantumData, huang2025quantum, Cruzeiro2025CompressionEntanglement, vanLoo2017QuantumState}.
We note that the GTT-based quantum state compression protocol presented in \mref{ssec:state_compression} offers a distinctive spectral approach compared to existing methods. Our approach uniquely leverages the tunability of GTT basis functions to achieve adaptive sparsity. Unlike approaches that rely on fixed bases or inherent state structures for compression, the GTT framework, with its arbitrary unitary base matrix $W$, allows for a fine-grained adaptation of the transformation to induce optimal sparsity in the transformed domain. This capability is important, as demonstrated by our numerical results, which show considerably higher reconstruction fidelities with fewer retained components for states sparse in the GTT basis compared to compression utilizing fixed transforms like the Hadamard or Quantum Fourier Transform. This fundamental distinction positions our GTT-based protocol as a versatile tool for enhancing efficiency in quantum information processing and communication.

Next, we consider the applicaiton of GTT in efficient encoding of the input function on a quantum computer. Efficient function encoding is a key step, and also a bottleneck, in realizing the potential of quantum computing in the solution of real-world problems, and remains an active area of research \cite{schuld2021effect, shin2023exponential, gonzalez2024efficient}. From a quantum differential equation solver to a quantum machine learning model, practically almost all quantum algorithms would need the classical input to be encoded in a way that is amenable to being processed on a quantum computer. Here we present an efficient way to encode a function using the adaptability of the GTT. We demonstrate in \mref{sec:fun_encoding} that the GTT-based approach is superior to the WHT-based approach for the examples considered.

Finally as our third application, in \mref{ssec:filtering} we provide an algorithm (ref.~Algorithm \ref{alg_gtt_natural_filtering}) for low-pass and high-pass filtering of digital signals using GTT. This method is applicable for both classical and quantum signal processing, contributing to the growing field of quantum signal processing \cite{shukla2023quantumdsp}. A computational example demonstrates the effectiveness of our proposed GTT-based approach for filtering applications.
We note that, the GTT-based filtering approach (via the fast GTT algorithm, Algorithm \ref{alg_fast_GTT}, presented in \mref{ssec:fast_GTT}) also offers a novel approach for classical digital filtering  applications \cite{proakis2007digital}. 

The proposed Generalized Tensor Transform (GTT) significantly reduces computational costs: classically, its fast algorithm achieves an $O(N \log_b N)$ complexity, mirroring Fast Fourier Transform (FFT) and Fast Walsh-Hadamard Transform (FWHT) efficiency, a substantial improvement over the naive $O(N^2)$ classical implementation. For quantum applications, the GTT algorithm, where $N=b^n$, attains both gate complexity and circuit depth of $O(\log_b N)$, representing a quadratic improvement over the Quantum Fourier Transform (QFT)'s $O((\log_b N)^2)$ complexity and an exponential advantage over classical $O(N \log_b N)$ methods.

\subsection{GTT-based quantum state compression and efficient transmission}
\label{ssec:state_compression}

An important application of the Generalized Tensor Transform (GTT) is enabling the compression and efficient transmission of quantum states that exhibit sparsity in a GTT-transformed basis. This is especially relevant for quantum communication, where minimizing the resources for quantum state transfer can considerably reduce exposure to decoherence, improve quantum bandwidth utilization, and enhance reliability.

Consider an $n$-qudit quantum state, $|\psi\rangle = \sum_{x=0}^{N - 1} a_x |x\rangle$ (with $N=b^n$), which may not be sparse in the computational basis. However, suppose there is a specific GTT, $\Gfun_N = \Wfun^{\otimes n}$, that transforms $|\psi\rangle$ into a sparse state, $|\hat{\psi}\rangle = \Gfun_N |\psi\rangle = \sum_{y=0}^{N - 1} \hat{a}_y |y\rangle$. Here, only $k \prec\prec N$ of the coefficients $\hat{a}_y$ are significantly non-zero. This property is analogous to how classical signals like images can be sparse in a Fourier or wavelet domain. The GTT's inherent tunable parameters, which allow for adaptive basis construction, provide a significant benefit in achieving sparse representations compared to fixed transforms like the Hadamard or Fourier transform. The benefit of this adaptive sparsity will be further illustrated in  \mref{sec:fun_encoding}.

This sparsity allows for powerful compression strategies. We discuss three main approaches: a purely classical protocol, a hybrid quantum-classical protocol and a fully quantum protocol that offers theoretical advantages by preserving coherence throughout the compression process.

\subsubsection{Purely classical protocol for state transmission}
\label{ssec:classical_protocol}

Consider a purely classical protocol for transmitting information about an $N$-dimensional state, assuming it is sparse in a GTT basis.

\begin{enumerate}
    \item \textit{Classical GTT Calculation}: Alice computes the Generalized Tensor Transform on the $N$-dimensional classical vector. This operation inherently calculates all $N$ coefficients and typically costs $O(N \log_b N)$ operations.
    \item \textit{Classical Coefficient Selection}: From the $N$ coefficients, Alice identifies the $k$ dominant ones. This typically involves sorting coefficients with the associated cost ($O(N \log_b N)$). 
    \item \textit{Classical Data Transmission}: Alice transmits the $k$ indices and their corresponding complex amplitude values to Bob. 
    \item \textit{Classical Reconstruction}: Bob receives the classical data. He then reconstructs the $N$-dimensional vector from these $k$ coefficients and computes the inverse GTT classically, which also costs $O(N \log_b N)$ operations.
\end{enumerate}
The total classical complexity for transmitting a sparse state by first computing its full classical GTT representation is therefore dominated by $O(N \log_b N)$.

\subsubsection{Hybrid quantum-classical compression protocol}
\label{sec:hybrid_compression}

This protocol uses the GTT to reveal sparsity, but then employs classical processing to identify and extract the dominant components before re-encoding them into a smaller quantum state for transmission.

The process involves the following steps:

\begin{enumerate}
    \item \textit{Transform (Quantum)}: Alice applies the Generalized Tensor Transform ($\Gfun_N = \Wfun^{\otimes n}$) to her initial $n$-qudit state $|\psi\rangle$, obtaining the transformed state $|\hat{\psi}\rangle = \Gfun_N |\psi\rangle$. This step is purely quantum.

    \item \textit{Analyze and Truncate (Hybrid/Lossy Step)}: Alice then performs a measurement or amplitude estimation on $|\hat{\psi}\rangle$ to obtain classical information about its coefficients $\hat{a}_y$. Based on this classical data, she identifies the set $S_k \subset \{0, \ldots, N - 1\}$ of indices corresponding to the $k$ most dominant coefficients (e.g., largest in magnitude). Information from discarded coefficients (those not in $S_k$) is lost here, making this a lossy compression step. This classical analysis informs the subsequent quantum state preparation.

    \item \textit{Encode (Quantum State Preparation)}: Using the $k$ selected classical amplitudes and their corresponding indices from $S_k$, Alice prepares a new, normalized, $k$-dimensional quantum state. This typically involves preparing a state of $\lceil \log_b k \rceil$ qudits. For instance, if $S_k = \{y_0, y_2, \ldots, y_{k-1}\}$, she prepares the state $\sum_{j=1}^k \hat{a}_{y_j} \ket{j}$. This is the step where classical information is converted back into a quantum state for efficient transmission.

    \item \textit{Transmit}: Alice sends the compressed $\lceil \log_b k \rceil$-qudit state to Bob via a quantum channel. Alice must also classically transmit the set of chosen indices $S_k$ (or equivalent mapping information) to Bob as side information. This classical information is essential for Bob to correctly interpret and reconstruct the received state.

    \item \textit{Decode and Reconstruct}: Upon receipt, Bob first uses the classical side information ($S_k$) to decode the compressed state. He maps the received $\lceil \log_b k \rceil$-qudit state back to its corresponding $k$-dimensional subspace representation in the original $n$-qudit space, spanning $\{|y\rangle : y \in S_k\}$. This gives him an approximate transformed state $|\hat{\psi}_{\text{approx}}\rangle$. Bob then applies the inverse GTT, $\Gfun_N^\dagger$, to this approximate state to reconstruct an approximation of the original state: $|\psi_{\text{approx}}\rangle = \Gfun_N^\dagger |\hat{\psi}_{\text{approx}}\rangle$.
\end{enumerate}

\subsubsection{Fully quantum compression protocol}
Next, we describe a fully quantum algorithm (ref.~Algorithm \ref{alg:quantum_compression}) for compressing states that are sparse in a GTT basis, assuming the indices of the dominant sparse components ($S_k$) are known via classical heuristics.  In what follows,  for simplicity of explanation we assume that $b=2$.
The protocol uses three quantum registers:
\begin{itemize}
    \item {$R_{\text{orig}}$}: The original $n$-qubit register holding the state to be compressed.
    \item {$q_a$}: A single ancilla qubit used as a probabilistic flag.
    \item {$R_{\text{comp}}$}: The $\lceil \log_b k \rceil$-qubit register for the compressed state, which will be transmitted.
\end{itemize}

\begin{algorithm}
\caption{Quantum sparse state compression with known sparse indices}
\label{alg:quantum_compression}

\KwInput{$n$-qubit quantum state $\ket{\psi} = \sum_{x=0}^{2^n-1} a_x \ket{x}$ (such that its GTT- transformed space is sparse).}
\KwData{Known index set $S_k = \{y_0, y_1, \ldots, y_{k-1}\}$ of dominant GTT coefficients.}
\KwOutput{Compressed $\lceil \log_2 k \rceil$-qubit state and classical $S_k$ information.}

\BlankLine
\tcc{\textbf{0. Initial System State (Alice's side):}}

$\ket{\Psi_0} = \ket{\psi}_{R_{\text{orig}}} \otimes \ket{0}_{q_a} \otimes \ket{0}_{R_{\text{comp}}}$, where $\ket{\psi}_{R_{\text{orig}}} = \sum_{x=0}^{N-1} a_x \ket{x}$.

\BlankLine
\tcc{\textbf{1. Transform Phase:}}

Apply GTT: $R_{\text{orig}} \leftarrow \Gfun_{2^n} R_{\text{orig}}$.

Current state: $\ket{\Psi_1} = \ket{\hat{\psi}}_{R_{\text{orig}}} \otimes \ket{0}_{q_a} \otimes \ket{0}_{R_{\text{comp}}}$,
where $\ket{\hat{\psi}}_{R_{\text{orig}}} = \sum_{y=0}^{N-1} \hat{a}_y \ket{y} = \left(\sum_{y_j \in S_k} \hat{a}_{y_j} \ket{y_j} + \sum_{y \notin S_k} \hat{a}_y \ket{y}\right)_{R_{\text{orig}}}$.

\BlankLine
\tcc{\textbf{2. Compression Phase:}}

Define a bijective mapping $f: S_k \rightarrow \{0, \ldots, k-1\}$.

Apply an oracle $O_S$ that flips $q_a$ if the state in $R_{\text{orig}}$ is in $S_k$:
$O_S \ket{y}_{R_{\text{orig}}} \ket{0}_{q_a} = \ket{y}_{R_{\text{orig}}} \ket{1}_{q_a}$ if $y \in S_k$
$O_S \ket{y}_{R_{\text{orig}}} \ket{0}_{q_a} = \ket{y}_{R_{\text{orig}}} \ket{0}_{q_a}$ if $y \notin S_k$.

Current state after flagging:
$\ket{\Psi_2} = \sum_{y_j \in S_k} \hat{a}_{y_j} \ket{y_j}_{R_{\text{orig}}} \ket{1}_{q_a} \ket{0}_{R_{\text{comp}}} + \sum_{y \notin S_k} \hat{a}_y \ket{y}_{R_{\text{orig}}} \ket{0}_{q_a} \ket{0}_{R_{\text{comp}}}$.

Apply controlled unitary $C_{q_a}U_{\text{map}}$ to $R_{\text{orig}}$ and $R_{\text{comp}}$
(where $U_{\text{map}}$ performs $\ket{y_j}_{R_{\text{orig}}} \ket{0}_{R_{\text{comp}}} \rightarrow \ket{0}^{\otimes n}_{R_{\text{orig}}} \ket{f(y_j)}_{R_{\text{comp}}}$ for $y_j \in S_k$).

Current state after controlled transfer:
$\ket{\Psi_3} = \sum_{y_j \in S_k} \hat{a}_{y_j} \ket{0}^{\otimes n}_{R_{\text{orig}}} \ket{1}_{q_a} \ket{f(y_j)}_{R_{\text{comp}}} + \sum_{y \notin S_k} \hat{a}_y \ket{y}_{R_{\text{orig}}} \ket{0}_{q_a} \ket{0}_{R_{\text{comp}}}$.

Measure $q_a$ \tcp*{$q_a=\ket{1}$ with probability $\mathcal{N}^2 = \sum_{y_j \in S_k} |\hat{a}_{y_j}|^2$}

\If{$q_a=\ket{1}$ }
{
    The state collapses to:
    $\frac{1}{\mathcal{N}} \sum_{y_j \in S_k} \hat{a}_{y_j} \ket{0}^{\otimes n}_{R_{\text{orig}}} \ket{1}_{q_a} \ket{f(y_j)}_{R_{\text{comp}}}$.

    Transmitted state (by tracing out $R_{\text{orig}}$ and $q_a$): $\ket{\hat{\psi}'_{\text{trunc}}}_{\text{transmitted}} = \frac{1}{\mathcal{N}} \sum_{y_j \in S_k} \hat{a}_{y_j} \ket{f(y_j)}_{R_{\text{comp}}}$.
}
\Else(\tcp*[f]{i.e., $q_a = \ket{0}$, with probability $1-\mathcal{N}^2$}) {
     \Return The message: ``the process failed for this trial."
}

\BlankLine
\tcc{\textbf{3. Transmit Phase:}}

Send $R_{\text{comp}}$ (i.e., $\ket{\hat{\psi}'_{\text{trunc}}}_{\text{transmitted}}$) via quantum channel to Bob (if $q_a=\ket{1}$).

Send classical data $S_k$ to Bob via classical channel.

\BlankLine
\tcc{\textbf{4. Decode and Reconstruct Phase:}}

Bob receives $\ket{\hat{\psi}'_{\text{trunc}}}_{\text{transmitted}}$ and $S_k$.

Bob's system starts as:
$\ket{\Psi_4} = \ket{\hat{\psi}'_{\text{trunc}}}_{\text{transmitted}} \otimes \ket{0}_{R'_{\text{orig}}}$.

Apply controlled unitary $C_{R_{\text{comp}}}U_{\text{decomp}}$ to $R_{\text{comp}}$ and $R'_{\text{orig}}$,
(where $U_{\text{decomp}}$ performs $\ket{f(y_j)}_{R_{\text{comp}}} \ket{0}_{R'_{\text{orig}}} \rightarrow \ket{f(y_j)}_{R_{\text{comp}}} \ket{y_j}_{R'_{\text{orig}}}$ for $y_j \in S_k$).

Current state after decompression:
$\ket{\Psi_5} = \frac{1}{\mathcal{N}} \sum_{y_j \in S_k} \hat{a}_{y_j} \ket{f(y_j)}_{R_{\text{comp}}} \ket{y_j}_{R'_{\text{orig}}}$.

Reconstructed sparse state in $R'_{\text{orig}}$ is:
$\ket{\hat{\psi}_{\text{approx}}}_{R'_{\text{orig}}} = \frac{1}{\mathcal{N}} \sum_{y_j \in S_k} \hat{a}_{y_j} \ket{y_j}_{R'_{\text{orig}}}$.

Apply inverse GTT: $R'_{\text{orig}} \leftarrow \Gfun_{2^n}^{\dagger} R'_{\text{orig}}$.

Final reconstructed state:
$\ket{\psi_{\text{approx}}}_{R'_{\text{orig}}} = \Gfun_{2^n}^{\dagger} \left( \frac{1}{\mathcal{N}} \sum_{y_j \in S_k} \hat{a}_{y_j} \ket{y_j} \right)_{R'_{\text{orig}}}$.

\Return Approximate reconstruction $\ket{\psi_{\text{approx}}}$.
\end{algorithm}

\begin{remark}
The bijective mapping $f: S_k \rightarrow \{0, \ldots, k-1\}$ can be chosen to be the simplest possible, $f(y_j) = j$, where $j$ is the ordered index of $y_j$ within the set $S_k$ (e.g., if $S_k = \{y_0, y_1, \ldots, y_{k-1}\}$ with $y_0 < y_1 < \ldots < y_{k-1}$, then $f(y_0)=0, f(y_1)=1$, etc.). This simplifies the classical encoding of the mapping and the quantum circuit design for $U_{\text{map}}$ and $U_{\text{decomp}}$.
\end{remark}

\subsubsection*{Key steps of the quantum compression algorithm:}

The algorithm begins with the application of the Generalized Tensor Transform (GTT). In the transformed domain the signal is sparse.  Next, a series of multi-controlled-X (C$^n$X) gates flip an ancilla qubit $q_a$ to $\ket{1}$ for those basis states corresponding to the $k$ dominant coefficients. This flags the desired components without disturbing quantum coherence. A controlled unitary $C_{q_a}U_{\text{map}}$ is then applied: if $q_a = \ket{1}$, it transfers the amplitude to a compressed $\lceil \log_b k \rceil$-qubit register and resets the original register to $\ket{0}^{\otimes n}$, a crucial uncomputation step ensuring reversibility. Finally, the inverse GTT is applied, and the ancilla qubit is measured. A successful measurement ($\ket{1}$) projects the system onto the desired subspace, leaving the compressed state in $R_{\text{comp}}$ ready for transmission.
At the receiver's end the inverse of compression step is applied to reconstruct the received signal.

\subsubsection{Example: Fully quantum compression protocol}

We will illustrate the fully quantum compression protocol with a small, concrete example, focusing on the conceptual construction of the $O_S$ oracle and the $U_{\text{map}}$ unitary.

Consider an initial $n=3$ qubit state in the $R_{\text{orig}}$ register, denoted by $Q_2 Q_1 Q_0$. We assume its Generalized Tensor Transform (GTT) results in a state $\ket{\hat{\psi}}$ that is sparse. For demonstration, let the dominant components be at indices $y_0=2=(010)_2$ and $y_1=5=(101)_2$, such that $S_k = \{2, 5\}$, implying $k=2$. The compressed state will reside in a single qubit register $R_{\text{comp}}$, denoted by $C_0$, as $\lceil \log_2 2 \rceil = 1$. A fixed mapping $f: S_k \rightarrow \{0, \ldots, k-1\}$ is chosen such that $f(2)=0$ and $f(5)=1$.

\subsubsection*{Compression phase:} 
The system's state after the GTT (Step 1 of the algorithm), assuming negligible amplitudes for non-$S_k$ components for illustrative simplicity, is:
$$
\ket{\Psi_1} = (\hat{a}_2 \ket{010} + \hat{a}_5 \ket{101})_{R_{\text{orig}}} \otimes \ket{0}_{q_a} \otimes \ket{0}_{R_{\text{comp}}}.
$$

\noindent \textit{Oracle Construction ($O_S$):} The oracle $O_S$ flags states corresponding to $y \in S_k$ by flipping the ancilla qubit $q_a$. Its action is defined as:
\begin{align*}
O_S \ket{y}_{R_{\text{orig}}} \ket{0}_{q_a} = 
\begin{cases}
\ket{y}_{R_{\text{orig}}} \ket{1}_{q_a} & \text{if } y \in S_k, \\
\ket{y}_{R_{\text{orig}}} \ket{0}_{q_a} & \text{if } y \notin S_k.
\end{cases}
\end{align*}
For our example, this means 
$O_S \ket{010}_{R_{\text{orig}}} \ket{0}_{q_a} = \ket{010}_{R_{\text{orig}}} \ket{1}_{q_a}$ and 
$O_S \ket{101}_{R_{\text{orig}}} \ket{0}_{q_a} = \ket{101}_{R_{\text{orig}}} \ket{1}_{q_a}$.
This can be typically implemented using multi-controlled X (MCX) gates. 
After $O_S$, the state becomes:
$$
\ket{\Psi_2} = \hat{a}_2 \ket{010}_{R_{\text{orig}}} \ket{1}_{q_a} \ket{0}_{R_{\text{comp}}} + \hat{a}_5 \ket{101}_{R_{\text{orig}}} \ket{1}_{q_a} \ket{0}_{R_{\text{comp}}} 
.
$$
\textit {Controlled Unitary ($C_{q_a}U_{\text{map}}$):} This unitary, controlled by $q_a$, performs the amplitude transfer and uncomputation. When $q_a=\ket{1}$, $U_{\text{map}}$ maps the relevant basis states from $R_{\text{orig}}$ to $R_{\text{comp}}$ while resetting $R_{\text{orig}}$ to $\ket{000}$. Specifically, it performs:
\begin{align*}
\ket{010}_{R_{\text{orig}}} \ket{0}_{R_{\text{comp}}} &\xrightarrow{U_{\text{map}}} \ket{000}_{R_{\text{orig}}} \ket{0}_{R_{\text{comp}}}, \\
\ket{101}_{R_{\text{orig}}} \ket{0}_{R_{\text{comp}}} &\xrightarrow{U_{\text{map}}} \ket{000}_{R_{\text{orig}}} \ket{1}_{R_{\text{comp}}}.
\end{align*}
The circuit for $U_{\text{map}}$ can be synthesized using multi-controlled gates that detect the specific $y_j$ states. For $\ket{010}$, the circuit would set $R_{\text{comp}}$ to $\ket{0}$ and uncompute $R_{\text{orig}}$ to $\ket{000}$. For $\ket{101}$, it would set $R_{\text{comp}}$ to $\ket{1}$ and uncompute $R_{\text{orig}}$ to $\ket{000}$. 
After $C_{q_a}U_{\text{map}}$, the state is:
$$
\ket{\Psi_3} = \hat{a}_2 \ket{000}_{R_{\text{orig}}} \ket{1}_{q_a} \ket{0}_{R_{\text{comp}}} + \hat{a}_5 \ket{000}_{R_{\text{orig}}} \ket{1}_{q_a} \ket{1}_{R_{\text{comp}}}.
$$

\noindent \textit{Measurement of $q_a$:} Alice measures $q_a$. A successful outcome ($q_a=\ket{1}$), occurring with probability $\mathcal{N}^2 = |\hat{a}_2|^2 + |\hat{a}_5|^2$, projects the state. The transmitted compressed state (by tracing out $R_{\text{orig}}$ and $q_a$) is:
$$
\ket{\hat{\psi}'_{\text{trunc}}}_{\text{transmitted}} = \frac{1}{\mathcal{N}} (\hat{a}_2 \ket{0}_{R_{\text{comp}}} + \hat{a}_5 \ket{1}_{R_{\text{comp}}}).
$$

\subsubsection*{Decode and reconstruct phase:}

Bob receives $\ket{\hat{\psi}'_{\text{trunc}}}_{\text{transmitted}}$ and the classical side information $S_k=\{2,5\}$.

\noindent \textit{Controlled Unitary ($C_{R_{\text{comp}}}U_{\text{decomp}}$):} Bob applies the inverse operation, $C_{R_{\text{comp}}}U_{\text{decomp}}$, which is controlled by $R_{\text{comp}}$ and maps to a new register $R'_{\text{orig}}$. $U_{\text{decomp}}$ effectively performs:
\begin{align*}
\ket{0}_{R_{\text{comp}}} \ket{000}_{R'_{\text{orig}}} &\xrightarrow{U_{\text{decomp}}} \ket{0}_{R_{\text{comp}}} \ket{010}_{R'_{\text{orig}}}, \\
\ket{1}_{R_{\text{comp}}} \ket{000}_{R'_{\text{orig}}} &\xrightarrow{U_{\text{decomp}}} \ket{1}_{R_{\text{comp}}} \ket{101}_{R'_{\text{orig}}}.
\end{align*}
After this, the state is:
$$
\ket{\Psi_5} = \frac{1}{\mathcal{N}} (\hat{a}_2 \ket{0}_{R_{\text{comp}}} \ket{010}_{R'_{\text{orig}}} + \hat{a}_5 \ket{1}_{R_{\text{comp}}} \ket{101}_{R'_{\text{orig}}}).
$$
Tracing out $R_{\text{comp}}$, Bob obtains the approximate sparse state in $R'_{\text{orig}}$:
$$
\ket{\hat{\psi}_{\text{approx}}}_{R'_{\text{orig}}} = \frac{1}{\mathcal{N}} (\hat{a}_2 \ket{010} + \hat{a}_5 \ket{101}).
$$

\noindent \textit{Inverse GTT:} Finally, Bob applies the inverse GTT, $\Gfun_N^\dagger$, to $\ket{\hat{\psi}_{\text{approx}}}_{R'_{\text{orig}}}$ to reconstruct the approximate original state $\ket{\psi_{\text{approx}}}$.

\subsubsection{Numerical example: Quantum state compression for $n=3$ qubits}
\label{sssec:numerical_example_compression}

In the following we present a numerical simulation using Qiskit to illustrate the Generalized Tensor Transform (GTT)-based quantum state compression protocol. This example demonstrates the key steps of transforming a state to reveal sparsity, truncating it, and then perfectly reconstructing it in a noiseless environment.

We consider an $n=3$ qubit system (total dimension $N=2^3=8$). The single-qubit unitary component $\Wfun$ of our GTT, $\Gfun_N = \Wfun^{\otimes n}$, is  defined to be $W = U3$ gate (Eq.~\meqref{eq:U3}) with parameters $\theta = \pi/4$, $\phi = \pi/3$, and $\lambda = \pi/6$.
For these parameters, $\Wfun$ approximates:
\[
\Wfun \approx
\begin{bmatrix}
0.924 & -0.331-0.191i \\
0.191+0.331i & 0.924i
\end{bmatrix}.
\]
While this $\theta$ choice results in non-uniform magnitudes for $\Wfun$ elements, the resulting $\Gfun_N$ remains a valid unitary tensor-product transform, suitable for demonstrating sparsity-based compression.

The simulation proceeds as follows:
\begin{enumerate}
	\item \textit{$\Gfun_N$ operator:} The matrix representation of the Generalized Tensor Transform $\Gfun_N$ for $n=3$ qubits, derived from the specified $\Wfun$ parameters, is:
	\begin{center}
		\footnotesize
		\renewcommand{\arraystretch}{1.1}
		\scalebox{0.85}{$
			\begin{bmatrix}
				0.789 & -0.283-0.163i & -0.283-0.163i & 0.068+0.117i & -0.283-0.163i & 0.068+0.117i & 0.068+0.117i & -0.056i \\
				0.163+0.283i & 0.789i & -0.135i & 0.163-0.283i & -0.135i & 0.163-0.283i & -0.028+0.049i & -0.117+0.068i \\
				0.163+0.283i & -0.135i & 0.789i & 0.163-0.283i & -0.135i & -0.028+0.049i & 0.163-0.283i & -0.117+0.068i \\
				-0.068+0.117i & -0.283+0.163i & -0.283+0.163i & -0.789 & 0.049-0.028i & 0.135 & 0.135 & 0.283+0.163i \\
				0.163+0.283i & -0.135i & -0.135i & -0.028+0.049i & 0.789i & 0.163-0.283i & 0.163-0.283i & -0.117+0.068i \\
				-0.068+0.117i & -0.283+0.163i & 0.049-0.028i & 0.135 & -0.283+0.163i & -0.789 & 0.135 & 0.283+0.163i \\
				-0.068+0.117i & 0.049-0.028i & -0.283+0.163i & 0.135 & -0.283+0.163i & 0.135 & -0.789 & 0.283+0.163i \\
				-0.056 & -0.117-0.068i & -0.117-0.068i & -0.163-0.283i & -0.117-0.068i & -0.163-0.283i & -0.163-0.283i & -0.789i
			\end{bmatrix}$.}
	\end{center}
	
		\item \textit{Alice's initial normalized state $|\psi\rangle$:}
	The initial 3-qubit state $|\psi\rangle$ in the computational basis, before compression, is:
	\begin{center}
		\small
		\scalebox{0.8}{$\begin{bmatrix}
				0.693-0.048i & -0.373+0.083i & -0.373+0.083i & -0.258-0.107i & -0.239+0.161i & 0.117-0.107i & 0.117-0.107i & 0.115-0.015i
			\end{bmatrix}^{\text{T}}$.}
	\end{center}

	\item \textit{Alice's transformed state $|\hat{\psi}\rangle$:}
	After applying $\Gfun_N$ to $|\psi\rangle$, the transformed state $|\hat{\psi}\rangle$ in the GTT basis is obtained:
	\begin{center}
		\small
		$\begin{bmatrix}
			0.914 & 0.0 & 0.0 & 0.406 & 0.0 & 0.0 & 0.0 & 0.0
		\end{bmatrix}^{\text{T}}.$
	\end{center}
	This state exhibits clear sparsity, with two dominant coefficients corresponding to indices 0 and 3. 
	
	\item \textit{Alice's truncated and compressed state:}
	Alice truncates $|\hat{\psi}\rangle$ by keeping only the $k=2$ dominant coefficients (at indices 0 and 3). This truncated state, after re-normalization, is:\\
	\begin{center}
		coefficients (approximated):
		\small
		$\begin{bmatrix}
			0.914 & 0.0 & 0.0 & 0.406 & 0.0 & 0.0 & 0.0 & 0.0
		\end{bmatrix}^{\text{T}}$, with   indices retained ($S_k$): $[0 \ 3]$.
	\end{center}
	This state is then encoded into a compressed form, which requires only $\lceil \log_2 2 \rceil = 1$ qubit for transmission. The coefficients of this 1-qubit compressed state are:
	\small
	$\begin{bmatrix}
		0.914 & 0.406
	\end{bmatrix}^{\text{T}}$.
	The classical side information ($S_k = [0 \ 3]$) is transmitted alongside the 1-qubit quantum state.

		\item \textit{Bob's decoded and reconstructed state $|\psi_{\text{approx}}\rangle$:}
	Upon receipt, Bob decodes the compressed 1-qubit state back into the 8-dimensional GTT basis (using the classical side information). The decoded state $|\hat{\psi}_{\text{approx}}\rangle$ is:
	\begin{center}
		\small
		$\begin{bmatrix}
			0.914 & 0.0 & 0.0 & 0.406 & 0.0 & 0.0 & 0.0 & 0.0
		\end{bmatrix}^{\text{T}}$.
	\end{center}
	Finally, Bob applies the inverse GTT ($\Gfun_N^\dagger$) to reconstruct an approximation of the original state $|\psi\rangle$:
	\begin{center}
		\small
		\scalebox{0.8}{$
			\begin{bmatrix}
				0.693-0.048i & -0.373+0.083i & -0.373+0.083i & -0.258-0.107i & -0.239+0.161i & 0.117-0.107i & 0.117-0.107i & 0.115-0.015i
			\end{bmatrix}^{\text{T}}$.}
	\end{center}
	
\end{enumerate}

\subsubsection*{Reconstruction fidelity:}
We recall that the fidelity $F(|\psi\rangle, |\phi\rangle)$ between two normalized quantum states $|\psi\rangle$ and $|\phi\rangle$ is a measure that quantifies their similarity or overlap. It is formally defined as the squared magnitude of their inner product:
$$
F(|\psi\rangle, |\phi\rangle) = |\langle \psi | \phi \rangle|^2.
$$
The fidelity value ranges from $0$ to $1$. A fidelity of $F=1$ indicates that the states are identical (or differ only by an unobservable global phase), while $F=0$ signifies that they are perfectly orthogonal and thus maximally dissimilar. In the context of quantum state compression, fidelity is used to quantify how well a reconstructed state approximates the original state after compression.

For the case under consideration, the fidelity between the original state $|\psi\rangle$ and the reconstructed state $|\psi_{\text{approx}}\rangle$ is calculated to be $1.0$.
A fidelity of $1.0$ indicates perfect reconstruction in this noiseless simulation, achieved because the initial state was designed to be exactly sparse in the chosen GTT basis, and all dominant components were retained. This example clearly demonstrates the potential for significant quantum state compression using the Generalized Tensor Transform.

\subsubsection*{Comparison with Hadamard-based truncation:}
To further highlight the strength of the GTT-based compression for states designed to be sparse in its basis, we can compare its performance with a fixed, well-known transform like the Hadamard transform ($H^{\otimes n}$). Using the same initial state $|\psi\rangle$ and applying the Hadamard transform yields the following state in the Hadamard basis:
\begin{center}
\small \scalebox{0.85}{
$\begin{bmatrix}
  -0.072-0.020i & 0.211+0.083i & 0.211+0.083i & 0.291+0.014i & -0.149+0.028i & 0.461-0.041i & 0.461-0.041i & 0.544-0.240i
\end{bmatrix}^{\text{T}}$}.
\end{center}
Unlike the GTT-transformed state, this Hadamard-transformed state does not exhibit the same level of clear sparsity, indicating that $|\psi\rangle$ is not naturally sparse in the Hadamard basis. To quantify the impact on compression, we apply the same truncation strategy: retaining only the two ($k=2$) dominant coefficients.
For the Hadamard-transformed state, the two dominant coefficients are found at indices [7 6]. The state after truncating to these two components (and re-normalizing) is given below.
Truncated coefficients (approximated in Hadamard basis, $k=2$):
\begin{center}
\small
$\begin{bmatrix}
  0.0 & 0.0 & 0.0 & 0.0 & 0.0 & 0.0 & 0.612-0.055i & 0.722-0.319i
\end{bmatrix}^{\text{T}}$.
\end{center}
The norm of discarded coefficients is $0.4315$, which is substantial.
Upon applying the inverse Hadamard transform to this truncated state, and comparing the result with the original state $|\psi\rangle$, the reconstruction fidelity for Hadamard-based compression (retaining $k=2$ coefficients) is calculated to be approximately: $0.5685$.

\subsubsection*{Comparison with Quantum Fourier Transform (QFT):}
Another fundamental quantum transform, widely used in various quantum algorithms (e.g., phase estimation, Shor's algorithm), is the Quantum Fourier Transform (QFT). We extend our comparison to include the QFT to further demonstrate the efficacy of our GTT for this particular state. Applying the QFT to the same initial state $|\psi\rangle$ yields the following state in the QFT basis:
\begin{center}
\small\scalebox{0.85}
{$\begin{bmatrix}
    -0.072-0.020i & 0.209-0.392i & 0.217+0.009i & 0.401-0.187i & 0.211+0.083i & 0.316-0.102i & 0.286+0.088i & 0.392+0.386i
\end{bmatrix}^{\text{T}}$.}
\end{center}

Similar to the Hadamard case, the QFT-transformed state does not exhibit the strong sparsity observed in the GTT basis for this specific initial state $|\psi\rangle$. To quantify the compression performance, we again truncate the QFT-transformed state by retaining only its two ($k=2$) dominant coefficients.
For the QFT-transformed state, the two dominant coefficients are found at indices [1, 7]. The truncated state (keeping $k=2$ coefficients) in the QFT basis, after re-normalization, is given below.
Truncated coefficients (approx. in QFT basis, $k=2$):
\begin{center}
\small
$\begin{bmatrix}
  0.0 & 0.295-0.555i & 0.0 & 0.0 & 0.0 & 0.0 & 0.0 & 0.554+0.546i
\end{bmatrix}^{\text{T}}$.
\end{center}
The norm of the discarded coefficients is $0.4999$.
After applying the inverse QFT to this truncated state and comparing it with the original state $|\psi\rangle$, the reconstruction fidelity for QFT-based compression (retaining $k=2$ coefficients) is approximately: $0.5001$.

\subsubsection*{Summary of fidelity comparison (retaining $k=2$ coefficients):}

To summarize, consider Alice's initial state:
\begin{equation*}
\small
\scalebox{0.8}{$ 
\ket{\psi} = \ket{S_1} = \begin{bmatrix}
  0.693 - 0.048i & -0.373 + 0.083i & -0.373 + 0.083i & -0.258 - 0.107i & -0.239 + 0.161i & 0.117 - 0.107i & 0.117 - 0.107i & 0.115 - 0.015i
\end{bmatrix}^{\text{T}}$.}
\end{equation*}
For this state, the fidelity obtained using GTT-based compression was $1.0000$, while Hadamard-based and QFT-based compression gave fidelities of $0.5685$ and $0.5001$, respectively.

Now consider the following two initial states:
\begin{equation*}
\small
\scalebox{0.8}{$ 
\ket{\psi} = \ket{S_2} = \begin{bmatrix}
  0.706 - 0.076i & -0.371 + 0.096i & -0.371 + 0.003i & -0.241 - 0.078i & -0.238 + 0.173i & 0.113 - 0.111i & 0.133 - 0.078i & 0.102 - 0.022i
\end{bmatrix}^{\text{T}}$.}
\end{equation*}
\begin{equation*}
\small
\scalebox{0.8}{$ 
\ket{\psi} = \ket{S_3} = \begin{bmatrix}
  0.718 - 0.101i & -0.370 + 0.108i & -0.370 + 0.015i & -0.242 - 0.082i & -0.237 + 0.101i & 0.128 - 0.085i & 0.147 - 0.052i & 0.091 - 0.028i
\end{bmatrix}^{\text{T}}$.}
\end{equation*}
For these states, GTT-based compression yielded fidelities of $0.9797$ and $0.9637$, respectively. The corresponding fidelities from Hadamard-based compression were $0.5737$ and $0.5815$, while QFT-based compression yielded $0.4845$ and $0.4879$. These results are summarized in Table~\ref{tab_fidelity}.

\begin{table}[h!]
\centering
\caption{Fidelity comparison for compression retaining $k=2$ coefficients.} \label{tab_fidelity}
\begin{tabular}{@{} cccc @{}}
\toprule
\textbf{Initial State} & \textbf{GTT Fidelity} & \textbf{Hadamard Fidelity} & \textbf{QFT/FFT Fidelity} \\
\midrule
$\ket{S_1}$  & 1.0000 & 0.5685 & 0.5001 \\
$\ket{S_2}$ & 0.9797 & 0.5737 & 0.4845 \\
$\ket{S_3}$ & 0.9637 & 0.5815 & 0.4879 \\
\bottomrule
\end{tabular}
\end{table}

To provide a clear visual assessment of the performance of the proposed GTT-based compression approach against standard quantum transforms for the chosen input quantum states, we present a comparison of the initial quantum states and their reconstructed versions. As shown collectively in Figure \ref{fig:all_state_comparisons_stacked}, we illustrate the magnitudes of the amplitudes of the original states (represented by blue bars) in the computational basis, alongside those of the states reconstructed after compression using GTT (green bars), Hadamard (orange bars), and Quantum Fourier Transform (QFT) (red bars). Each compression method involved retaining $k=2$ dominant coefficients, allowing for a direct comparison of their information retention capabilities under a fixed compression ratio.

Specifically, Figure \ref{fig:state_comparison_S1} shows the comparison for the initial signal $\ket{S_1}$, where the GTT-reconstructed state demonstrates a very close resemblance to the original, reflecting its high fidelity. This trend is consistently observed for the other two distinct initial quantum signals: Figure \ref{fig:state_comparison_S2} for $\ket{S_2}$ and Figure \ref{fig:state_comparison_S3} for $\ket{S_3}$. In each instance, the GTT-reconstructed states (green bars) visually align remarkably well with the original states (blue bars), underscoring the GTT's effectiveness in preserving essential information even with aggressive truncation. 

Conversely, the reconstructed states from both the Hadamard (orange bars) and QFT (red bars) transforms exhibit more deviations from the original signals' spectra. This visual evidence directly correlates with their lower numerical fidelities, reinforcing the argument that for the chosen intial quantum states, a customizable transform like the GTT offers superior compression and reconstruction quality compared to generic, fixed transforms. The ability of GTT to concentrate state information into fewer coefficients is thus visibly confirmed, leading to near-perfect reconstruction where other transforms suffer substantial loss.

\begin{figure}[htbp]
    \centering

    \begin{subfigure}{0.75\textwidth} %
        \centering
        \includegraphics[width=\linewidth]{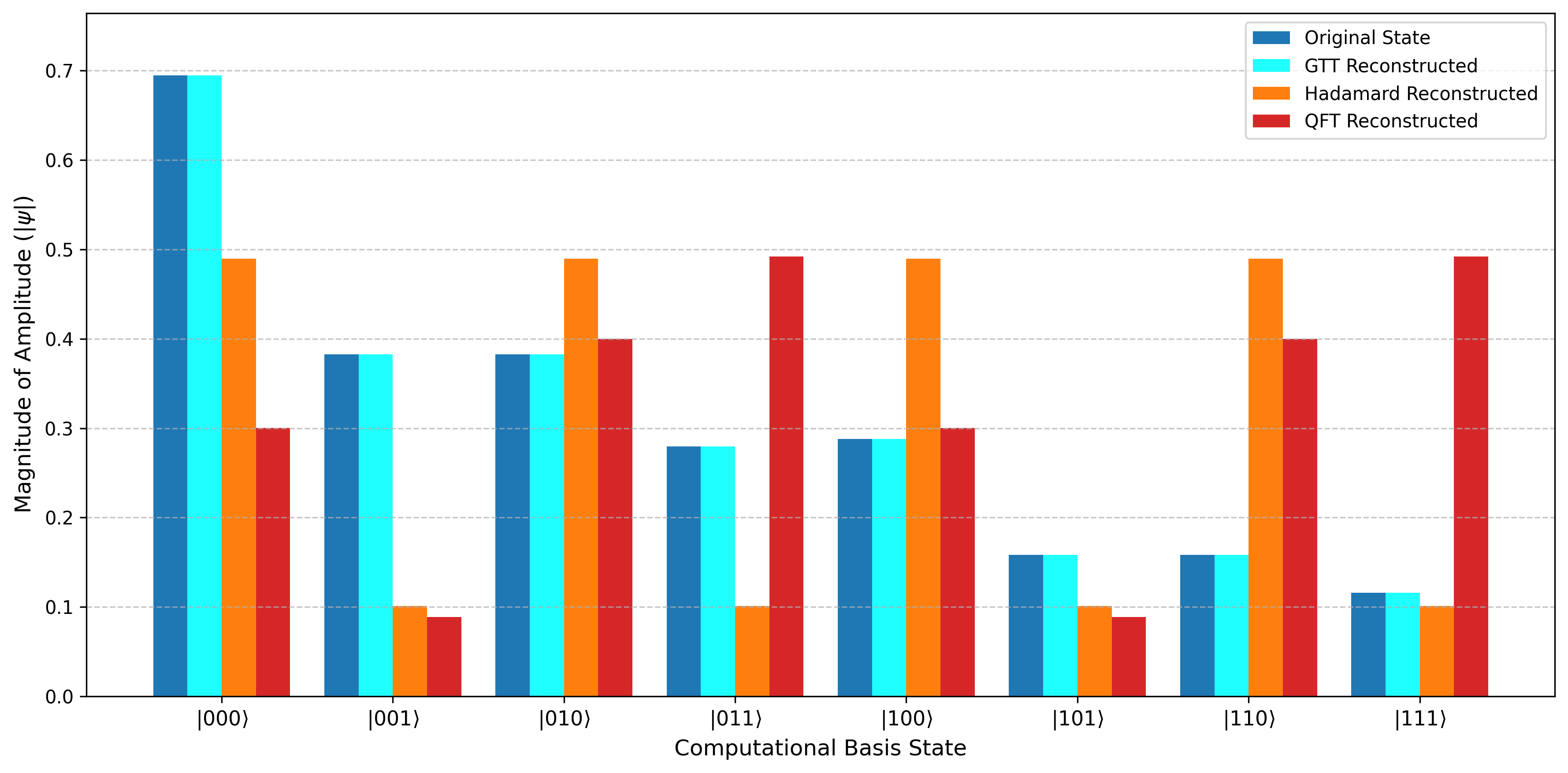}
        \caption{
        Comparison of original the initial quantum state  $\ket{S_1}$ and reconstructed signals.
        }
        \label{fig:state_comparison_S1}
    \end{subfigure}
    \par %
    \vspace{1em} %

    \begin{subfigure}{0.75\textwidth}
        \centering
        \includegraphics[width=\linewidth]{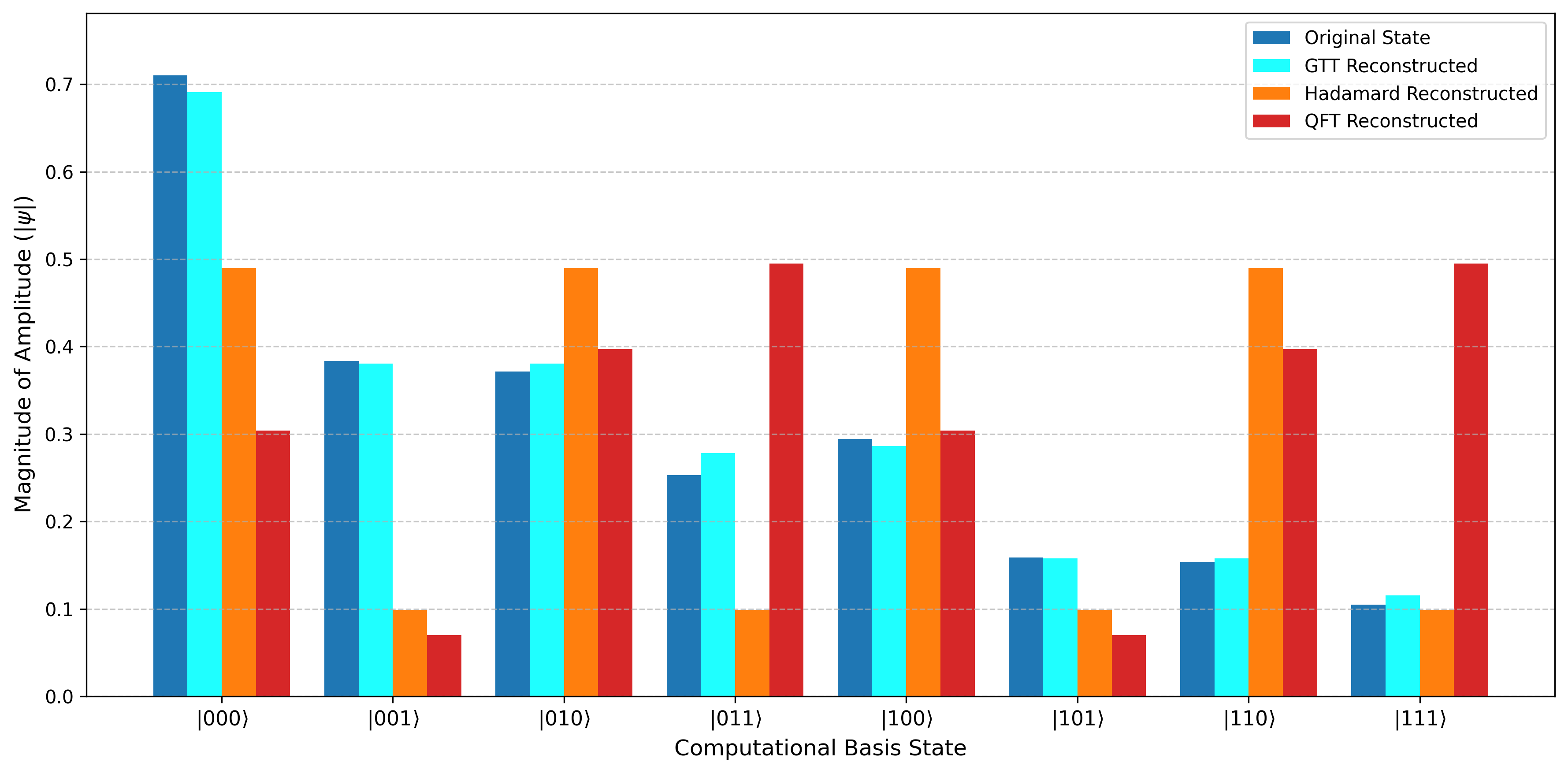}
         \caption{
        Comparison of original the initial quantum state  $\ket{S_2}$ and reconstructed signals.
        }
        \label{fig:state_comparison_S2}
    \end{subfigure}
    \par
    \vspace{1em}

    \begin{subfigure}{0.75\textwidth}
        \centering
        \includegraphics[width=\linewidth]
        {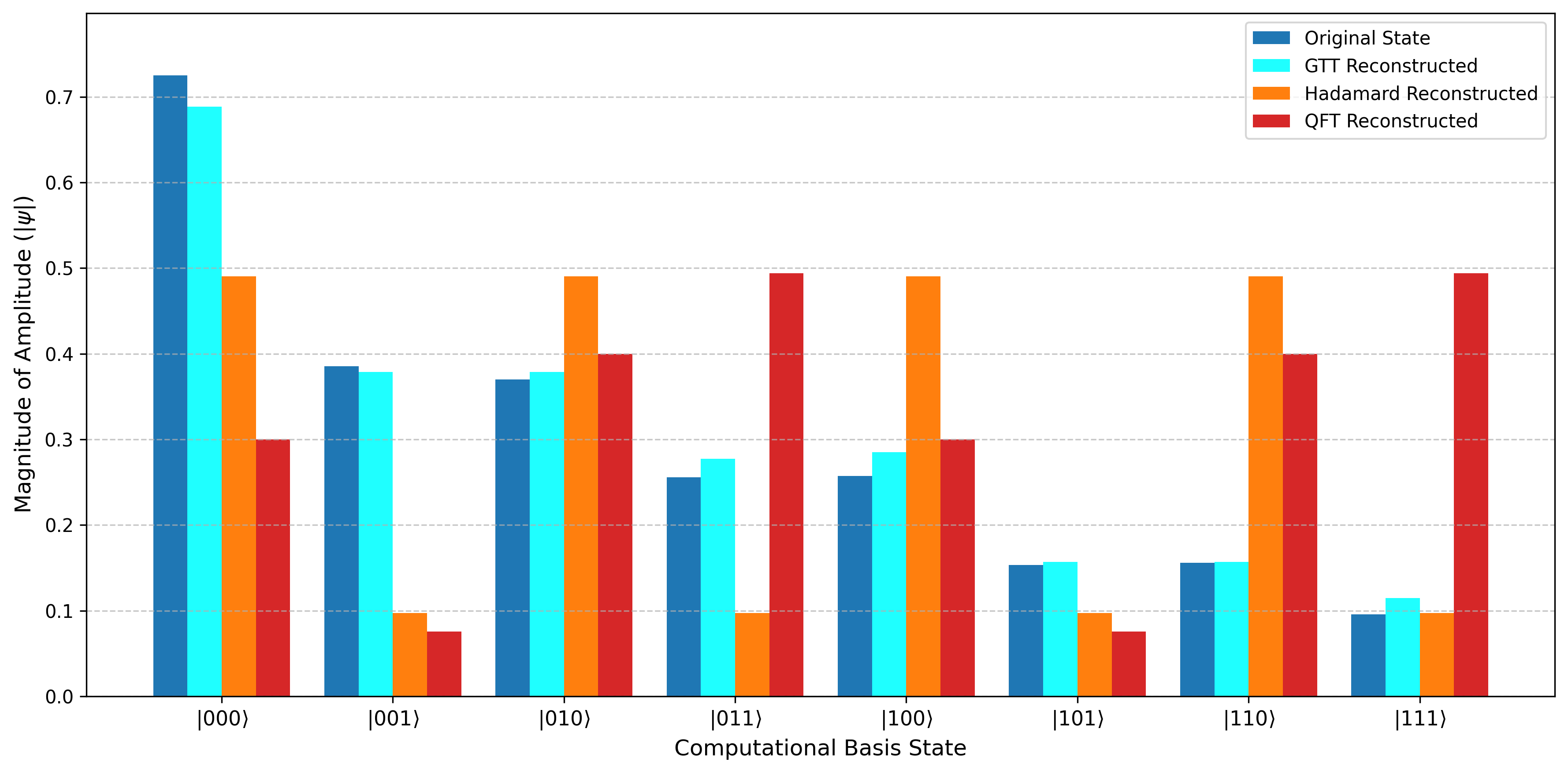}
         \caption{
        Comparison of original the initial quantum state  $\ket{S_3}$ and reconstructed signals.
        }
        \label{fig:state_comparison_S3}
    \end{subfigure}

    \caption{Comparison of original and reconstructed signals for different initial quantum signals. Each subfigure presents the magnitudes of the amplitudes of the original quantum state (blue bars) in the computational basis, along with its reconstructed approximations using GTT (green bars), Hadamard (orange bars), and QFT (red bars) based compression techniques. In all cases, only $k=2$ dominant coefficients were retained during compression.
    }
    \label{fig:all_state_comparisons_stacked}
\end{figure}

This multi-transform comparison conclusively demonstrates the power of the GTT for the examples chosen. While Hadamard and QFT are powerful, generic transforms, the GTT's customizable nature allows it to be capable of efficintly representing a broader class of signals.  This ability leads to higher compression efficiency and reconstruction fidelity for states not naturally sparse in generic bases, making the GTT a highly effective tool for quantum state compression.

\subsubsection{Computational complexity}
In this section, we analyze the computational and communication complexities of the proposed quantum state compression protocols, and compare them with a purely classical approach. We quantify gate complexity for quantum circuits and arithmetic operations for classical protocols, assuming an $n$-qudit state with $N=b^n$ dimensions and $k \prec \prec N$ dominant coefficients in a Generalized Tensor Transform (GTT) basis. 

\subsubsection*{Purely classical protocol complexity:}
For an $N$-dimensional state sparse in a GTT basis, a classical transmission protocol entails Alice computing the GTT on the classical amplitude vector ($\mathcal{O}(N \log_b N)$ arithmetic operations to obtain all $N$ coefficients). If the $k$ dominant indices are not \textit{a priori} known, identifying them necessitates sorting, incurring an additional $\mathcal{O}(N \log_b N)$ cost. Even if the indices are known, the GTT computation itself remains $\mathcal{O}(N \log_b N)$. Alice then transmits these $k$ indices and their complex amplitudes, requiring $\mathcal{O}(k(\log_b N + P))$ bits (where $P$ denotes amplitude precision). Bob's reconstruction involves an inverse GTT, also $\mathcal{O}(N \log_b N)$. Consequently, the overall classical computational complexity is dominated by $\mathcal{O}(N \log_b N)$, with a classical communication overhead of $\mathcal{O}(k(\log_b N + P))$ bits. 

\subsubsection*{Hybrid quantum-classical compression protocol complexity:}
The hybrid quantum-classical protocol integrates quantum transformation with classical coefficient identification. The initial quantum transformation, $\Gfun_N = \Wfun^{\otimes n}$, applied to the $n$-qudit state requires $\mathcal{O}(\log_b N)$ quantum gates (e.g., for Hadamard-like GTTs, since $n = \log_b N$). The subsequent ``Analyze and Truncate'' step is hybrid and lossy, converting quantum amplitudes to classical information.

If the $k$ dominant indices $S_k$ \textit{are} known \textit{a priori}, this step simplifies to Alice performing amplitude estimation only for these $k$ specified components. This involves $\mathcal{O}(k/\epsilon)$ quantum operations (query complexity), where $\epsilon$ is the desired precision for the amplitudes. Alice then prepares a $\lceil \log_b k \rceil$-qudit state from the $k$ selected classical amplitudes, a process requiring $\mathcal{O}(k)$ quantum gates using established state preparation techniques. For transmission, Alice sends this compressed quantum state along with $\mathcal{O}(k \log_b N)$ classical bits for the index set $S_k$. Bob's reconstruction involves applying an inverse GTT to the received state, costing $\mathcal{O}(\log_b N)$ quantum gates, guided by the classical $S_k$ information. In this scenario (indices known \textit{a priori}), the hybrid approach offers a quantum communication advantage of $\mathcal{O}(\log_b k)$ qudits (versus $\log_b N$ for the original state), with a dominant computational cost (quantum gates) of $\mathcal{O}(\log_b N + k/\epsilon + k)$.

If the $k$ dominant indices are \textit{not} known \textit{a priori}, discovering them from the quantum state through comprehensive amplitude amplification and estimation (e.g., iteratively finding and uncomputing amplitudes above a dynamically adjusted threshold) would typically incur a quantum query complexity scaling as $\mathcal{O}(\sqrt N)$ to find the dominant indices, plus $\mathcal{O}(k/\epsilon)$ to estimate their amplitudes. This discovery phase presents a  bottleneck for the overall protocol in the absence of \textit{a priori} knowledge. 

\subsubsection*{Fully quantum compression protocol complexity:}
The fully quantum protocol, designed for full coherence preservation, operates under the crucial assumption that the indices of the $k$ dominant sparse components ($S_k$) are known \textit{a priori}, presumably via classical heuristics as a pre-processing step. The ``Transform Phase,'' involving the application of $\Gfun_{2^n}$, demands $\mathcal{O}(\log_b N)$ quantum gates. The ``Compression Phase'' primarily involves an oracle $O_S$ and a controlled unitary $C_{q_a}U_{\text{map}}$.

Based on recent advancements in multi-controlled gate decomposition \cite{zindorf2025multi}, a multi-controlled X gate with $n_{\text{controls}}$ controls can be implemented with $\mathcal{O}(n_{\text{controls}})$ CNOT cost and constant depth (e.g., depth 2 for CNOTs, T, and H gates, for $n_{\text{controls}} \ge 5$).
\begin{itemize}
    \item The $O_S$ oracle, consisting of $k$ distinct $\log_b N$-qudit controlled-X (C$^{\log_b N}$X) gates, each flipping an ancilla qudit $q_a$ based on a specific $y_j \in S_k$, would therefore have a total gate complexity of $\mathcal{O}(k \cdot \log_b N)$ if these operations are performed sequentially. Its depth would be $\mathcal{O}(k)$ (constant depth per gate, applied $k$ times sequentially).
    \item The controlled unitary $C_{q_a}U_{\text{map}}$ performs $k$ distinct mappings, transferring amplitude and uncomputing $R_{\text{orig}}$. Assuming the dominant cost of synthesizing these mappings also arises from multi-controlled X-like operations that can leverage the $\mathcal{O}(\log_b N)$ gate cost and constant depth from \cite{zindorf2025multi}, its total gate complexity would also be $\mathcal{O}(k \cdot \log_b N)$.
\end{itemize}
The ``Transmit Phase'' entails sending $\lceil \log_b k \rceil$ qudits quantumly and $\mathcal{O}(k \log_b N)$ classical bits for $S_k$. Bob's ``Decode and Reconstruct Phase'' involves applying the inverse controlled unitary $C_{R_{\text{comp}}}U_{\text{decomp}}$ (similar gate complexity to $C_{q_a}U_{\text{map}}$) and the inverse GTT ($\mathcal{O}(\log_b N)$ gates). Therefore, the dominant quantum gate complexity of the fully quantum protocol is $\mathcal{O}(\log_b N + k \cdot \log_b N) = \mathcal{O}(k \log_b N)$, assuming $k \ge 1$. This notably avoids the $\mathcal{O}(N \log_b N)$ cost of classical coefficient identification by leveraging \textit{a priori} knowledge of $S_k$ and significantly improves upon prior multi-controlled gate complexities. The quantum communication cost is significantly reduced to $\mathcal{O}(\log_b k)$ qudits, a considerable advantage when $k \prec \prec N$.

\subsection{Function encoding and Generalized Tensor Transform (GTT)}

\label{sec:fun_encoding}

A continuous classical function $f(x)$ defined over an interval (e.g., $[0,1]$) can be encoded into quantum states. One way to achieve this is by discretizing the function into $N=b^n$ points, where $n$ is the number of qudits of dimension $b$. Each discrete function value $f(x_j)$ is then mapped to the amplitude of a computational basis state $|j\rangle$, forming a quantum state:
$$ |\psi_f\rangle = \frac{1}{\mathcal{N}} \sum_{j=0}^{N-1} f(x_j) |j\rangle ,$$
where $\mathcal{N} = \sqrt{\sum_{j=0}^{N-1} |f(x_j)|^2}$ is a normalization constant ensuring the state is a valid quantum state ($\langle\psi_f|\psi_f\rangle = 1$).

While this method allows for representation, the number of amplitudes $N$ grows exponentially with the number of qudits $n$. For practical applications on contemporary quantum hardware, which has limited qudit availability, it is imperative to develop methods for efficient state compression and representation, thereby mitigating the ``curse of dimensionality'' by capturing essential function information with fewer effective components in a transformed basis.

The Generalized Tensor Transform (GTT) offers a solution to this challenge by identifying an optimal unitary transformation that induces sparsity in the function's representation within the transformed domain. Unlike fixed transforms such as the Hadamard or Quantum Fourier Transform, the GTT provides inherent flexibility through its tunable parameters.

We recall that the GTT is a multi-qudit unitary operator $\Gfun_N$, defined as the $n$-fold tensor product of a general $b \times b$ unitary gate $\Wfun$, i.e.,
$ \Gfun_N = \Wfun^{\otimes n}$, where $N=b^n$. For the common case of $b=2$ (qubits), the single-qubit gate $\Wfun$ can be taken as a general unitary matrix 
\[
\Wfun = U3(\theta, \phi, \lambda) =
\begin{bmatrix}
	\cos\left(\frac{\theta}{2}\right) & -e^{i\lambda} \sin\left(\frac{\theta}{2}\right) \\
	e^{i\phi} \sin\left(\frac{\theta}{2}\right) & e^{i(\phi + \lambda)} \cos\left(\frac{\theta}{2}\right)
\end{bmatrix},
\]
which is the same matrix as used in \meqref{eq:U3}. 

We note that the tunable parameters ($\theta, \phi, \lambda$) are fundamental to the GTT's adaptability. By optimizing these parameters, $\Gfun_N$ can be precisely tailored to align with the specific characteristics of the function being encoded, facilitating maximal sparsity.

The GTT process for function compression involves the following steps:
\begin{enumerate}
    \item \textit{Initial encoding:} The discretized function $f(x)$ is encoded into an initial quantum state $|\psi_f\rangle$.
    \item \textit{Transformation:} The initial quantum state is transformed into a new basis by applying the GTT operator: $|\tilde{\psi}\rangle = \Gfun_N(\theta, \phi, \lambda) |\psi_f\rangle $.
    \item \textit{Sparsity and truncation:} The amplitudes of $|\tilde{\psi}\rangle$ are analyzed. If the chosen GTT is effective, $|\tilde{\psi}\rangle$ will exhibit sparsity, meaning the vast majority of its amplitudes are negligible. Only the $k$ components with the largest magnitudes are retained, forming a compressed state $|\tilde{\psi}_{\text{compressed}}\rangle$.
    \item \textit{Reconstruction:} To reconstruct the function's approximation, the inverse GTT operator $\Gfun_N^\dagger(\theta, \phi, \lambda)$ is applied to the compressed state: $|\psi_{\text{reconstructed}}\rangle = \Gfun_N^\dagger(\theta, \phi, \lambda) |\tilde{\psi}_{\text{compressed}}\rangle$.
    \item \textit{Optimization:} The GTT parameters $(\theta, \phi, \lambda)$ are classically optimized to maximize the fidelity $F$, (where $F = |\langle\psi_f|\psi_{\text{reconstructed}}\rangle|^2$) for a predetermined number of retained components $k$.
\end{enumerate}

The primary advantage of GTT over fixed transforms, such as the Hadamard transform, lies in its inherent flexibility to adapt to the specific structure of the function being encoded. The multi-qubit Hadamard transform $H^{\otimes n}$, for instance, is defined by a fixed single-qubit gate $H = \frac{1}{\sqrt{2}}\begin{bmatrix} 1 & 1 \\ 1 & -1 \end{bmatrix}$. While highly effective for functions already sparse in the Walsh-Hadamard basis, its performance is suboptimal for functions whose information is not concentrated within this predefined basis.

In contrast, GTT's parameters $\theta$, $\phi$, and $\lambda$ provide three degrees of freedom, enabling the single-qudit unitary operator $\Wfun$ (and consequently $\Gfun_N$) to be dynamically ``tuned'' or ``rotated.'' This adaptive mechanism allows GTT to discover a basis in which the function's representation becomes maximally sparse. This flexibility means that GTT can often concentrate the function's energy into a significantly smaller number of components ($k$) than any fixed transform, leading to superior compression.

\subsubsection{Effective representation with a few components}
The ability of GTT to represent a function effectively with only a few components directly stems from this adaptive sparsity. A function is considered sparse in a given transform domain if it can be expressed as a superposition of a small number of basis states within that domain. Such a suitable domain could be identified, by fine tuning and appropriate selection of the  parameters $(\theta, \phi, \lambda)$. Once this optimal GTT domain is identified, the compression is achieved by retaining only the $k$ components with the largest magnitudes in this domain, while discarding the remaining negligible coefficients. This is also relevant to the discussion presented in \mref{ssec:state_compression}.

By effectively discarding these negligible coefficients, GTT facilitates high compression ratios while maintaining high reconstruction fidelity. This efficiency is critical for quantum applications, as it allows the representation of complex functions using a minimal number of qubits or computational resources. This makes GTT a powerful tool for quantum digital signal processing and data compression on both current and future quantum computing platforms.

\subsubsection{Computational example: Function encoding}

Next we provide a detailed computational example to demonstrate the efficacy of the Generalized Tensor Transform (GTT) for compressing a function. We compare GTT's reconstruction performance against the fixed Hadamard transform.

The function under consideration is defined on the interval $[0,1]$ and discretized into $N = 2^{n}$ points, where $n$ is the number of qubits. For this example, we set $n = 4$, yielding $N = 16$ discrete points.
The function $f(x)$ is constructed to represent a signal that, while having an underlying structure potentially amenable to GTT compression, also includes perturbations to test the robustness of the method. It is defined as a polynomial function with added sinusoidal and exponential terms:
$$
f(x) = -978.7 x^7 + 3677 x^6 - 5575 x^5 + 4366 x^4 - 1875 x^3 + 431.6 x^2 - 47.57 x + 1.886 + 0.1\sin(0.1x) - 0.01e^{-x}.
$$
This function is discretized using the midpoints of the $N = 16$ subintervals within $[0,1]$. This discretized data is then normalized to form the initial quantum state vector $\ket{\psi_{\text{initial}}}$. 

The GTT optimization procedure is employed to discover the optimal GTT parameters $(\theta, \phi, \lambda)$ that maximize the fidelity of reconstruction for a specified number of retained components ($k$). This optimization was performed using the \texttt{minimize} function from the \texttt{scipy.optimize} library in Python. The objective function for the minimization was defined as $1 - F $  (where $F$ is the Fidelity) aiming to find parameters that yield the highest possible fidelity. An initial guess for the GTT parameters was set to $[\phi = 0.0, \lambda = \pi]$ and different initial guesses were tried for $\theta $ between $[0, \frac{\pi}{4}]$. The optimization used the L-BFGS-B method with bounds $[0, 2\pi]$ for each parameter.
For a comparative analysis, the reconstruction fidelity using the fixed Hadamard transform is also computed. The outcomes for $k=4, 8,$ and $12$ components are presented in Table \ref{tab:gtt_hadamard_results}.

\begin{table}[h!]
    \centering
    \caption{Compression results for the perturbed polynomial function ($N_{\text{components}}=16$)}
    \label{tab:gtt_hadamard_results}
    \begin{tabular}{ccccc}
        \toprule
        $k$ & Optimal $\theta$  & GTT fidelity & Hadamard fidelity \\
        \midrule
        4  & \num{0.2072} & \num{0.9181} & \num{0.2377} \\
        8  & \num{0.3242}  & \num{0.9848} & \num{0.5386} \\
        12 & \num{0.1560}  & \num{0.9964} & \num{0.8425} \\
        \bottomrule
    \end{tabular}
\end{table}

The results demonstrate that the GTT consistently achieves substantially higher reconstruction fidelity compared to the Hadamard transform across all tested compression levels. As the number of retained components ($k$) increases, the GTT fidelity rapidly approaches unity, signifying highly accurate reconstruction of the original function.

\subsubsection{Visual representation}
Figures \ref{fig:k4_reconstruction}, \ref{fig:k8_reconstruction}, and \ref{fig:k12_reconstruction} visually illustrate the original normalized function and its reconstructions using both GTT and Hadamard transforms for $k=4, 8,$ and $12$ retained components, respectively. These plots provide strong visual evidence corroborating the superior performance of GTT as indicated by the numerical fidelity values.

\begin{figure}[H]
    \centering
    \begin{subfigure}[b]{0.65\textwidth} %
        \centering
        \includegraphics[width=\textwidth]{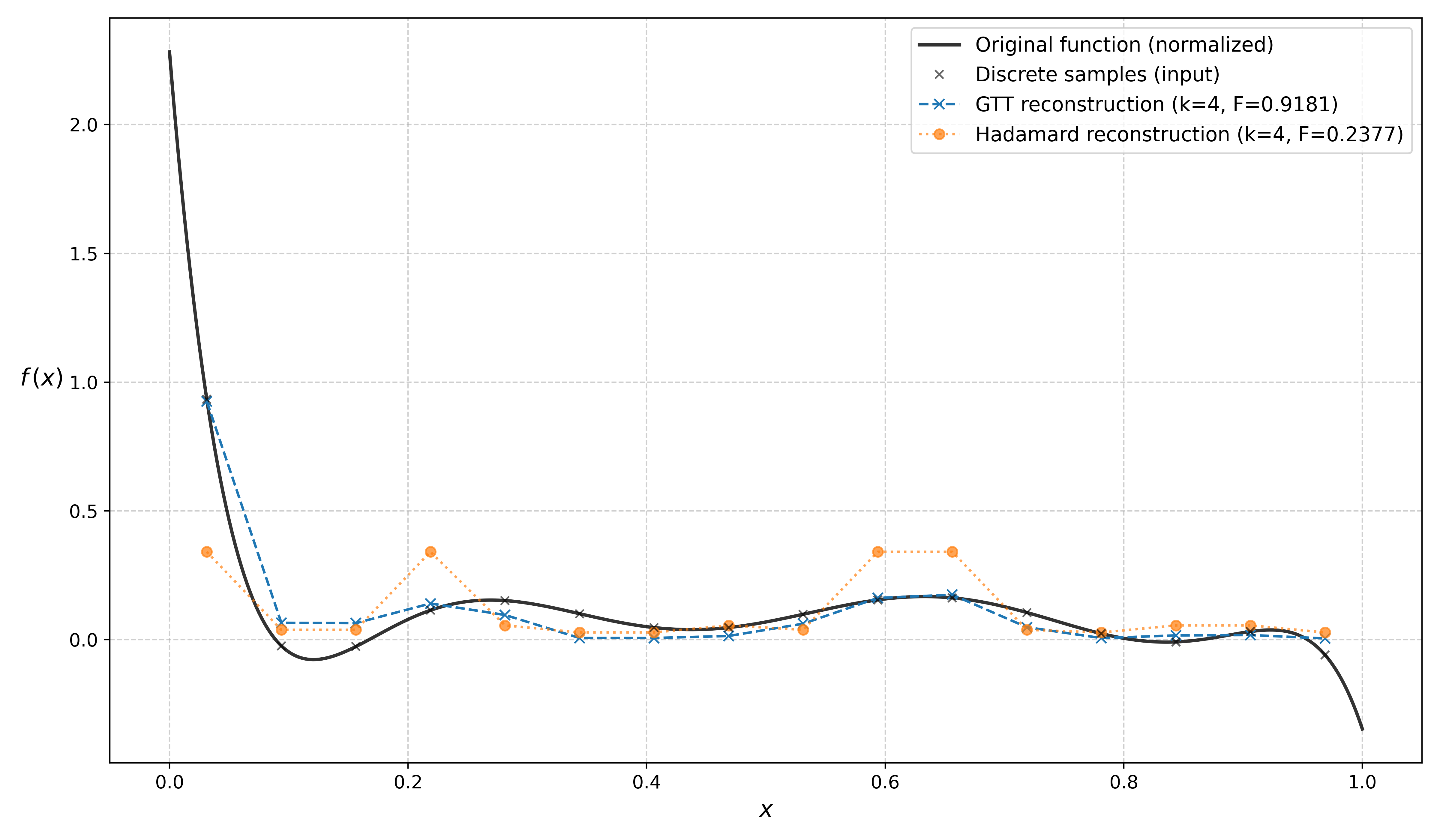}
        \caption{Reconstruction for $k=4$ components.}
        \label{fig:k4_reconstruction}
    \end{subfigure}
    \vfill %
    \begin{subfigure}[b]{0.65\textwidth}
        \centering
        \includegraphics[width=\textwidth]{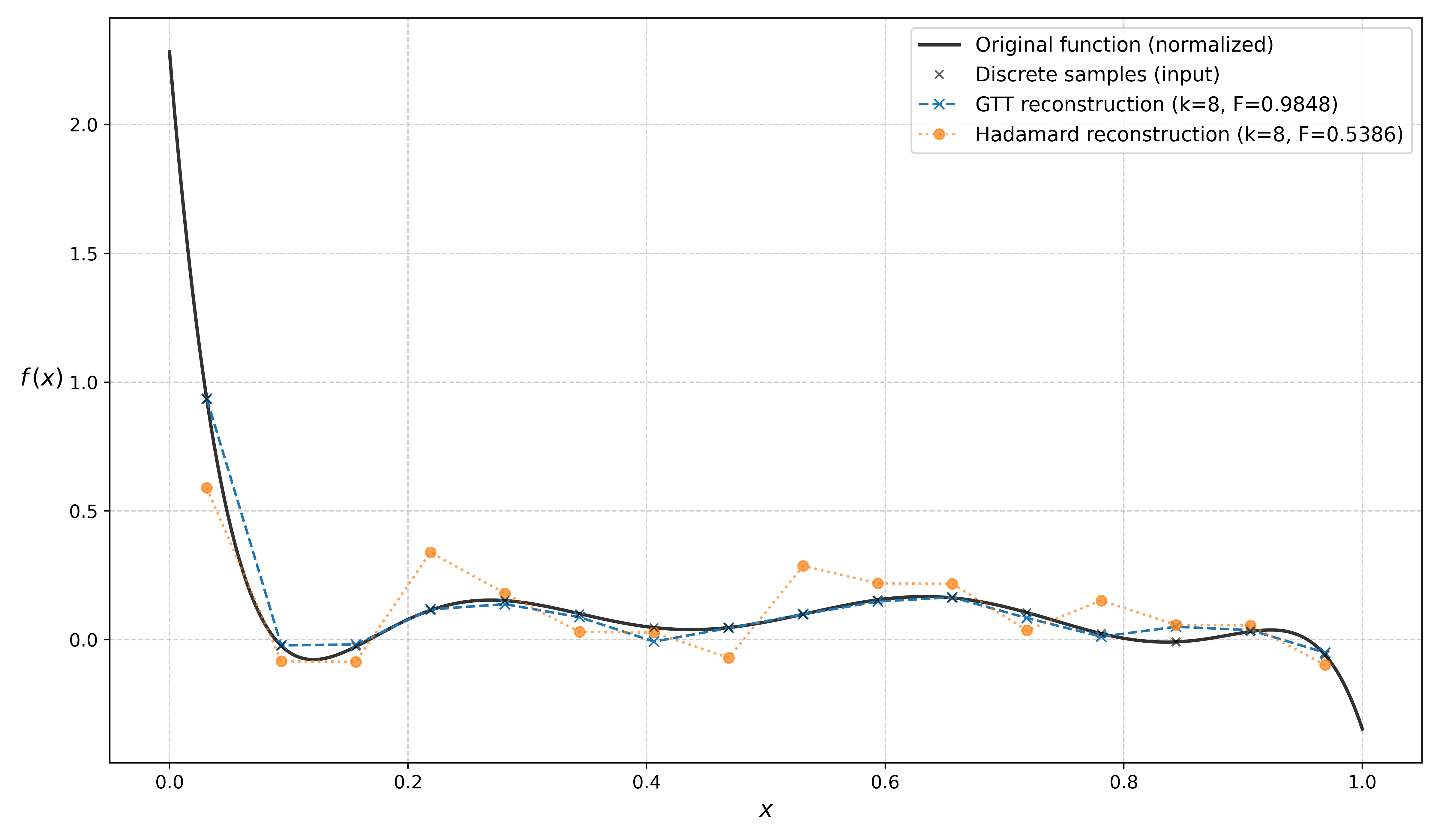}
        \caption{Reconstruction for $k=8$ components.}
        \label{fig:k8_reconstruction}
    \end{subfigure}
    \vfill
    \begin{subfigure}[b]{0.65\textwidth}
        \centering
        \includegraphics[width=\textwidth]{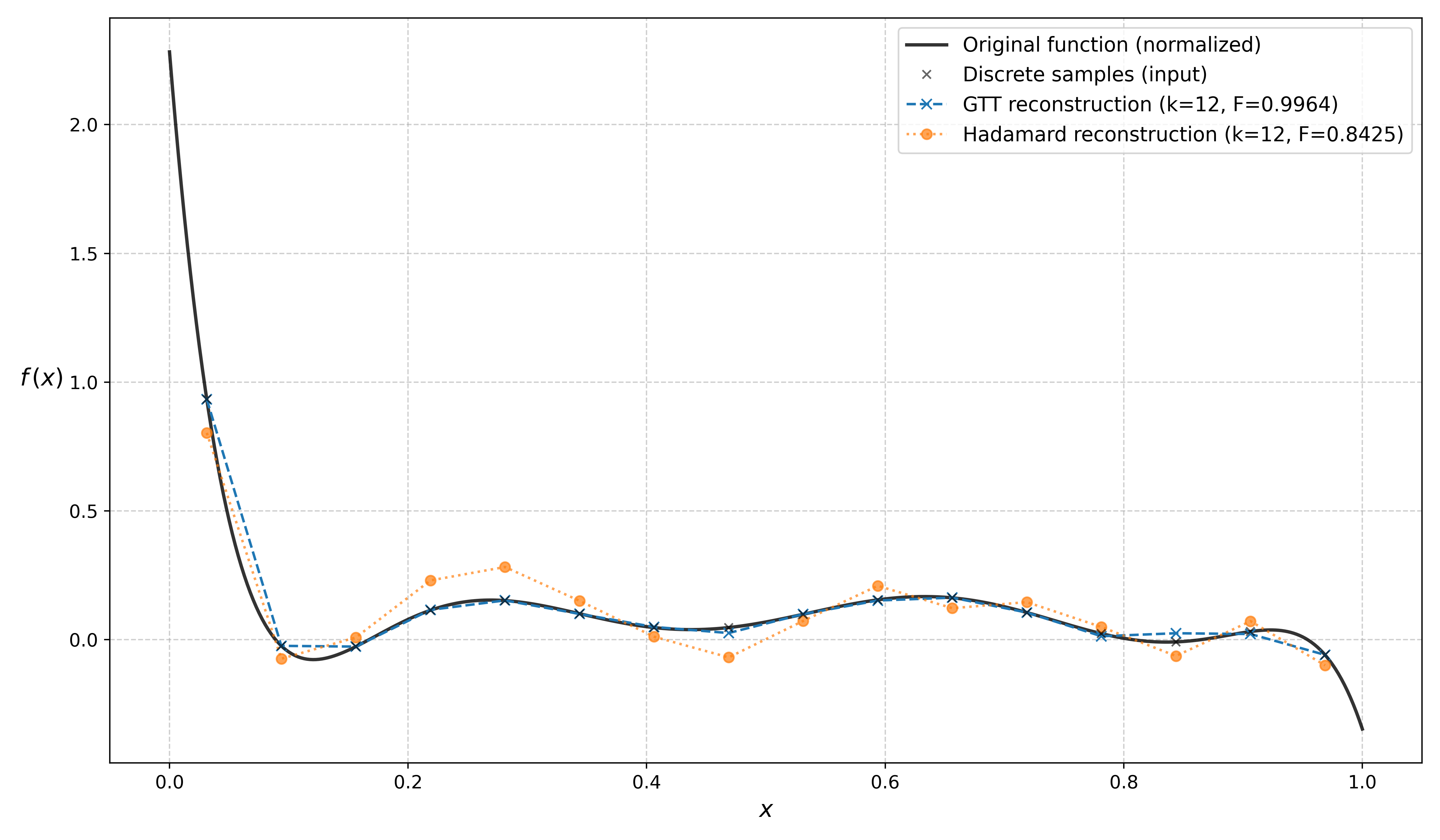}
        \caption{Reconstruction for $k=12$ components.}
        \label{fig:k12_reconstruction}
    \end{subfigure}
    \caption{Reconstruction of the function $f(x)$ for $N=16$ discrete points, comparing GTT and Hadamard transforms for varying numbers of retained components ($k$). The black `x' markers represent the normalized discretized input points, which lie on the normalized plot of the function $f(x)$. The optimal values of the tunable parameter $\theta $ used for GTT based reconstructions can be found in Table~\ref{tab:gtt_hadamard_results}.}
    \label{fig:overall_reconstruction}
\end{figure}

\subsection{Quantum digital signal filtering using GTT}
\label{ssec:filtering}

A fundamental task in signal processing is filtering, which involves isolating specific frequency components or features from a given signal. While traditional filtering often operates in a transformed domain (e.g., Fourier domain) where frequency components are naturally ordered, the Generalized Tensor Transform (GTT) offers a versatile alternative.  We introduce a quantum algorithm for signal filtering using the GTT, with a particular focus on filtering directly in the \textit{natural ordering} of the GTT domain. This approach is 
particually suited 
 for classes of signals whose energy is predominantly concentrated within the initial indices when represented in the natural GTT basis. For such signals, low-pass filtering in the natural order allows for the direct retention of the most significant components, effectively suppressing noise or unwanted features that manifest at higher natural indices. 
We present a quantum approach to low-pass and high-pass filtering using the GTT in natural ordering, as outlined in Algorithm~\ref{alg_gtt_natural_filtering}.

\subsubsection{Quantum digital signal filtering using GTT in natural order}

\begin{algorithm}[H]
    \DontPrintSemicolon
    \KwInput{The normalized input signal $\ket{\Psi}$. \\ The Generalized Tensor Transform (GTT) operator $G_N (\theta) = W^{\otimes n}$, where $W$ is a unitary gate with real entries (restricted to $\phi=0, \lambda=\pi$, and $\theta \in [0, \pi]$). \\ A cutoff value $c$ for filtering in the natural ordering.}
    \KwOutput{The state $\ket{0} \otimes \ket{\widetilde{\Psi}_{l}} + \ket{1} \otimes \ket{\widetilde{\Psi}_{h}} $, where
    $\ket{\widetilde{\Psi}_{l}} $ and $\ket{\widetilde{\Psi}_{h}} $ are low-pass and high-pass filtered signals in the time domain, respectively.}
    \Fn{Filter $\ket{\Psi}$, $G_N (\theta)$}{
        Prepare the state $\ket{0} \otimes \ket{\Psi}$ using $n +1$ qubits, where the ancilla qubit (the leftmost or the most significant qubit) is initialized to $\ket{0}$. \\
        Apply $ X \otimes G_N$ to the state $\ket{0} \otimes \ket{\Psi}$ to obtain the state $\ket{1} \otimes \ket{\widehat{\Psi}}$, where $\ket{\widehat{\Psi}} = G_N \ket{\Psi}$ is the Generalized Tensor Transform (GTT) of the input signal $\ket{\Psi}$ in natural ordering. \\
        Apply appropriate multi-controlled $X$ gates on the ancilla qubit, to split the state $\ket{1} \otimes \ket{\widehat{\Psi}}$ into low and high components based on their natural ordering indices:
        \begin{align*}
            \left(X \otimes \sum_{k < c } \ket{k}\bra{k} + I \otimes \sum_{k \geq c } \ket{k}\bra{k}\right) \left[\ket{1} \otimes \ket{\widehat{\Psi}}\right]
            = \ket{0} \otimes \ket{\widehat{\Psi}_{l}} + \ket{1} \otimes \ket{\widehat{\Psi}_{h}},
        \end{align*}
        where $\ket{\widehat{\Psi}_{l}} = \sum_{k < c } \braket{k \, | \,\widehat{\Psi}} \ket{k}$ is the low-index component and $ \ket{\widehat{\Psi}_{h}} = \sum_{k \geq c } \braket{k \, | \,\widehat{\Psi}} \ket{k}$ is the high-index component in the natural-ordered GTT domain. \\
        This is an optional step. Although not part of the low-pass or high-pass filtering, the states $\ket{\widehat{\Psi}_{l}}$ and $\ket{\widehat{\Psi}_{h}}$ can be further processed by applying appropriate quantum gates as needed. \\
        Apply inverse GTT operators to convert the filtered signals back to the time domain:
        \begin{align*}
            \left( I \otimes G_N^\dagger(\theta) \right) \left[ \ket{0} \otimes \ket{\widehat{\Psi}_{l}} + \ket{1} \otimes \ket{\widehat{\Psi}_{h}} \right] = \ket{0} \otimes \ket{\widetilde{\Psi}_{l}} + \ket{1} \otimes \ket{\widetilde{\Psi}_{h}},
        \end{align*}
        where $\ket{\widetilde{\Psi}_{l}} = G_N^\dagger(\theta) \ket{\widehat{\Psi}_{l}} $ and $\ket{\widetilde{\Psi}_{h}} = G_N^\dagger(\theta) \ket{\widehat{\Psi}_{h}} $ are low-pass and high-pass filtered signals in the time domain, respectively. \\
        \Return{$\ket{0} \otimes \ket{\widetilde{\Psi}_{l}} + \ket{1} \otimes \ket{\widetilde{\Psi}_{h}} $.}
    }
    \caption{A quantum algorithm for low-pass and high-pass filtering using the Generalized Tensor Transform (GTT) in natural ordering.}
    \label{alg_gtt_natural_filtering}
\end{algorithm}

\begin{remark}
\begin{enumerate}\leavevmode
    \item  We note that, if in Step 2 of Algorithm \ref{alg_gtt_natural_filtering}, instead of $\left( X \otimes G_N(\theta) \right)$, the unitary $\left( I \otimes G_N(\theta) \right)$ is applied to the state $\ket{0} \otimes \ket{\Psi}$, keeping all the remaining steps the same, then Algorithm \ref{alg_gtt_natural_filtering} returns the state $\ket{0} \otimes \ket{\widetilde{\Psi}_{h}} + \ket{1} \otimes \ket{\widetilde{\Psi}_{l}} $. In other words, in this case, the high-pass filtered component of the signal $\ket{\widetilde{\Psi}_{h}}$ is associated with the state of the ancilla qubit being $\ket{0} $ and the low-pass filtered component of the signal $\ket{\widetilde{\Psi}_{l}}$ is associated with the state of the ancilla qubit being $\ket{1} $.
    \item We note that, in an earlier work (ref.~~\cite{shukla2023quantumdsp}) we provided a quantum algorithm for filtering based on spectral analysis using Walsh-Hadamard basis functions.
\end{enumerate}
\end{remark}

\begin{figure}[H]
    \centering
    \begin{subfigure}[b]{0.95\textwidth}
        \centering
        \includegraphics[width=0.45\textwidth]{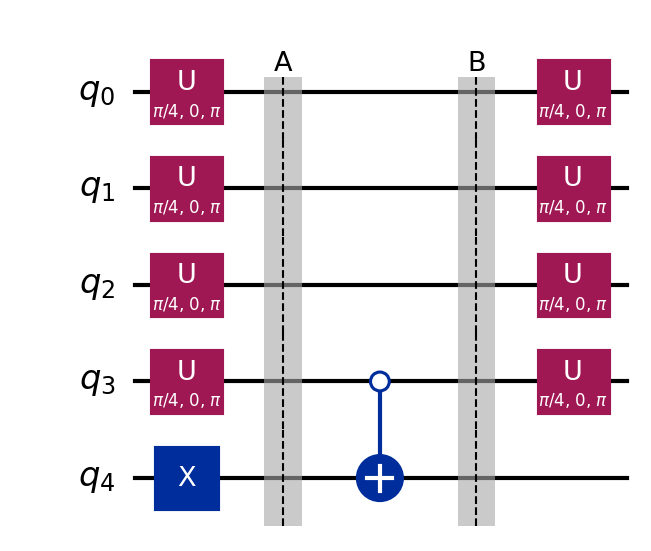}
        \caption{Quantum circuit for GTT-based low-pass filtering with a cutoff index of $N/2=8$. 
        }
        \label{fig:gtt_filter_N2}
    \end{subfigure}

    \begin{subfigure}[b]{0.95\textwidth}
        \centering
        \includegraphics[width=0.45\textwidth]{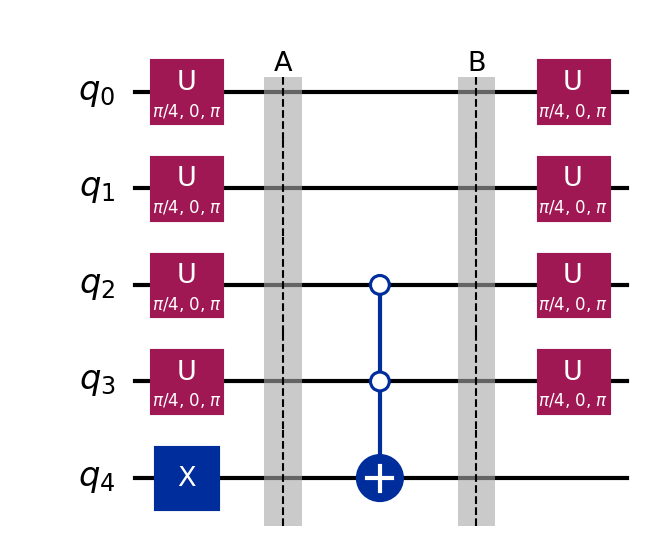}
        \caption{Quantum circuit for GTT-based low-pass filtering with a cutoff index of $N/4=4$. 
        }
        \label{fig:gtt_filter_N4}
    \end{subfigure}

    \begin{subfigure}[b]{0.95\textwidth}
        \centering
        \includegraphics[width=0.45\textwidth]{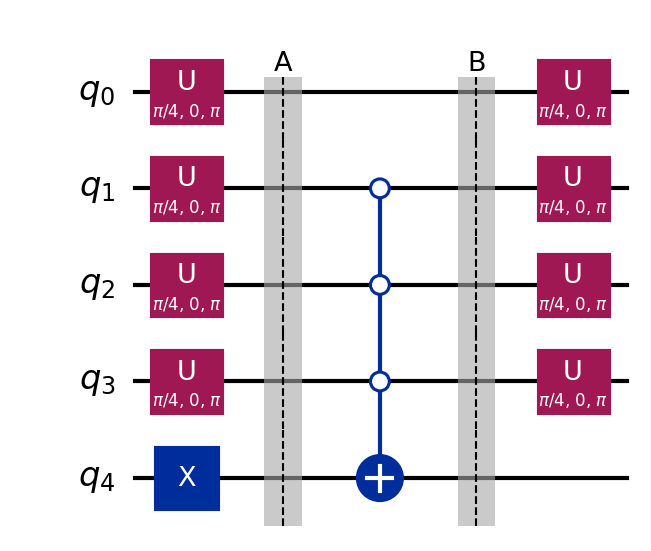}
        \caption{Quantum circuit for GTT-based low-pass filtering with a cutoff index of $N/8=2$. 
        }
        \label{fig:gtt_filter_N8}
    \end{subfigure}

    \caption{Quantum filter circuits illustrating GTT-based low-pass filtering (with $\theta = \frac{\pi}{4}$) for an $N=16$ dimensional signal in natural order, as described in Algorithm~\ref{alg_gtt_natural_filtering}. Each circuit demonstrates how different cutoff indices ($N/2$, $N/4$, $N/8$) are implemented by varying the control conditions on the signal qubits ($q_0-q_3$) that determine the auxiliary qubit's state ($q_4$), effectively separating low- and high-index components.
    }
    \label{fig:gtt_all_filters} %
\end{figure}

The quantum algorithm for GTT-based low-pass and high-pass filtering in natural ordering is implemented on a quantum circuit with $n$ signal qubits ($q_0$ to $q_{n-1}$) encoding the input signal $\ket{\Psi}$ and one auxiliary qubit ($q_n$). The process, as outlined in Algorithm \ref{alg_gtt_natural_filtering} and illustrated in Figure \ref{fig:gtt_all_filters}, proceeds as follows:

\begin{enumerate}
    \item \textit{Initialization:} The input signal $\ket{\Psi}$ is prepared on the $n$ signal qubits. Simultaneously, the auxiliary qubit ($q_n$), initially in $\ket{0}$, is flipped to $\ket{1}$ by applying an $X$ gate. This sets the initial state of the system to $\ket{1} \otimes \ket{\Psi}$. This step is visible at the very beginning of the quantum circuits in Figure \ref{fig:gtt_all_filters}, where the ancilla qubit ($q_4$ for $N=16$) receives an $X$ gate after the initial state preparation of the signal qubits.

    \item \textit{Generalized Tensor Transform (GTT) application:} The GTT operator, $G_N(\theta) = W^{\otimes n}$, is then applied to the $n$ signal qubits. Each $W$ gate is realized as a $U3(\theta, 0, \pi)$ gate applied individually to each signal qubit. This transforms the signal from the time domain to the GTT domain in natural ordering, resulting in the state $\ket{1} \otimes G_N \ket{\Psi} = \ket{1} \otimes \ket{\widehat{\Psi}}$. In Figure \ref{fig:gtt_all_filters}, this operation is represented by the series of $U3$ gates (with $\theta = \frac{\pi}{4}$) applied to $q_0$-$q_3$, following the initial state preparation and the $X$ gate on $q_4$.

    \item \textit{Filtering in the GTT domain (between barrier A and barrier B):} This is the crucial filtering stage where components are separated based on their natural ordering indices. Between the barriers A and B, a series of multi-controlled $X$ gates are applied. Each multi-controlled $X$ gate targets the auxiliary qubit ($q_n$), with the $n$ signal qubits ($q_0$ to $q_{n-1}$) serving as controls. The control state for each multi-controlled $X$ gate is determined by the binary representation of the natural ordering index $k$ that corresponds to a low-pass component ($k < c$). Since the auxiliary qubit began in $\ket{1}$, applying an $X$ gate when the signal qubits are in a low-pass state $|k\rangle$ (where $k < c$) flips the auxiliary qubit to $\ket{0}$. Conversely, for high-pass components ($k \geq c$), no $X$ gate is applied to the auxiliary qubit, ensuring it remains in $\ket{1}$. This process effectively transforms the state to $\ket{0} \otimes \ket{\widehat{\Psi}_{l}} + \ket{1} \otimes \ket{\widehat{\Psi}_{h}}$, where $\ket{\widehat{\Psi}_{l}}$ and $\ket{\widehat{\Psi}_{h}}$ are the low-pass and high-pass components in the GTT domain, respectively. Figures \ref{fig:gtt_filter_N2}, \ref{fig:gtt_filter_N4}, and \ref{fig:gtt_filter_N8} clearly show these controlled-$X$ gates between the barriers, with the number of such gates varying according to the cutoff value $c$. For instance, in Figure \ref{fig:gtt_filter_N4} (cutoff $N/4=4$), controlled-$X$ gates are applied for states $|0000\rangle$ through $|0011\rangle$ (natural indices 0 to 3).

    \item \textit{Inverse GTT application:} The final step involves applying the inverse GTT operator, $G_N^\dagger(\theta)$, to the signal qubits ($q_0$ to $q_{n-1}$) in both the $\ket{0}$ and $\ket{1}$ auxiliary subspaces. As the $W$ gate used here is its own inverse, the inverse GTT is realized by applying the same GTT operator again. This converts the filtered components back from the GTT domain to the time domain. The final state is $\ket{0} \otimes \ket{\widetilde{\Psi}_{l}} + \ket{1} \otimes \ket{\widetilde{\Psi}_{h}}$, where $\ket{\widetilde{\Psi}_{l}}$ is the low-pass filtered signal and $\ket{\widetilde{\Psi}_{h}}$ is the high-pass filtered signal. This stage is represented by another set of $U3$ gates on the signal qubits, after Barrier B, in Figure \ref{fig:gtt_all_filters}.

\end{enumerate}

\subsubsection*{Role of the ancilla qubit:}
The auxiliary qubit ($q_n$) is central to this filtering algorithm. By initially setting it to $\ket{1}$ and then selectively flipping it to $\ket{0}$ for low-pass components ($k < c$) during the filtering stage (between Barrier A and Barrier B), it acts as a quantum register that tags the frequency components. This channels the low-pass components into the auxiliary qubit's $\ket{0}$ subspace, while the high-pass components reside in its $\ket{1}$ subspace. Consequently, a measurement of the auxiliary qubit at the end of the circuit reveals which filtered signal is encoded on the main signal qubits.
is encoded on the main signal qubits.

\subsubsection{Computational example: Filtering based on natural spectral ordering}
We now demonstrate the filtering algorithm in the natural ordering of spectral components (Algorithm~\ref{alg_gtt_natural_filtering}). Using the same input signal of length $N=16$ ($n=4$ signal qubits) as before, $\ket{\Psi}$:
\[
 \ket{\Psi} = \scalebox{0.75}{$\left[
\begin{array}{rrrrrrrrrrrrrrrr}
0.9000 & 0.7000 & 0.5000 & 0.3000 & 0.1000 & -0.1000 & -0.3000 & -0.5000 & -0.4000 & -0.2000 & 0 & 0.2000 & 0.3000 & 0.1000 & -0.1000 & 0
\end{array}
\right]^T$.}
\]
We set the cutoff for filtering in the natural ordering of spectral components to $c = 4$. This means natural spectral components  $0, 1, 2, 3$ are retained in the low-pass channel, while components $4, 5, \dots, 15$ are diverted to the high-pass channel.
First, the normalized input signal $\ket{\Psi}$ is transformed into the GTT-based natural spectral domain. The entries in the following array correspond to the transformed vector in the natural spectral domain:
\[
 \scalebox{0.7}{$\left[
\begin{array}{rrrrrrrrrrrrrrrr}
0.5948 & -0.1243 & 0.0062 & -0.0363 & 0.2615 & -0.2185 & -0.3490 & 0.1497 & 0.4269 & -0.0261 & 0.1044 & -0.1462 & 0.3788 & -0.0754 & 0.0551 & 0.0413
\end{array}
\right]^T$.}
\]
The quantum circuit simulation yields two output signals in the time domain: $\widetilde{\Psi}_l$ (low-pass filtered, associated with the auxiliary qubit in state $\ket{0}$) and $\widetilde{\Psi}_h$ (high-pass filtered, associated with the auxiliary qubit in state $\ket{1}$). Their respective amplitudes are:
 	\[
 \widetilde{\Psi}_l = \scalebox{0.7}{$\left[
\begin{array}{rrrrrrrrrrrrrrrr}
0.3931 & 0.2818 & 0.1704 & 0.0835 & 0.1628 & 0.1167 & 0.0706 & 0.0346 & 0.1628 & 0.1167 & 0.0706 & 0.0346 & 0.0675 & 0.0483 & 0.0292 & 0.0143
\end{array}
\right]^T $} \text{ and}
\] 
\[
 \widetilde{\Psi}_h = \scalebox{0.67}{$\left[
\begin{array}{rrrrrrrrrrrrrrrr}
0.1940 & 0.1749 & 0.1557 & 0.1122 & -0.0976 & -0.1819 & -0.2663 & -0.3608 & -0.4238 & -0.2472 & -0.0706 & 0.0959 & 0.1282 & 0.0169 & -0.0945 & -0.0143
\end{array}
\right]^T$.}
\]
To verify the filtering accuracy, we transform $\widetilde{\Psi}_l$ and $\widetilde{\Psi}_h$ back into the GTT-based natural spectral domain. The components are shown below, preserving sign for the passed components and using magnitude for the suppressed components. Values close to zero indicate successful filtering. 

\noindent \textit{Low-pass filtered signal ($\widetilde{\Psi}_l$) (natural spectral components):}
\[
 \scalebox{0.75}{$\left[
\begin{array}{rrrrrrrrrrrrrrrr}
0.5948 & -0.1243 & 0.0062 & -0.0363 & 0 & 0 & 0 & 0 & 0 & 0 & 0 & 0 & 0 & 0 & 0 & 0
\end{array}
\right]^T$.}
\]

\noindent \textit{High-pass filtered signal ($\widetilde{\Psi}_h$) (natural spectral components):}
\[
 \scalebox{0.75}{$\left[
\begin{array}{rrrrrrrrrrrrrrrr}
0 & 0 & 0 & 0 & 0.2615 & -0.2185 & -0.3490 & 0.1497 & 0.4269 & -0.0261 & 0.1044 & -0.1462 & 0.3788 & -0.0754 & 0.0551 & 0.0413
\end{array}
\right]^T$.}
\]

As observed from the components in the GTT-based natural spectral domain, the low-pass filtered signal $\widetilde{\Psi}_l$ successfully retains components with natural indices less than the cutoff $c=4$, while components with natural index $4$ and above are suppressed (magnitudes close to zero). Conversely, the high-pass filtered signal $\widetilde{\Psi}_h$ primarily contains components with natural index $4$ and above, with lower natural components effectively filtered out. These results are in strong agreement with classical filtering predictions, confirming the correctness and effectiveness of the quantum algorithm for filtering based on natural ordering.

The effect of the filter is further illustrated by the stem plots of the GTT-based natural spectral components of the original, low-pass filtered, and high-pass filtered signals, as shown in 
Figure \ref{fig:gtt_natural_components}. 
As can be observed in Figures \ref{fig:gtt_natural_original}, \ref{fig:gtt_natural_low_pass}, and \ref{fig:gtt_natural_high_pass}, the filtering process effectively separates the GTT-based natural spectral components.

\begin{figure}[H]
\centering
\begin{subfigure}[b]{0.9\textwidth}
\centering
\includegraphics[scale=0.45]{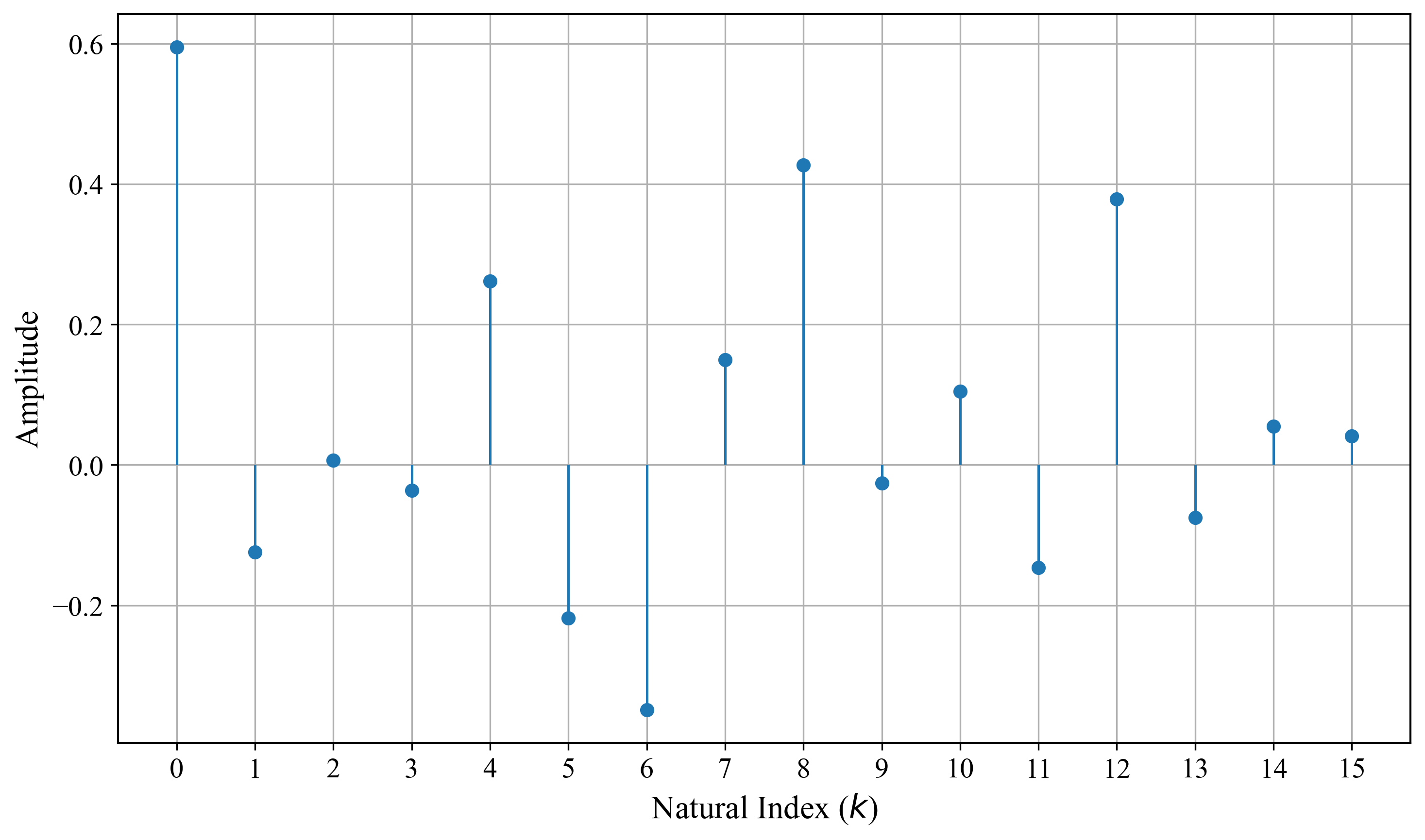}
\caption{GTT-based natural spectral components of the original input signal $\ket{\Psi}$.}
\label{fig:gtt_natural_original}
\end{subfigure}
\hfill
\begin{subfigure}[b]{0.9\textwidth}
\centering
\includegraphics[scale=0.45]{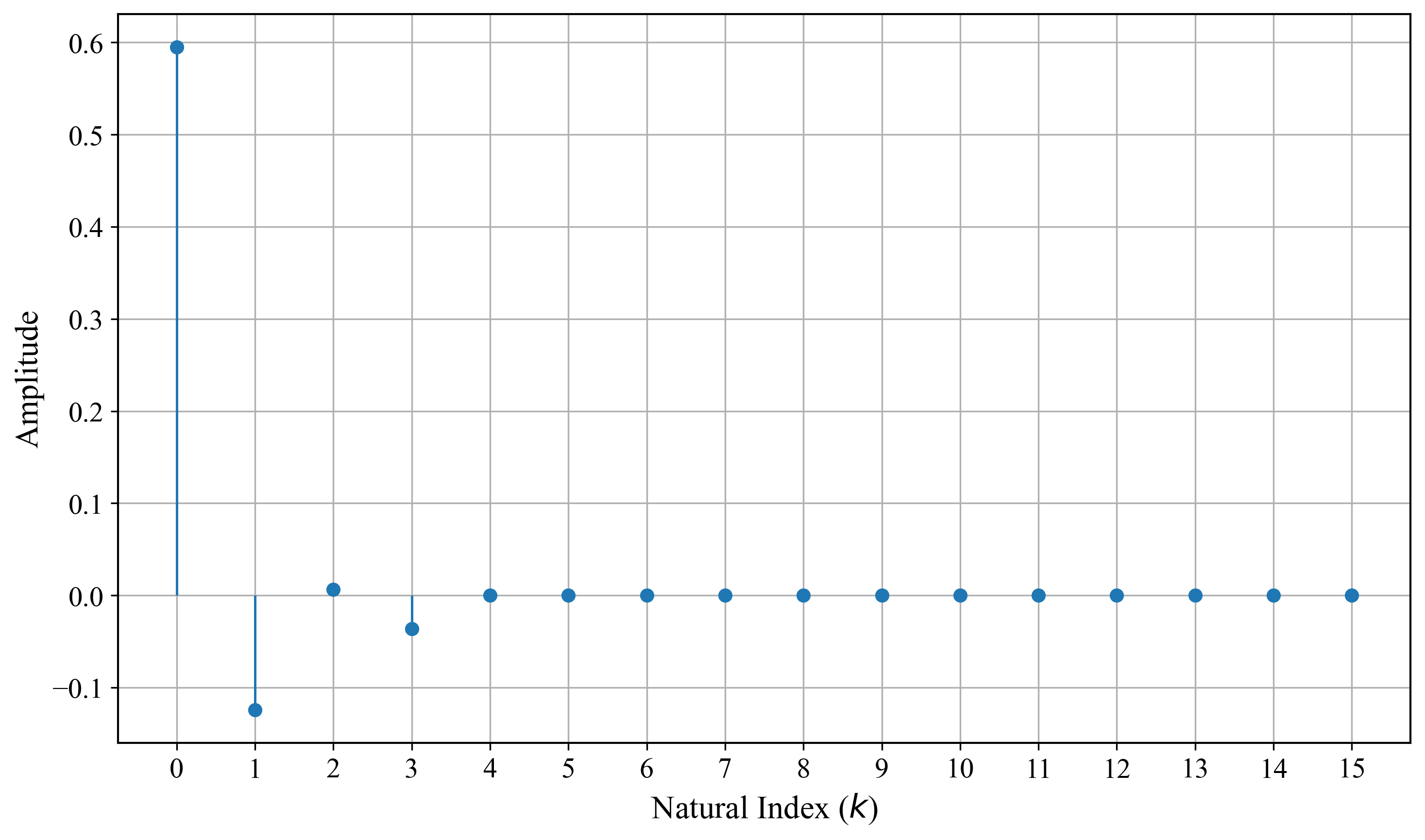}
\caption{GTT-based natural spectral components of the low-pass filtered signal.  Components with natural index 4 and above are suppressed.}
\label{fig:gtt_natural_low_pass}
\end{subfigure}
\hfill
\begin{subfigure}[b]{0.9\textwidth}
\centering
\includegraphics[scale=0.45]{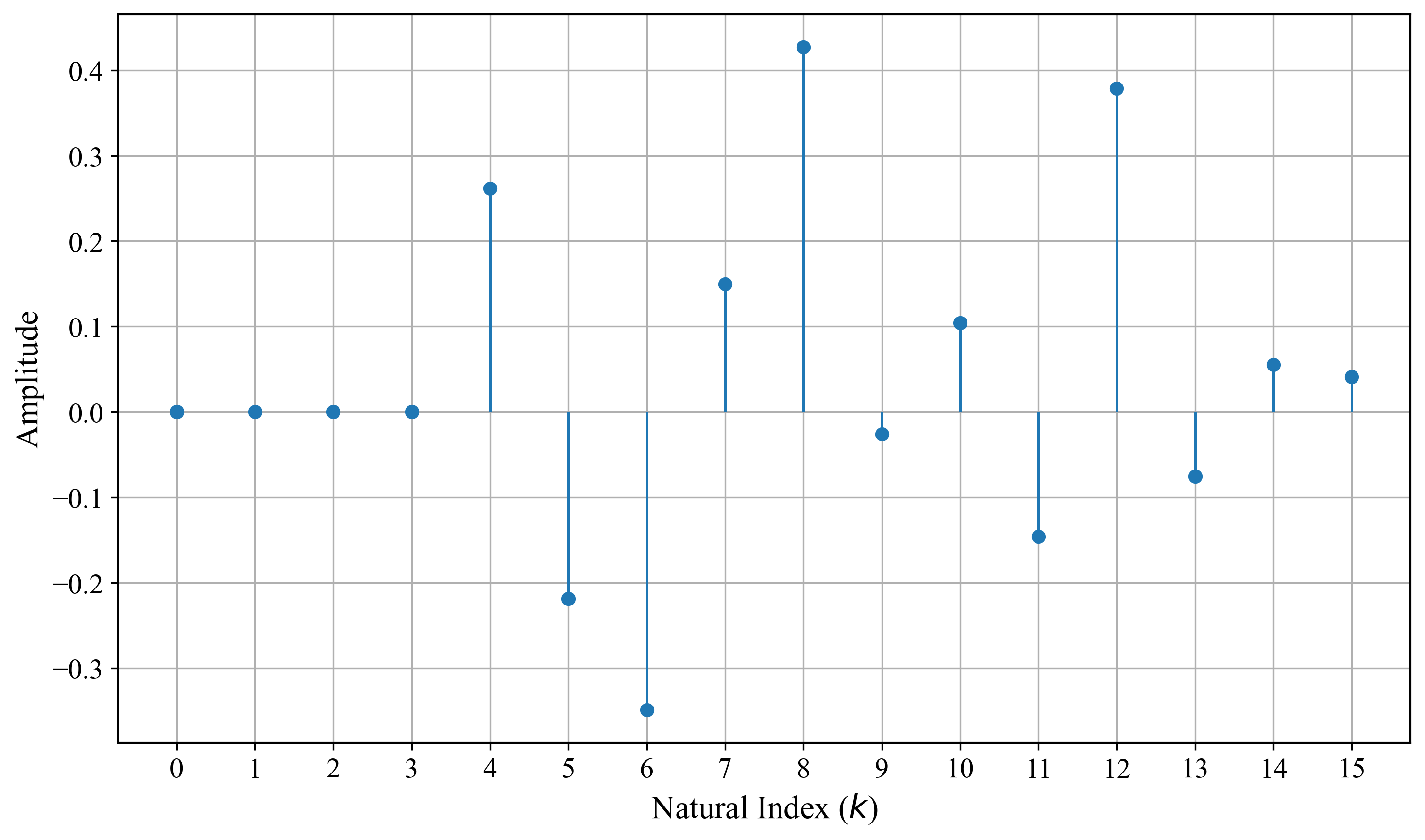}
\caption{GTT-based natural spectral components of the high-pass filtered signal. Components with natural index below 4 are suppressed.}
\label{fig:gtt_natural_high_pass}
\end{subfigure}
\caption{Stem plots illustrating the GTT-based natural spectral components of the original, low-pass filtered, and high-pass filtered signals.}
\label{fig:gtt_natural_components}
\end{figure}

\subsubsection{Computational complexity}
The proposed GTT-based filtering algorithm operating in the natural ordering, as detailed in Algorithm \ref{alg_gtt_natural_filtering}, offers notable computational advantages. If we disregard the costs associated with initial state preparation (assuming that the input signal is readily available as the output of a preceding quantum subroutine) and final measurement, the core quantum operations within our proposed GTT-based filtering approach requires a gate complexity and circuit depth of $O(\log_2 N)$, where $N=2^n$ is the dimension of the signal and $n=\log_2 N$ is the number of signal qubits. This stands in contrast to Quantum Fourier Transform (QFT)-based filtering, which needs a gate complexity and circuit depth of $O((\log_2 N)^2)$ for $n$ qubits. Consequently, our proposed GTT-based algorithm demonstrates a quadratic improvement in both gate complexity and circuit depth when compared to QFT-based filtering. Furthermore, when considering the classical Fast Fourier Transform (FFT)-based approach for filtering, which incurs a computational cost of $O(N \log_2 N)$, the QFT-based quantum filtering already provides an exponential advantage (again, abstracting away state preparation and measurement costs). Building upon this, our GTT-based approach further enhances efficiency by offering a quadratic improvement over the QFT method. It is important to acknowledge that for many quantum algorithms, including the one presented here, the overheads associated with universal state preparation and comprehensive measurement remain significant challenges in achieving a practical quantum advantage over their classical counterparts. However, this GTT-based filtering algorithm, similar to QFT-based approaches, becomes particularly beneficial in scenarios where the input signal is readily available as the output of another quantum subroutine, thereby circumventing the need for explicit state preparation. Moreover, the algorithm allows for the extraction of global features or statistics of the filtered signal without incurring the full measurement costs associated with a complete reconstruction of the output state, which can be advantageous in specific applications.

\section{Conclusion}
\label{sec:conclusion}

In this work, we proposed a novel and highly adaptive Generalized Tensor Transform (GTT) framework. This framework is constructed through the $n$-fold tensor product of an arbitrary $b \times b$ unitary matrix $W$. This construction generalizes well-known transforms such as the Walsh-Hadamard and multi-dimensional Discrete Fourier Transforms, providing a richer and more adaptable set of orthonormal basis functions. A key feature of the GTT framework is the ability to achieve fine-grained control over transformations by tuning the parameters of the base matrix $W$, allowing the adaptive shaping of each basis function while preserving the characteristic discrete, piecewise-constant, and often discontinuous structure that makes Walsh-type bases effective in both classical and quantum settings.

One of the main contributions of this work is the development and rigorous analysis of a fast classical algorithm for GTT computation. While a naive classical implementation via direct matrix multiplication incurs a prohibitive computational cost of $O(N^2)$, our proposed fast algorithm achieves an exponentially lower complexity of $O(N \log_b N)$. This efficiency is analogous to the computational advantages seen in well-established algorithms such as the Fast Fourier Transform (FFT) and Fast Walsh-Hadamard Transform (FWHT), making GTTs practically applicable for large-scale classical data processing.

For quantum applications, our GTT-based algorithm, implemented in the natural spectral ordering, achieves both a gate complexity and a circuit depth of $O(\log_b N)$, where $N = b^n$ denotes the length of the input vector. This represents a quadratic improvement over Quantum Fourier Transform (QFT), which requires $O((\log_b N)^2)$ gates and depth for $n$ qudits. Furthermore, this quantum GTT algorithm offers an exponential advantage over classical FFT-based and FWHT-based methods, which incur a computational cost of $O(N \log_b N)$. This demonstrates a clear path toward enhanced efficiency in quantum algorithms, especially in cases where the input signal is readily available as the output of a preceding quantum subroutine, circumventing the challenges of universal state preparation. 

We have further demonstrated the diverse applications of GTTs in quantum computing. The framework's inherent flexibility allows basis representations to be tailored to the specific structure of quantum data and the requirements of distinct computational tasks. Our numerical results confirm GTT's superior performance in applications such as quantum state compression, function encoding and digital signal filtering. In quantum state compression and function encoding applications, GTT consistently achieves significantly higher fidelities with fewer retained components compared to fixed transforms (such as FWHT or FFT). This adaptive sparsity is crucial for enabling the efficient representation and processing of complex quantum information. We also provided new classical and quantum algorithms for digital signal filtering based on the GTT framework. 

In conclusion, the Generalized Tensor Transform framework stands as a powerful and adaptable mathematical toolset that potentially offers compelling computational advantages across both classical and, most notably, quantum domains. Its ability to offer fine-grained control over basis function generation has potential to provide improved performance in key quantum computing applications. Future work will explore  its integration into advanced quantum machine learning models, and develop robust strategies for real-time parameter tuning on noisy intermediate-scale quantum (NISQ) devices.

\bibliographystyle{unsrt} 

\begin{thebibliography}{10}

\bibitem{grozdanov_inequality_nodate}
Vassil~S Grozdanov and Stanislava~S Stoilova.
\newblock The inequality of {E}rd{\H{o}}s-{T}uran-{K}oksma: {W}alsh and {H}aar functions over finite groups.
\newblock {\em Math. Balkanica (NS)}, 19:349--366, 2005.

\bibitem{horadam_generalised_2005}
Kathy~J Horadam.
\newblock A generalised hadamard transform.
\newblock In {\em Proceedings. International Symposium on Information Theory, 2005. ISIT 2005.}, pages 1006--1008. IEEE, 2005.

\bibitem{wang_improved_2010}
Ren-hong Wang and Wei Dan.
\newblock Improved {H}aar and {W}alsh functions over triangular domains.
\newblock {\em Journal of the Franklin Institute}, 347(9):1782--1794, 2010.

\bibitem{watari1958generalized}
Chinami Watari.
\newblock On generalized {W}alsh fourier series.
\newblock {\em Tohoku Mathematical Journal, Second Series}, 10(3):211--241, 1958.

\bibitem{episkoposian_greedy_2011}
Sergo~A. Episkoposian.
\newblock On greedy algorithms with respect to generalized {W}alsh system.

\bibitem{yuan2021generalized}
Xixi Yuan and Zhanchuan Cai.
\newblock A generalized {W}alsh system and its fast algorithm.
\newblock {\em IEEE Transactions on Signal Processing}, 69:5222--5233, 2021.

\bibitem{beauchamp1975walsh}
Kenneth~George Beauchamp.
\newblock {\em {W}alsh functions and their applications}.
\newblock Academic Press, 1975.

\bibitem{shapiro_walsh_1974}
JB~Shapiro and T~Reich.
\newblock {W}alsh functions for simplifying computation of blood pressure wave power spectra.
\newblock {\em Annals of Biomedical Engineering}, 2(3):265--273, 1974.

\bibitem{shukla2023quantumdsp}
Alok Shukla and Prakash Vedula.
\newblock A quantum approach for digital signal processing.
\newblock {\em The European Physical Journal Plus}, 138(12):1--24, 2023.

\bibitem{broadbent_analysis_1992}
Hilary~A Broadbent and York~A Maksik.
\newblock Analysis of periodic data using {W}alsh functions.
\newblock {\em Behavior Research Methods, Instruments, \& Computers}, 24(2):238--247, 1992.

\bibitem{smith_walsh_1979}
EH~Smith and WM~Walmsley.
\newblock {W}alsh functions and their use in the assessment of surface texture.
\newblock {\em Wear}, 57(1):157--166, 1979.

\bibitem{shukla2022hybridimage}
Alok Shukla and Prakash Vedula.
\newblock A hybrid classical-quantum algorithm for digital image processing.
\newblock {\em Quantum Information Processing}, 22(1):3, 2022.

\bibitem{rohida2024hybrid}
Mohit Rohida, Alok Shukla, and Prakash Vedula.
\newblock Hybrid classical-quantum image processing via polar {W}alsh basis functions.
\newblock {\em Quantum Machine Intelligence}, 6(2):72, 2024.

\bibitem{o1978edge}
Frank O'Gorman.
\newblock Edge detection using {W}alsh functions.
\newblock {\em Artificial Intelligence}, 10(2):215--223, 1978.

\bibitem{shukla2023hybridode}
Alok Shukla and Prakash Vedula.
\newblock A hybrid classical-quantum algorithm for solution of nonlinear ordinary differential equations.
\newblock {\em Applied Mathematics and Computation}, 442:127708, 2023.

\bibitem{palanisamy_minimum_1983}
KR~Palanisamy and Ganti~Prasada Rao.
\newblock Minimum energy control of time-delay systems via walsh functions.
\newblock {\em Optimal Control Applications and Methods}, 4(3):213--226, 1983.

\bibitem{paraskevopoulos_transfer_1980}
PN~Paraskevopoulos and SJ~Varoufakis.
\newblock Transfer function determination from impulse response via walsh functions.
\newblock {\em International Journal of Circuit Theory and Applications}, 8(1):85--89, 1980.

\bibitem{goldberg_genetic_nodate}
David~E Goldberg.
\newblock Genetic algorithms and {W}alsh functions: {P}art {I}, {A} gentle introduction.
\newblock {\em Complex systems}, 3:129--152, 1989.

\bibitem{grover1996fast}
Lov~K Grover.
\newblock A fast quantum mechanical algorithm for database search.
\newblock In {\em Proceedings of the Twenty-eighth Annual ACM Symposium on Theory of Computing}, pages 212--219. ACM, 1996.

\bibitem{shukla2025efficientsearch}
Alok Shukla and Prakash Vedula.
\newblock An efficient implementation of a quantum search algorithm for arbitrary {N}.
\newblock {\em The European Physical Journal Plus}, 140(6):1--10, 2025.

\bibitem{shukla2024efficient}
Alok Shukla and Prakash Vedula.
\newblock An efficient quantum algorithm for preparation of uniform quantum superposition states.
\newblock {\em Quantum Information Processing}, 23(2):38, 2024.

\bibitem{bernstein1993quantum}
Ethan Bernstein and Umesh Vazirani.
\newblock Quantum complexity theory.
\newblock In {\em Proceedings of the twenty-fifth annual ACM symposium on Theory of computing}, pages 11--20, 1993.

\bibitem{shukla2023generalizationpbv}
Alok Shukla and Prakash Vedula.
\newblock A generalization of {B}ernstein--{V}azirani algorithm with multiple secret keys and a probabilistic oracle.
\newblock {\em Quantum Information Processing}, 22(6):244, 2023.

\bibitem{shor1994algorithms}
Peter~W Shor.
\newblock Algorithms for quantum computation: discrete logarithms and factoring.
\newblock In {\em Proceedings 35th annual symposium on foundations of computer science}, pages 124--134. Ieee, 1994.

\bibitem{deutsch1992rapid}
David Deutsch and Richard Jozsa.
\newblock Rapid solution of problems by quantum computation.
\newblock {\em Proceedings of the Royal Society of London. Series A: Mathematical and Physical Sciences}, 439(1907):553--558, 1992.

\bibitem{simon1997power}
Daniel~R Simon.
\newblock On the power of quantum computation.
\newblock {\em SIAM journal on computing}, 26(5):1474--1483, 1997.

\bibitem{Brassard_2002}
Gilles Brassard, Peter H{\o}yer, Michele Mosca, and Alain Tapp.
\newblock Quantum amplitude amplification and estimation, 2002.

\bibitem{suzuki2020amplitude}
Yohichi Suzuki, Shumpei Uno, Rudy Raymond, Tomoki Tanaka, Tamiya Onodera, and Naoki Yamamoto.
\newblock Amplitude estimation without phase estimation.
\newblock {\em Quantum Information Processing}, 19(2):1--17, 2020.

\bibitem{shukla2024partialsum}
Alok Shukla and Prakash Vedula.
\newblock Efficient quantum algorithm for weighted partial sums and numerical integration.
\newblock {\em Advanced Quantum Technologies}, 2025.

\bibitem{harrow2009quantum}
Aram~W Harrow, Avinatan Hassidim, and Seth Lloyd.
\newblock Quantum algorithm for linear systems of equations.
\newblock {\em Physical Review Letters}, 103(15):150502, 2009.

\bibitem{walsh1923closed}
Joseph~L Walsh.
\newblock A closed set of normal orthogonal functions.
\newblock {\em American Journal of Mathematics}, 45(1):5--24, 1923.

\bibitem{beer1981walsh}
Tom Beer.
\newblock {W}alsh transforms.
\newblock {\em American Journal of Physics}, 49(5):466--472, 1981.

\bibitem{ahner1988walsh}
Henry~F Ahner.
\newblock {W}alsh functions and the solution of nonlinear differential equations.
\newblock {\em American Journal of Physics}, 56(7):628--633, 1988.

\bibitem{geadah1977natural}
Youssef~A. Geadah and MJG Corinthios.
\newblock Natural, dyadic, and sequency order algorithms and processors for the {W}alsh--{H}adamard transform.
\newblock {\em IEEE Transactions on Computers}, 26(05):435--442, 1977.

\bibitem{shukla2024sequency}
Alok Shukla and Prakash Vedula.
\newblock On sequency-complete and sequency-ordered matrices.
\newblock {\em arXiv preprint arXiv:2402.11003}, 2024.

\bibitem{shukla2022quantumzerocrossings}
Alok Shukla.
\newblock A quantum algorithm for counting zero-crossings.
\newblock {\em arXiv preprint arXiv:2212.11814}, 2022.

\bibitem{Bai2020QuantumCompression}
Ge~Bai, Yuxiang Yang, and Giulio Chiribella.
\newblock Quantum compression of tensor network states.
\newblock {\em New Journal of Physics}, 22(4):043015, 2020.

\bibitem{Rozema2014QuantumData}
Lee~A Rozema, Dylan~H Mahler, Alex Hayat, Peter~S Turner, and Aephraim~M Steinberg.
\newblock Quantum data compression of a qubit ensemble.
\newblock {\em Physical Review Letters}, 113(16):160504, 2014.

\bibitem{huang2025quantum}
He-Liang Huang and Chen Ding.
\newblock Quantum state compression shadow.
\newblock {\em The European Physical Journal Special Topics}, pages 1--8, 2025.

\bibitem{Cruzeiro2025CompressionEntanglement}
Yu~Guo, Hao Tang, Jef Pauwels, Emmanuel~Zambrini Cruzeiro, Xiao-Min Hu, Bi-Heng Liu, Yun-Feng Huang, Chuan-Feng Li, Guang-Can Guo, and Armin Tavakoli.
\newblock Compression of entanglement improves quantum communication.
\newblock {\em Laser \& Photonics Reviews}, 19(10):2401110, 2025.

\bibitem{vanLoo2017QuantumState}
Beno{\^\i}t Vermersch, P-O Guimond, Hannes Pichler, and Peter Zoller.
\newblock Quantum state transfer via noisy photonic and phononic waveguides.
\newblock {\em Physical Review Letters}, 118(13):133601, 2017.

\bibitem{schuld2021effect}
Maria Schuld, Ryan Sweke, and Johannes~Jakob Meyer.
\newblock Effect of data encoding on the expressive power of variational quantum-machine-learning models.
\newblock {\em Physical Review A}, 103(3):032430, 2021.

\bibitem{shin2023exponential}
Seongwook Shin, Yong-Siah Teo, and Hyunseok Jeong.
\newblock Exponential data encoding for quantum supervised learning.
\newblock {\em Physical Review A}, 107(1):012422, 2023.

\bibitem{gonzalez2024efficient}
Javier Gonzalez-Conde, Thomas~W Watts, Pablo Rodriguez-Grasa, and Mikel Sanz.
\newblock Efficient quantum amplitude encoding of polynomial functions.
\newblock {\em Quantum}, 8:1297, 2024.

\bibitem{proakis2007digital}
John~G Proakis.
\newblock {\em Digital signal processing: principles, algorithms, and applications, 4/E}.
\newblock Pearson Education India, 2007.

\bibitem{zindorf2025multi}
Ben Zindorf and Sougato Bose.
\newblock Multi-controlled quantum gates in linear nearest neighbor.
\newblock {\em arXiv preprint arXiv:2506.00695}, 2025.

\end{thebibliography}

\end{document}